\documentclass[authoryear,reqno,11pt,a4paper,numafflabel]{elsarticle}   

\journal{Quantitative Finance}

\usepackage{amsmath,amssymb,graphicx,hyperref,setspace,ifthen,soul,amsthm,float} 
\usepackage{booktabs}
\usepackage{algorithm}
\usepackage{subfigure}
\usepackage{multirow}   
\usepackage{makecell}   
\usepackage{cleveref}
\usepackage{microtype}

\usepackage[authoryear]{natbib}
\robustify\bfseries
\hypersetup{
	pdftitle={},
	pdfauthor={Lilian Hu},
	pdfkeywords={Simulation of stochastic volatility models via operator splitting},
	colorlinks=true,
	linkcolor=red,
	citecolor=blue,
	urlcolor=blue,
	bookmarksnumbered=true,
	pdfstartview=
}

\date{\today}

\newtheorem{proposition}{Proposition}
\newtheorem{remark}{Remark}
\newtheorem{corollary}{Corollary}
\begin{document}
\setstretch{1.05}

\begin{frontmatter}

	\title{Simulation of stochastic volatility models via operator splitting schemes}

	\author[1]{Lilian Hu\corref{corrauthor}}
	\ead{l66hu@uwaterloo.ca}

	\author[2]{Congxin He}

	\author[2]{Yue Kuen Kwok}
	\author[3]{Gongqiu Zhang}

	\address[1]{\textit{\small{Department of Statistics and Actuarial Science, University of Waterloo, Canada}}}

	\address[2]{\textit{\small{Financial Technology Thrust,\\ Hong Kong University of Science and Technology (Guangzhou), China}}}

	\address[3]{\textit{\small{School of Science and Engineering,\\ Chinese University of Hong Kong, Shenzhen, China}}}

	\cortext[corrauthor]{Corresponding author}

	\begin{abstract}
		The standard Euler discretization schemes for numerical option pricing under stochastic volatility models are known to exhibit high biases and potential unreliability. The alternative use of the exact (unbiased) simulation approach invariably involves numerical evaluation of integrated variance (and / or volatility) conditional on terminal variance (volatility) value. To resolve the technical challenge, most simulation schemes either employ the tedious Fourier inversion of conditional characteristic function or numerical approximation by moment matched distribution. We propose a general framework of constructing efficient and reliable simulation schemes for stochastic volatility models via the Strang operator splitting approximation. The simulation procedure completely circumvents the necessity of evaluation of conditional integrated variance (and / or) volatility. Our simulation schemes compete favorably well with most existing exact simulation schemes and the biased Euler schemes in terms of accuracy, efficiency, reliability and ease of implementation. Extensive numerical tests were conducted to illustrate the versatility and success of our operator splitting approach for most common stochastic volatility models, such as the Heston-type models, lifted Heston model, Hull-White model, and Barndorff-Nielsen and Shephard model. We also establish the proof of second-order convergence of the operator splitting schemes. 
	\end{abstract}


\end{frontmatter}
    \section{Introduction}\label{s:intro}\noindent
  There has been an extensive literature on simulation schemes for option pricing under various types of stochastic volatility models. The direct Euler / Milstein type discretization schemes are known to exhibit discretization biases and potential unreliability. In the case of the Heston model, the square root function in its variance dynamics violates the Lipschitz condition. This may hinder the proper numerical convergence of the Euler discretization schemes. Also, the occurrence of negative values of simulated variance under the Cox-Ingersoll-Ross (CIR) dynamics must be fixed heuristically. Readers may refer to \citet{Lord2010a} for comprehensive discussion of these fixes in discretization schemes.
  
  \citet{Broadie} pioneer the unbiased (exact) simulation of the asset price and variance for the Heston stochastic volatility model based on exact simulation of the conditional transition distribution of variance and time integrated variance conditional on the terminal variance value. The first simulation step of the terminal variance value is rather straightforward. However, the second step of simulation of conditional integrated variance poses some technical challenge. \citet{Broadie} propose to generate simulated conditional integrated variance via the numerical Fourier inversion of the corresponding characteristic function to obtain the distribution function and in conjunction with a root finding procedure in inverting the distribution function. This second step is proven to be highly computationally demanding, especially for numerical pricing of path dependent options with discrete monitoring of asset prices at multiple time points. The exact simulation approach of taking the Fourier inversion of the characteristic function of conditional time integral of variance or volatility has also been extended to other stochastic volatility models. These include the more recent works of \citet{cai2017sabr} for the SABR model, \citet{Li} for the Ornstein-Uhlenbeck (OU) driven model, \citet{Zeng2023} for the Heston model, 4/2 model and OU driven model, \citet{brigone2024104861} for the Hull-White (HW) model. 
%
   
  Recall that the bottleneck in achieving computational efficiency in the unbiased or so-called “exact” simulation approach is the numerical evaluation of the conditional integrated variance. To mitigate the tedious Fourier inversion procedure for each simulation path, \citet{Zeng2023} employ the Hilbert transform inversion in conjunction with interpolation techniques and manage to achieve relatively significant reduction in CPU time. Their simulation schemes are effective for pricing path dependent options. Instead of using the integral transform inversion to recover the conditional integrated variance, \citet{tse2012low} approximate this quantity by the moment matched Inverse Gaussian distribution. Similar technique of approximating the conditional integrated variance in the Heston model using the Inverse Gaussian limiting distribution is also proposed by \citet{AbiJaber2024heston}. In addition, \citet{kyriakou2023unified} utilize the Pearson curve to simulate the log-asset price via fitting the first four order moments of the conditional integrated variance. The moment matching approach has also been used by \citet{andersen2008simple} in his truncated Gaussian scheme for the Heston model in the simulation of the variance process. 
  
  Besides the above time-consuming exact simulation schemes, there have been several innovative simulation schemes that achieve remarkably high computational efficiency via explicit analytic formulas for the conditional integrated variance. For the Heston model, \citet{glasserman2011gamma} manage to derive the explicit analytic representation of the conditional integrated variance via an infinite series of gamma expansion. The computational efficiency of the new exact simulation is significantly enhanced when compared with that of \citet{Broadie}. \citet{malham2021series} propose an exact simulation scheme for the Heston model by deriving an alternative series expansion for the conditional integrated variance in terms of double infinite weighted sums of independent random variables through a measure change and decomposition of squared Bessel bridges. By exploring the Poisson conditioned gamma expansion of \citet{glasserman2011gamma}, \citet{choi_heston2024} achieve the most remarkable computational efficiency and accuracy when compared to all simulation schemes for the Heston model.  For the OU driven model, \citet{Choi2025OU} derive the time integrals of volatility and variance as an infinite sum of independent normal random variables via the Karhunen-Lo\`{e}ve expansion. For the same level of numerical accuracy, his simulation scheme is faster than the simulation scheme of \citet{Li} by a factor of several hundred. So far, such remarkable successes of computational efficiency using infinite series representation of integrated variance / volatility are limited to the Heston model and OU driven model.

  To resolve the limitation of reproducing the sudden spikes in the Heston-type models, 
  \citet{barndorff-nielsen2001non} introduce the non-Gaussian Ornstein-Uhlenbeck (OU) framework. 
  This is achieved by modeling the variance dynamics via a Background Driving Lévy Process (BDLP). 
  The Barndorff-Nielsen Shephard (BN-S) model effectively decouples the jump activity in returns from that in volatility, 
  offering a flexible structure to characterize market discontinuities. Despite the 
  theoretical elegance of the BN-S model, the search for efficient simulation schemes for pricing
  path dependent derivatives remains a computational challenge. The existing numerical schemes 
  broadly fall into two categories. The first category consists of semi-analytical methods based on 
  the use of Fourier transform. 
  For instance, \citet{Baviera2024} utilize the inversion of characteristic function 
  for pricing VIX futures. Such approach is highly efficient for European style products. However, it
   often falters when dealing with discretely monitored path dependent options 
  due to the lack of tractable characteristic 
  function for the integrated variance. 
  The second category of numerical schemes uses Monte Carlo simulation, which can handle complex payoff structures effectively. Furthermore, to avoid the bias introduced by the standard Euler-Maruyama discretization, research efforts have been devoted to designing exact simulation schemes.

  For the Gamma-OU and OU-Gamma processes, \citet{Dassio2019Gamma} derive exact sampling algorithms 
  based on the factorization of the Moment Generating Function (MGF), and later extend the approach to the 
  Tempered Stable-OU (TS-OU) and OU-TS models \citep{Dassio2021TS}. 
  Using an alternative exact simulation approach, \citet{sabino2020gamma} and \citet{sabino2020tempered} 
  propose the exact simulation schemes based on the probabilistic decomposition and 
  the series representation of the so-called $a$-remainder. 
  Though these two distinct exact simulation approaches are both
  insightful and analytically rigorous, a major shortcoming of the simulation schemes proposed by Sabino and his coauthors is that their 
  methods can only be applied to price some special types of options. Specifically, 
  their derivation approaches primarily focus on the 
  theoretical exact simulation of the variance process only. They do not consider 
  the pricing of standard equity options, leaving the critical challenge of the joint simulation of 
  the asset price and integrated variance largely unaddressed. 
  In summary, their schemes lack the versatility to efficiently handle standard equity options or 
  complex path dependent derivatives.
  To overcome the limitations of these variance-centric approaches, \citet{James2013} develop 
  an exact simulation algorithm for the joint simulation of asset price and variance for pricing path dependent options under the OU-Gamma specification
  using the Beta-Gamma algebra and Dirichlet means.
   However, the embedded cumbersome simulation techniques (such as the Double Coupling From The Past) make the 
  implementation highly complex, resulting in prohibitive computational cost for practical 
  option pricing calculations.

  This paper aims to provide a general framework of constructing accurate, reliable, efficient and easily implementable exact simulation schemes for most common stochastic volatility models. We employ the second-order operator splitting method \citep{strang1968construction} to the coupled system of stochastic differential equations (SDEs) of asset price and variance dynamics in the stochastic volatility models. The operator splitting method splits the original system of SDEs into two parts over a time step. The numerical procedure separately computes the solution to each part, then combines the two separate solutions to form an approximate solution to the original system of SDEs. The integral representation of the solution of the resulting approximate SDEs involves only the terminal variance, without the necessity of conditional integrated variance. Therefore, the resulting simulation schemes circumvent the time-consuming numerical evaluation of conditional integrated variance, which is typically required in earlier exact simulation schemes. The major source of numerical errors arises from the operator splitting procedure. Our numerical tests reveal second-order convergence to the option prices. Since the asset prices at successive time points are simulated, the operator splitting approach is best suited for pricing path dependent options. Related techniques of operator splitting have been used by \citet{bayer2024efficient} for the lifted Heston model and \citet{Choi2025SABR} for the SABR model with good successes. In this paper, we formulate this operator splitting approach more systematically and apply the approach to most common types of stochastic volatility models, including the Heston-type models, lifted Heston model, HW model and BN-S model. Our numerical experiments demonstrate that the operator splitting approach outperforms Algorithm 2 in \citet{brigone2024104861} for the HW model, Hilbert interpolation schemes for the 4/2 model in \citet{Zeng2023}, weak simulation scheme for the lifted Heston model in \citet{bayer2024efficient} and Euler schemes under all of the above stochastic volatility models.
  We also establish the theoretical proof of second-order convergence of the operator splitting schemes.

  The remaining sections of the paper are organized as follows. The next three sections focus on the theoretical derivation of the operator splitting schemes for various types of stochastic volatility models. Section 2 discusses the general theory and procedures of the operator splitting approach, as illustrated by the Heston model, 4/2 model and HW model. Section 3 shows the construction of the operator splitting scheme for the lifted Heston model and its solution procedures. Section 4 discusses the operator splitting schemes for the BN-S model for pricing path dependent options and VIX derivatives. Section 5 presents the theoretical proof of second-order convergence of the operator splitting schemes. Section 6 documents the numerical experiments that were conducted to demonstrate computational efficiency, robustness and versatility of the operator splitting schemes for pricing various types of options under various stochastic volatility models. The last section contains summary of the results and conclusive remarks.    
	\section{Second-order operator splitting scheme for the Heston model, 4/2 model and Hull-White model}\label{s:42_heston_model_split}\noindent
	The essence of the operator splitting procedure can be best illustrated by the Heston model, which is the prototype model in stochastic volatility models. We then extend the construction of the operator splitting scheme for the 4/2 model. Though the 4/2 model includes the Heston model as a special case, the operator splitting scheme for the 4/2 model is more involved. Lastly, we present the operator splitting scheme for the HW model.
	\subsection{Second-order operator splitting scheme for the Heston model}\label{ss:heston_split}\noindent
	 We start with a complete probability space $\{\Omega,F,Q\}$, where $Q$ is a risk neutral measure. Let $S_t$ denote the asset price process of the Heston model and define the log-asset price $X_t=\ln S_t$. The governing stochastic differential equations (SDEs) for $X_t$ and its instantaneous variance process $V_t$ under $Q$ for the Heston model are given by
	 \begin{equation}\label{eq:sde_heston}
		\begin{cases}
			\text{d}X_t=(r-\frac{V_t}{2})\,\text{d}t+\sqrt{V_t}\,\text{d}Z_t,\\
			\text{d}V_t=k(\theta-V_t)\,\text{d}t+\sigma\sqrt{V_t}\,\text{d}W_t,
		\end{cases}
	\end{equation}
	where the reversion speed $k$, mean reversion level $\theta$ and volatility of variance $\sigma$ are all positive, and $r$ is the riskless interest rate. Here, $V_t$ is governed by the Cox-Ingersoll-Ross (CIR) process. The correlation coefficient of the two Brownian motions $Z_t$ and $W_t$ is $\rho$. We write $$Z_t=\rho W_t+\sqrt{1-\rho^2}W_t^{\perp},$$ where $W_t^{\perp}$ is the Brownian motion independent of $W_t$. The direct integration of the dynamic equation of $X_t$ over $[t,\,t+h]$ gives the integral representation:
	\begin{equation*}
		X_{t+h}=X_t+rh-\frac 12 \int_t^{t+h}V_s\,\text{d}s+\rho\int_t^{t+h}\sqrt{V_s}\,\text{d}W_s+\sqrt{1-\rho^2}\int_t^{t+h}\sqrt{V_s}\,\text{d}W^{\perp}_s.
	\end{equation*}
	From the SDE of $V_t$ in Eq.~\eqref{eq:sde_heston}, we have 
	\begin{equation*}
		\sqrt{V_s}\,\text{d}W_s=\frac{\text{d}V_s}{\sigma}-\frac{k}{\sigma}(\theta-V_s)\,\text{d}s.
	\end{equation*}
	Combining the above two equations, we obtain the analytic representation of $X_{t+h}$ as follows:
	\begin{equation}\label{eq:asset_heston}
		\begin{aligned}
			X_{t+h}=X_t+rh+(\frac{\rho k}{\sigma}-\frac12)I_t^h+\frac{\rho}\sigma(V_{t+h}-V_t-k\theta h)+N\left(0,(1-\rho^2)I_t^h\right),
		\end{aligned} 
	\end{equation}
	where $I_t^h=\int_t^{t+h}V_s\,\text{d}s$ is the integrated variance over $[t,\,t+h]$. Here, $N(\mu,\sigma^2)$ refers to the normal random variable with mean $\mu$ and variance $\sigma^2$. 
	
	For pricing path dependent derivatives, the common procedure is to 
	simulate the conditional integrated variance $I_t^h|V_t,V_{t+h}$. 
	Correspondingly, it is necessary to simulate $V_{t+h}$ first at each time 
	step, then simulate the conditional integrated variance $I_t^h|V_t,\,V_{t+h}$. 
	The simulation of $V_{t+h}$ has been well studied in the literature. 
	\citet{andersen2008simple} propose the quadratic-exponential scheme
	 to overcome the challenge of dealing with the square root term in 
	 the variance process. We recommend the exact simulation scheme for 
	 the CIR process as proposed by \citet{alfonsi2015affine}, which is 
	 reliable and easily implementable (see \ref{ap:CIR_simulation}). 
	 To simulate $X_{t+h}$, the challenge is to simulate $I_t^h \mid V_t, V_{t+h}$, 
	 which involves the integration of the CIR process.
	
	\citet{Broadie} derive the characteristic function of the conditional 
	integrated variance and perform the simulation of $I_t^h|V_t,\,V_{t+h}$ via the 
	Fourier inversion of the corresponding characteristic function to obtain the distribution 
	function, then followed by the root-finding procedure to obtain the simulation value of conditional
	$I_t^h$. The integral transform inversion is performed for each simulation path, so it is computationally demanding. 
	An enhanced simulation procedure is adopted by \citet{Zeng2023} via their Hilbert 
	interpolation schemes. Computational efficiency of their scheme is enhanced by an interpolation scheme. For the Heston model, there are 
	other simulation approaches that are more efficient than the above integral 
	transform inversion procedure. Here, we summarize the three other more efficient Monte Carlo simulation methods
	 to simulate $I_t^h \mid V_t, V_{t+h}$. \citet{andersen2008simple} simply uses the 
	quadrature rule, like the trapezoidal rule, to evaluate the integrated 
	variance over $[t,\,t+h]$ once the simulated values of $V_t$ and $V_{t+h}$ are 
	available. \citet{tse2012low} and \citet{AbiJaber2024heston} use the inverse 
	Gaussian distribution as an approximation via moment matching. The most efficient 
	method is the Poisson-conditioned scheme for the direct simulation of $I_t^h$ via 
	an infinite series of gamma expansion \citep{choi_heston2024}. However, the success
	 of direct simulation used in Choi-Kwok's scheme is limited to the Heston model. Its 
	 extension to the 4/2 model is not feasible. 
	
	We illustrate how to use the second-order operator splitting approach to 
	derive an efficient simulation scheme that circumvents the evaluation 
	of conditional integrated variance. Our operator splitting approach shows general 
	applicability to most common stochastic volatility models, not limited to 
	the Heston model.
%
	It is desirable to choose the second-order operator splitting so that integration of the CIR process of $V_t$ is avoided. We split the operators in the system of the SDEs of the Heston model given in Eq.~\eqref{eq:sde_heston} as follows:
	\begin{subequations}
		\begin{equation}\label{eq:heston_split_first_sde}
			\begin{cases}
				\text{d}X^{(1)}_t=(r-\frac{V^{(1)}_t}2)\,\text{d}t+\sqrt{1-\rho^2}\sqrt{V^{(1)}_t}\,\text{d}W_t^{\perp},\\
				\text{d}V^{(1)}_t=k(\theta-V^{(1)}_t)\,\text{d}t,
			\end{cases}
		\end{equation}
		\begin{equation}\label{eq:heston_split_second_sde}
			\begin{cases}
				\text{d}X^{(2)}_t=\rho\sqrt{V^{(2)}_t}\,\text{d}W_t,\\
				\text{d}V^{(2)}_t=\sigma\sqrt{V^{(2)}_t}\,\text{d}W_t.
			\end{cases}
		\end{equation}
	\end{subequations}
	The solutions of above two pairs of SDEs are given by
	\begin{subequations}\label{eq:heston_split_sol}
		\begin{equation}\label{eq:heston_split_first_sol}
			\begin{cases}
				X^{(1)}_{t+h}=X^{(1)}_t+rh-\frac12 I^{(1),h}_t+N\left(0,(1-\rho^2)I^{(1),h}_t\right),\\
				V^{(1)}_{t+h}=e^{-kh}(V^{(1)}_t-\theta)+\theta,\\
				I^{(1),h}_t=\int_t^{t+h} V^{(1)}_s\,\text{d}s=\theta h+\frac{V^{(1)}_t-\theta}{k}(1-e^{-kh}),
			\end{cases} 
		\end{equation}
		\begin{equation}\label{eq:heston_split_second_sol}
			\begin{cases}
				X^{(2)}_{t+h}=X^{(2)}_t+\frac{\rho}{\sigma}(V^{(2)}_{t+h}-V^{(2)}_{t}),\\
				V^{(2)}_{t+h}=CIR(0,0,\sigma,h,V^{(2)}_{t}),
			\end{cases} 
		\end{equation}
	\end{subequations}
	where $CIR(\alpha,\beta,\sigma,h,V_0)$ refers to the solution 
	$V_h$ of the CIR process $$\text{d}V_t=(\alpha-\beta V_t)\,\text{d}t+\sigma\sqrt{V_t}\,\text{d}W_t,$$ 
	given $V_0$ and over time step $h$. Note that we choose the appropriate operator splitting 
	in $V_t$ so that 
	the SDE for $V_t^{(1)}$ involves no stochastic term. As a result, $I_t^{(1),h}=\int_t^{t+h} V_s^{(1)}\,\text{d}s$ can be expressed 
	explicitly in terms of $V_t^{(1)}$. We recommend the use of the simulation algorithms in 
	\citet[Prop~3.1.2]{alfonsi2015affine} to 
	simulate the CIR process for $4\alpha\geq\sigma^2$ and \citet[Prop~3.1.3]{alfonsi2015affine} for the other cases, 
	the details of which are summarized in \ref{ap:CIR_simulation}. 
	The second-order operator splitting procedure as advocated by \cite{strang1968construction} is illustrated generically via the following 
	differential equation:
	\begin{equation*}
		\frac{\text{d}}{\text{d}t}x(t)=U_1\left[x(t)\right]+U_2\left[x(t)\right],\,t>0,
	\end{equation*}
	where $U_1$ and $U_2$ are separate differential operators on $x(t)$. In the second-order operator splitting scheme, the solution operator $e^{h(U_1+U_2)}$ of the full differential equation over the time interval $[0,\,h]$ is approximated by 
	\begin{equation}\label{eq:split}
		x(h)=e^{h(U_1+U_2)}x(0)\approx e^{\frac{h}{2}U_1}e^{h U_2}e^{\frac{h}{2}U_1}x(0),
	\end{equation}
	where $e^{hU_1}$ and $e^{hU_2}$ represent the solution operators of the following splitting equations:
	\begin{equation*}
		\frac{\text{d}}{\text{d}t}x(t)=U_1\left[x(t)\right]\quad\text{and}\quad\frac{\text{d}}{\text{d}t}x(t)=U_2\left[x(t)\right],
	\end{equation*}
	respectively. By adopting the above operator splitting procedure to the splitting equations~\eqref{eq:heston_split_first_sde} and \eqref{eq:heston_split_second_sde}, we obtain the following second-order operator splitting scheme of the Heston model as stated in Proposition~\ref{prop:heston}. The proof of the second-order convergence of
	the scheme is established in \Cref{prop:terminal-weak-error}.
	\begin{proposition}\label{prop:heston}
		The solution of $(X,V)$ in the second-order operator splitting scheme of the Heston model in Eq.~\eqref{eq:sde_heston} is given by
		\begin{equation}\label{eq:sol_heston_split}
			\begin{cases}
				X^{\text{split}}_{t+h}=X_t+rh-\frac12I^{\text{split},h}_t+\frac{\rho}{\sigma}(\hat{V}^{(2)}_{t+h}-\tilde{V}^{(1)}_{t+h/2})+N\left(0,(1-\rho^2)I^{\text{split},h}_t\right),\\
				V^{\text{split}}_{t+h}=e^{-kh/2}(\hat{V}^{(2)}_{t+h}-\theta)+\theta,
			\end{cases}
		\end{equation}
		where $\tilde{V}^{(1)}_{t+h/2}=e^{-kh/2}(V_t-\theta)+\theta$, $\hat{V}^{(2)}_{t+h}\sim CIR(0,0,\sigma,h,\tilde{V}^{(1)}_{t+h/2})$ and $I^{\text{split},h}_t=\theta h+\frac{V_t+\hat{V}^{(2)}_{t+h}-2\theta}{k}(1-e^{-kh/2}).$
	\end{proposition}
	 \begin{proof}
	 	See \ref{ap:proof_heston}.
	 \end{proof}
	  \begin{remark}\label{re:heston_posi}
	  	Note that $V_{t+h}^{\text{split}}=e^{-kh/2}\hat{V}^{(2)}_{t+h}+\theta(1-e^{-kh/2})$ is always nonnegative, provided that $\hat{V}^{(2)}_{t+h}\geq 0$. Observe that $\hat{V}^{(2)}_{t+h}$ obtained via simulation of the CIR process is automatically nonnegative. Therefore, unlike the Euler scheme, the above second-order operator splitting scheme does not require any truncation on $V$ to circumvent potential negative value of $V$.
%
	 \end{remark}
	  \begin{remark}\label{re:heston_choi}
	  	The most efficient simulation scheme for the Heston model would be the Poisson-conditioned simulation scheme in \citet{choi_heston2024}.
		Their scheme achieves high accuracy and computational efficiency
		even for large time step $h$. 
		For example, in the pricing of European call option, 
		their scheme can agree with the benchmark value at least up to 
		three significant digits under $h=T$, where $T$ is the time to 
		maturity while their CPU time is below 0.1 seconds 
		[see Table~4 in \citet{choi_heston2024}]. 
		However, their approach of using an infinite series of Poisson-conditioned
		gamma expansion of integrated variance cannot be extended to the 4/2 model due to the presence of the term $\frac{1}{V_t}$ in the SDE. Our operator splitting approach shows general applicability to most common stochastic volatility models. 
	  \end{remark}

	  \begin{remark}
The straightforward approximation of $I_t^h$ in Eq.~\eqref{eq:asset_heston} by 
the trapezoidal rule, $I_t^h \approx \frac{h}{2}(V_t+V_{t+h})$, gives an 
alternative simple simulation scheme for the Heston model as follows:

\begin{equation}
\begin{split}
    X_{t+h}^{\mathrm{trap,Heston}} &= X_t + \frac{h}{2}\left(\frac{\rho k}{\sigma} - \frac{1}{2}\right)(V_t + V_{t+h}) + \frac{\rho}{\sigma}(V_{t+h} - V_t - k\theta h) \\
    &\quad + N\left(0, \frac{h}{2}(1-\rho^2)(V_t+V_{t+h})\right).
\end{split}
\end{equation}
Surprisingly, this simulation scheme as derived from the trapezoidal rule performs quite well. 
Its performance is only slightly below that of the operator splitting scheme 
(\ref{eq:sol_heston_split}) in terms of accuracy versus simulation time 
tradeoff [see numerical comparison with the second-order operator splitting scheme (\ref{eq:sol_heston_split}) in Subsection \ref{subsection:heston_4_2_hull_white}].
\end{remark}
	  \subsection{Second-order operator splitting scheme for the 4/2 model}\label{ss:42_split}\noindent
	  In this subsection, we demonstrate the use of operator splitting as applied to the 4/2 model, which has one additional term $\frac1{V_t}$ in the variance of $X_t$ when compared with the Heston model. The governing SDEs of the log-asset price $X_t$ and $V_t$ of the 4/2 model \citep{Grasselli} are given by
	  \begin{equation}\label{eq:sde_42}
	  	\begin{cases}
	  		\text{d}X_t=(r-ab-\frac{a^2V_t}{2}-\frac{b^2}{2V_t})\,\text{d}t+(a\sqrt{V_t}+\frac{b}{\sqrt{V_t}})\,\text{d}Z_t,\\
	  		\text{d}V_t=k(\theta-V_t)\,\text{d}t+\sigma\sqrt{V_t}\,\text{d}W_t,
	  	\end{cases}
	  \end{equation}
	  where $k$, $\theta$ and $\sigma$ are all non-negative. The correlation coefficient of the two Brownian motions $Z_t$ and $W_t$ is $\rho$. We focus on the case where $b\neq 0$. When $b=0$, the 4/2 model reduces to the Heston model, which has been considered in Subsection~\ref{ss:heston_split}.
	  By performing the direct integration of the dynamic equation of $X_t$, we obtain
	  \begin{equation*}
	  	\begin{aligned}
	  		X_{t+h}&=X_t+(r-ab)h-\frac{a^2}2\int_t^{t+h}V_s\,\text{d}s-\frac{b^2}2\int_t^{t+h}\frac1{V_s}\,\text{d}s+\rho a\int_t^{t+h}\sqrt{V_s}\,\text{d}W_s\\&+\rho b\int_t^{t+h}\frac1{\sqrt{V_s}}\,\text{d}W_s+\sqrt{1-\rho^2}\int_t^{t+h}(a\sqrt{V_s}+\frac b{\sqrt{V_s}})\,\text{d}W_s^{\perp}.
	  	\end{aligned}
	  \end{equation*}
	  To resolve the occurrence of the reciprocal terms of $V_t$ and $\sqrt{V_t}$ in the dynamics of $X_t$, it is instructive to consider the dynamics of $\ln V_t$ as follows:
	  $$\text{d}\ln V_t=\left[(k\theta-\frac{\sigma^2}2)\frac 1{V_t}-k\right]\text{d}t+\frac{\sigma}{\sqrt{V_t}}\,\text{d}W_t.$$ 
	  Putting the results together, we obtain the following analytic representation formula of $X_{t+h}$ under the 4/2 model:
	  \begin{equation}\label{eq:asset_42}
	  	\begin{aligned}
	  		X_{t+h}&=X_t+(r-ab)h-\frac{a^2}2I_t^h-\frac{b^2}{2}J_t^h+\frac{\rho a}\sigma(V_{t+h}-V_t-k\theta h+kI_t^h)\\&+\frac{\rho b}{\sigma}\left[\ln{V_{t+h}}-\ln{V_t}+kh+(\frac{\sigma^2}{2}-k\theta)J_t^h\right]+N\left(0,(1-\rho^2)(a^2I_t^h+b^2J_t^h+2abh)\right),
	  	\end{aligned} 
	  \end{equation}
	  where $I_t^h=\int_t^{t+h}V_s\,\text{d}s$ and $J_t^h=\int_t^{t+h}\frac1{V_s}\,\text{d}s$.
	    When compared with the Heston model, there is an additional integrated reciprocal variance $J_t^h$.
		 For pricing path dependent derivatives, 
		it is necessary to simulate the conditional stochastic integrals $I_t^h|V_t,V_{t+h}$
		and $J_t^h|V_t,V_{t+h}$. \citet{Baldeaux} derive 
		the Laplace transform of the conditional distribution of the integrated 
		reciprocal variance $J_t^h$ in the 3/2 model, which 
		corresponds to the case of $a=0$ in the 4/2 model. 
		He simulates the conditional integrated reciprocal variance via 
		the inverse transform technique and root-finding procedure. 
		\citet{Zeng2023} derive the conditional moment generating function (MGF) of the terminal log-asset price of the 4/2 model and 
		construct their simulation scheme via their Hilbert interpolation 
		algorithm. These methods are invariably tedious to implement and 
		time-consuming. 
	    
	    By following a similar idea as illustrated in the Heston scheme, we propose the following operator splitting procedures to avoid the challenge involved in evaluation of $I_t^h$ and $J_t^h$ as follows:
	   \begin{subequations}
	  	\begin{equation}\label{eq:42_split_first_sde}
	  		\begin{cases}
	  			\text{d}X^{(1)}_t=(r-ab-\frac{a^2V^{(1)}_t}{2}-\frac{b^2}{2V^{(1)}_t})\,\text{d}t+\sqrt{1-\rho^2}(a\sqrt{V^{(1)}_t}+\frac{b}{\sqrt{V^{(1)}_t}})\,\text{d}W_t^{\perp},\\
	  			\text{d}V^{(1)}_t=(k\theta-\frac{\sigma^2}2-kV^{(1)}_t)\,\text{d}t=k(\tilde{\theta}-V^{(1)}_t)\,\text{d}t,
	  		\end{cases}
	  	\end{equation}
	  	\begin{equation}\label{eq:42_split_second_sde}
	  		\begin{cases}
	  			\text{d}X^{(2)}_t=\rho(a\sqrt{V^{(2)}_t}+\frac{b}{\sqrt{V^{(2)}_t}})\,\text{d}W_t,\\
	  			\text{d}V^{(2)}_t=\frac{\sigma^2}2\,\text{d}t+\sigma\sqrt{V^{(2)}_t}\,\text{d}W_t,
	  		\end{cases}
	  	\end{equation}
	  \end{subequations}
	  where $\tilde{\theta}=\theta-\frac{\sigma^2}{2k}.$ Since we assume $b\neq0$, the presence of the term $\frac1{V_t}$ dictates that the Feller condition $2k\theta\geq \sigma^2$ must hold in order to guarantee that $V$ cannot achieve the zero value. Accordingly, we assume $\tilde{\theta}\geq 0$.
	  
	  Note that we do not split the SDE of $V$ in the 4/2 model in the same manner as that of the Heston model in Eqs.~(\ref{eq:heston_split_first_sde}, \ref{eq:heston_split_second_sde}). 
	  Instead, we include an extra term $\frac{\sigma^2}2\,\text{d}t$ in the second step of operator splitting procedure in Eq.~\eqref{eq:42_split_second_sde} and substract the same term in the drift term of the first step of operator splitting procedure in Eq.~\eqref{eq:42_split_first_sde}. The reason is that we have the term $\frac{b}{\sqrt{V^{(2)}_t}}\,\text{d}W_t$ in the SDE of $X_t^{(2)}$, so we require the SDE of $\ln V_t^{(2)}$ to include $\frac 1{\sqrt{V^{(2)}_t}}\,\text{d}W_t$ in its SDE. Also, we need to choose an appropriate drift term of $V_t^{(2)}$ to avoid $\frac 1{V_t^{(2)}}$ in the SDE of $\ln V_t^{(2)}$, so that the integration of $\int_{t}^{t+h}\frac 1{V_s^{(2)}}\,\text{d}s$ does not appear in the analytic representation formula of $X_{t+h}^{(2)}$. Based on the SDE of $V_t^{(2)}$ in Eq.~\eqref{eq:42_split_second_sde}, we obtain \begin{equation*}
	  	\text{d}\ln{V^{(2)}_t}=\frac{\sigma}{\sqrt{V^{(2)}_t}}\,\text{d}W_t,
	  \end{equation*}
	  so that it does not include the integration of $\frac1{V^{(2)}_t}$ in the solution of $X^{(2)}_{t+h}$ from Eq.~\eqref{eq:42_split_second_sde}. 
	  
	  The solutions of the two-step operator splitting procedures for the 4/2 model are given by
	  \begin{subequations}
	  	\begin{equation}\label{eq:42_split_first_sol}
	  		\begin{cases}
	  			X^{(1)}_{t+h}=X^{(1)}_t+(r-ab)h-\frac12 (I^{(1),h}_t+J^{(1),h}_t)+N\left(0,(1-\rho^2)(a^2I^{(1),h}_t+b^2J^{(1),h}_t+2abh)\right),\\
	  			V^{(1)}_{t+h}=e^{-kh}(V^{(1)}_t-\tilde{\theta})+\tilde{\theta},\\
	  			I^{(1),h}_t=\int_t^{t+h} V^{(1)}_s\,\text{d}s=\tilde{\theta} h+\frac{V^{(1)}_t-\tilde{\theta}}{k}(1-e^{-kh}),\\
	  		\end{cases} 
	  	\end{equation}
	  	\begin{equation}\label{eq:42_split_second_sol}
	  		\begin{cases}
	  			X^{(2)}_{t+h}=X^{(2)}_t+\frac{\rho a}{\sigma}(-\frac{\sigma^2 h}{2}+V^{(2)}_{t+h}-V^{(2)}_{t})+\frac{\rho b}{\sigma}(\ln{V^{(2)}_{t+h}}-\ln{V^{(2)}_{t}}),\\
	  			V^{(2)}_{t+h}=CIR(\frac{\sigma^2}2,0,\sigma,h,V^{(2)}_{t}),
	  		\end{cases} 
	  	\end{equation}
	  \end{subequations}
	  where \begin{equation*}
	  		J^{(1),h}_t=\int_t^{t+h}\frac1{V^{(1)}_s}\,\text{d}s=\begin{cases}
	  			\frac1{k\tilde{\theta}}\ln{|\frac{V^{(1)}_t+\tilde{\theta}(e^{kh}-1)}{V^{(1)}_t}|},\,\text{for }\tilde{\theta}>0,\\
	  		\frac{e^{kh}-1}{kV^{(1)}_t},\,\text{for }\tilde{\theta}=0.
	  		\end{cases}
	  \end{equation*}
	 Based on the above representation of solutions of the split ODE systems, 
	 we deduce the operator splitting simulation scheme of the 4/2 model as summarized in 
	 Proposition~\ref{prop:42}.
	  \begin{proposition}\label{prop:42}
	  	The second-order operator splitting scheme of the 4/2 model, given $X_t$ and $V_t$, is given by
	  	\begin{equation}\label{eq:sol_split_42}
	  		\begin{cases}
	  			\begin{aligned}
	  				X^{\text{split}}_{t+h}&=X_t+rh-\frac{C_t^h}2+\frac{\rho a}{\sigma}(-\frac{\sigma^2 h}{2}+\hat{V}^{(2)}_{t+h}-\tilde{V}^{(1)}_{t+h/2})+\frac{\rho b}{\sigma}(\ln{\hat{V}^{(2)}_{t+h}}-\ln{\tilde{V}^{(1)}_{t+h/2}})\\&+N\left(0,(1-\rho^2)C_t^h\right),
	  			\end{aligned}
	  			\\
	  			V^{\text{split}}_{t+h}=(\hat{V}^{(2)}_{t+h}-\tilde{\theta})e^{-kh/2}+\tilde{\theta},\\
	  			C_t^h=2abh+a^2\left[\tilde{\theta}h+\frac{V_t+\hat{V}^{(2)}_{t+h}-2\tilde\theta}{k}(1-e^{-kh/2})\right]+b^2\tilde{J}_t^h,
	  		\end{cases}
	  	\end{equation}
	  	where $\tilde{V}^{(1)}_{t+h/2}=(V_t-\tilde{\theta})e^{-kh/2}+\tilde{\theta}$, $\hat{V}^{(2)}_{t+h}\sim CIR(\frac{\sigma^2}{2},0,\sigma,h,\tilde{V}^{(1)}_{t+h/2})$ and 
	  	\begin{equation*}
	  		\tilde{J}_t^h=\begin{cases}
	  			\frac{e^{kh/2}-1}{k}(\frac1{V_t}+\frac1{\hat{V}^{(2)}_{t+h}}), \text{ for }\tilde{\theta}=0,\\
	  			\frac1{k\tilde{\theta}}\left(\ln{|\frac{V_t+\tilde{\theta}(e^{kh/2}-1)}{V_t}|}+\ln{|\frac{\hat{V}^{(2)}_{t+h}+\tilde{\theta}(e^{kh/2}-1)}{\hat{V}^{(2)}_{t+h}}|}\right), \text{ for }\tilde{\theta}>0
	  		\end{cases}
	  	\end{equation*} 
	  \end{proposition}
	  \begin{proof}
	  	Following a similar derivation procedure as that of the Heston model in \ref{ap:proof_heston}.
	  \end{proof}
	    
	  \begin{remark}\label{re:42_posi}
	  	Similar to Remark~\ref{re:heston_posi}, we do not require any 
		truncation on $V_{t+h}^{\text{split}}$ for the 4/2 model. 
		Since the Feller condition holds, we observe $\tilde{\theta}=\theta-\frac{\sigma^2}{2k}\geq0$. As shown in Eq.~\eqref{eq:42_split_first_sol}, $V^{(1)}_{t+h}=e^{-kh/2}V^{(1)}_t+\tilde{\theta}(1-e^{-kh/2})$ is always positive if $V^{(1)}_t> 0$. In Eq.~\eqref{eq:42_split_second_sol}, the CIR process $CIR(\frac{\sigma^2}2,0,\sigma,h,V_t^{(2)})$ satisfies the Feller condition, so the simulated random variable of this CIR process is always positive. 
	  \end{remark}
	  	  \subsection{Second-order operator splitting scheme for the Hull-White model}\label{ss:hull_white_model_split}\noindent 
	  The governing SDEs of the log-asset price $X_t$ and its variance process $V_t$ in the HW model are given by
	  \begin{equation}\label{eq:sde_hull_white}
	  	\begin{cases}
	  		\text{d}X_t=(r-\frac{V_t}2)\,\text{d}t+\sqrt{V_t}\,\text{d}Z_t,\\
	  		\text{d}V_t=\mu V_t\,\text{d}t+\xi V_t\,\text{d}W_t,
	  	\end{cases}
	  \end{equation}
	  where the correlation coefficient of the two Brownian motions $Z_t$ and $W_t$ is $\rho$. By following similar technical calculations, we obtain the analytic representation formulas of $X_{t+h}$ and $V_{t+h}$ as follows:
	  \begin{equation}\label{eq:formula_hull_white}
	  	\begin{cases}
	  		X_{t+h}=X_t+rh-\frac12I_t^h+\frac{2\rho}{\xi}(\sqrt{V_{t+h}}-\sqrt{V_t})+(\frac{\rho\xi}{4}-\frac{\rho\mu}{\xi})J_t^h+N\left(0,(1-\rho^2)I_t^h\right),\\
	  		V_{t+h}=V_te^{(\mu-\frac{\xi^2}{2})h+\xi \Delta W_t^h},
	  	\end{cases}
	  \end{equation}
	  where $I_t^h=\int_t^{t+h}V_s\,\text{d}s$, $J_t^h=\int_t^{t+h}\sqrt{V_s}\,\text{d}s$, $\Delta W_t^h=W_{t+h}-W_t$ and 
	  the normal random variable $N\left(0,(1-\rho^2)I_t^h\right)$ is independent 
	  of the variance process $V$. 
	  The challenge lies on the joint simulation of $I_t^h$ and $J_t^h$, 
	  which involves the integrals of the log normal random variable. 
	  \citet{brigone2024104861} derive the Laplace transform of the 
	  reciprocal of $J_t^h$ conditional on $V_t$ and $V_{t+h}$ 
	  and the Laplace transform of $X_{t+h}$ given $J_t^h$, $V_t$ and $V_{t+h}$. 
	  They simulate the log-asset price by inverting the above two Laplace transforms 
	  via the Fourier-cosine method. Their approach still requires 
	  the tedious inverse Fourier transform and root-finding procedure. In our operator splitting approach, we avoid the integration of the log normal random variable as follows:
	  \begin{subequations}
	  	\begin{equation}\label{eq:split_hull_sde1}
	  		\begin{cases}
	  			\text{d}X^{(1)}_t=(r-\frac{V_t^{(1)}}{2})\,\text{d}t+\sqrt{1-\rho^2}\sqrt{V^{(1)}_t}\,\text{d}W_t^{\perp},\\
	  			\text{d}V^{(1)}_t=(\mu-\frac{\xi^2}4)V^{(1)}_t\,\text{d}t,
	  		\end{cases}
	  	\end{equation}
	  	\begin{equation}\label{eq:split_hull_sde2}
	  		\begin{cases}
	  			\text{d}X^{(2)}_t=\rho\sqrt{V^{(2)}_t}\,\text{d}W_t,\\
	  			\text{d}V^{(2)}_t=\frac{\xi^2}{4}V^{(2)}_t\,\text{d}t+\xi V^{(2)}_t\,\text{d}W_t,
	  		\end{cases}
	  	\end{equation}
	  \end{subequations}
	  so that the solutions of these two parts are given by
	  \begin{subequations}\label{eq:split_hull_sol}
	  	\begin{equation}\label{eq:split_hull_sol1}
	  		\begin{cases}
	  			X^{(1)}_{t+h}=X^{(1)}_t+rh-\frac{I^{(1),h}_t}2+N\left(0,(1-\rho^2)I^{(1),h}_t\right),\\
	  			V^{(1)}_{t+h}=V^{(1)}_te^{ch}
	  		\end{cases}
	  	\end{equation}
	  	\begin{equation}\label{eq:split_hull_sol2}
	  		\begin{cases}
	  			X^{(2)}_{t+h}=X^{(2)}_t+\frac{2\rho}{\xi}(\sqrt{V^{(2)}_{t+h}}-\sqrt{V^{(2)}_t}),\\
	  			V^{(2)}_{t+h}=V^{(2)}_t e^{-\frac{\xi^2h}{4}+\xi\Delta W_t^h},
	  		\end{cases}
	  	\end{equation}
	  \end{subequations}
	  where $c=\mu-\frac{\xi^2}{4}$, $I^{(1),h}_t=\int_t^{t+h}V^{(1)}_s\,\text{d}s=\frac{V^{(1)}_t}{c}(e^{ch}-1)$ for $c\neq0$ and $I^{(1),h}_t=\int_t^{t+h}V^{(1)}_s\,\text{d}s=V^{(1)}_t h$ for $c=0$. By splitting the original SDEs of the HW model in Eq.~\eqref{eq:sde_hull_white} as in the above, we successfully avoid the integration of the random variable and only require to integrate the deterministic variable $V^{(1)}_s$ in Eq.~\eqref{eq:split_hull_sde1}.
	  %
	  By following a similar procedure as illustrated by the Heston model, we obtain the following proposition for the operator splitting scheme of the HW model.
	  \begin{proposition}\label{prop: hwmodel}
	  	The second-order operator splitting scheme of the HW model given $X_t$ and $V_t$ is given by
	  	\begin{equation}\label{eq:hull_sec_split}
	  		\begin{cases}
	  			X^{\text{split}}_{t+h}=X_t+rh-\frac12 \tilde{I}_t^h+\frac{2\rho}{\xi}(\sqrt{\hat{V}^{(2)}_{t+h}}-\sqrt{\tilde{V}^{(1)}_{t+h/2}})+N\left(0,(1-\rho^2)\tilde{I}_t^h\right),\\
	  			V^{\text{split}}_{t+h}=e^{ch/2}\hat{V}^{(2)}_{t+h},
	  		\end{cases}
	  	\end{equation}
	  	where $\tilde{V}^{(1)}_{t+h/2}=e^{ch/2}V_t$, $\hat{V}^{(2)}_{t+h}=\tilde{V}^{(1)}_{t+h/2} e^{-\frac{\xi^2}{4}h+\xi\Delta W^h_t}$ and
	  	\begin{equation*}
	  		\tilde{I}^h_t=\begin{cases}
	  			\frac{V_t+\hat{V}^{(2)}_{t+h}}{c}(e^{ch/2}-1) ,\text{ for } c\neq0,\\
	  			\frac{(V_t+\hat{V}^{(2)}_{t+h})h}2, \text{ for } c=0.
	  			\end{cases}
	  	\end{equation*} 
	  \end{proposition}
	  \begin{remark}\label{re:hull_white}
	  	The second-order operator splitting scheme for the Hull-White model in \Cref{prop: hwmodel} does not need any truncation on the variance process $V$. This is because the solutions of $V$ in both in Eqs.~(\ref{eq:split_hull_sol}a,b) are both positive given that the initial value of $V_t$ is positive.
	  \end{remark}

	  \begin{remark}
Similar to the Heston model, we may use the trapezoidal rule to 
obtain the approximations: $I_t^h \approx \frac{h}{2}(V_t+V_{t+h})$ 
and $J_t^h \approx \frac{h}{2}(\sqrt{V_t}+\sqrt{V_{t+h}})$ in Eq.~\eqref{eq:formula_hull_white}
for the HW model. The resulting simplified trapezoidal simulation scheme 
of the HW model also performs quite well [see the numerical comparison with 
the operator splitting scheme (\ref{eq:hull_sec_split}) in Section \ref{subsection:heston_4_2_hull_white}].
\end{remark}
	  \section{Second-order operator splitting scheme for the lifted Heston model}\label{s:lifted_heston_model_split}\noindent
	  The lifted Heston model is the Markovian version of the rough Heston model 
	  as proposed by \citet{abi2019lifting}. We adopt the formulation of SDEs 
	  of the lifted Heston model in \citet{bayer2024efficient}. The variance 
	  process $V_t^N$ consists of $N$ components, where $V_t^N=\sum_{i=1}^Nw_iV_{t,i}$, $\theta\geq0$, $\lambda\geq0$.
	   Here, $w_i$ are called the weights. We adopt the assumption that $w_i\geq0,\,i=1,2,...,N.$ The governing SDEs of the log-asset price $X_t$ and the variance component $V_{t,i}$ are given by
	  \begin{equation}\label{eq:lifted_heston_sde}
	  	\begin{cases}
	  		\text{d}X_t=-\frac{V_t^N}{2}\,\text{d}t+\sqrt{V_t^N}\,\text{d}Z_t,\\
	  		\text{d}V_{t,i}=-x_i(V_{t,i}-v_0^i)\,\text{d}t+(\theta-\lambda V_t^N)\,\text{d}t+\nu\sqrt{V_t^N}\,\text{d}W_t,\, i=1,2,...,N,
	  	\end{cases}
	  \end{equation}
	  with initial condition: $V_{0,i}=v_0^i.$ Here, $x_i \ge 0$ are called the nodes. 
	  All variance components $V_{t,i}$ share the same Brownian motion $W_t$. 
	  The correlation coefficient of the two Brownian motions $Z_t$ and $W_t$ 
	  is $\rho$. By performing integration of the SDEs over $[t,\,t+h]$, we obtain
	  \begin{equation}\label{eq:asset_lifted_heston}
	  	\begin{aligned}
	  		X_{t+h}&=X_t-\frac{(x_iv_0^i+\theta)\rho h}{\nu}+\frac{\rho}{\nu}(V_{t+h,i}-V_{t,i})+(\frac{\rho\lambda}{\nu}-\frac12)\int_t^{t+h}V^N_s\,\text{d}s\\&+\frac{\rho x_i}{\nu}\int_t^{t+h}V_{s,i}\,\text{d}s+N\left(0,(1-\rho^2)\int_t^{t+h}V^N_s\,\text{d}s\right).
	  	\end{aligned}
	  \end{equation}
	  Again, the challenging part is to simulate the integrated variance terms $\int_t^{t+h}V_{s,i}\,\text{d}s$ and $\int_t^{t+h}V^N_s\,\text{d}s$ jointly. Accordingly, similar to the Heston model and 4/2 model, we split the above SDEs into the following two parts:
	  \begin{subequations}
	  	\begin{equation}\label{eq:lifted_heston_split_first_sde}
	  		\begin{cases}
	  			\text{d}X^{(1)}_t=-\frac{V_t^{N,(1)}}{2}\,\text{d}t+\sqrt{1-\rho^2}\sqrt{V_t^{N,(1)}}\,\text{d}W_t^{\perp},\\
	  			\text{d}V_{t,i}^{(1)}=-x_i(V_{t,i}^{(1)}-v_0^i)\,\text{d}t+(\theta-\lambda V_t^{N,(1)})\,\text{d}t,\,i=1,2,...,N,
	  		\end{cases}
	  	\end{equation}
	  	\begin{equation}\label{eq:lifted_heston_split_second_sde}
	  		\begin{cases}
	  			\text{d}X^{(2)}_t=\rho\sqrt{V_t^{N,(2)}}\,\text{d}W_t,\\
	  			\text{d}V^{(2)}_{t,i}=\nu\sqrt{V_t^{N,(2)}}\,\text{d}W_t,\,i=1,2,...,N,
	  		\end{cases}
	  	\end{equation}
	  \end{subequations}
	  where $V_t^{N,(1)}=\sum_{i=1}^Nw_iV_{t,i}^{(1)}$ and $V_t^{N,(2)}=\sum_{i=1}^Nw_iV_{t,i}^{(2)}$.
	  
	  We define $\mathbf{W}=[w_1,...w_i,...w_N]$ as the row vector whose components are the weights in $V_t^N$. The solution $V^{(1)}_{t,i}$ of the first set of SDEs in Eq.~\eqref{eq:lifted_heston_split_first_sde} can be found by the solution of following system of ordinary differential equations:
	  \begin{equation}\label{eq:V_i_first_matrix}
\frac{\mathrm{d}}{\mathrm{d}t}\mathbf{V}_t^{(1)}
+A\mathbf{V}_t^{(1)}
=\mathbf{b},
\end{equation}
where $\mathbf{V}_t^{(1)}=[V_{t,1}^{(1)},...V_{t,2}^{(1)},...V_{t,N}^{(1)}]^T$, $A=\lambda \mathbf{1}^T\mathbf{W}+\text{diag}(X)$ and $\mathbf{b}=\mathbf{1}^T\theta+\text{diag}(X)\mathbf{v_0}$. Here, $\mathbf{1}$ is the row vector of length $N$ and all elements are 1, and $\text{diag}(X)$ is the diagonal matrix with diagonal elements $\text{diag}(X)(i,i)=x_i,\,i=1,2,...,N,$ and $\mathbf{v}_0=[v^1_0,...v^i_0,...v^N_0]^T.$ For the solution of the second SDE in Eq.~\eqref{eq:lifted_heston_split_second_sde}, we have 
	  \begin{equation*}
	  	\text{d}V^{N,(2)}_t=\nu(\sum_{i=1}^Nw_i)\sqrt{V^{N,(2)}_t}\,\text{d}W_t,
	  \end{equation*}
	  which reveals that $V^{N,(2)}_t$ is a CIR process. As noted from Eq.~\eqref{eq:lifted_heston_split_second_sde}, all $V^{(2)}_{t,i}$ share the same Brownian motion $W_t$. We write $V^{(2)}_{t+h,i}=V^{(2)}_{t,i}+H,\,i=1,2,...,N,$. We then have $$V^{N,(2)}_{t+h}=\sum_{i=1}^N w_iV^{(2)}_{t+h,i}=\sum_{i=1}^N w_i V^{(2)}_{t,i}+(\sum_{i=1}^N w_i)H=V^{N,(2)}_t+(\sum_{i=1}^N w_i)H,$$ so we obtain
	  \begin{equation}\label{eq:V_t_i}
	  	V^{(2)}_{t+h,i}=V^{(2)}_{t,i}+H=V^{(2)}_{t,i}+\frac{V^{N,(2)}_{t+h}-V^{N,(2)}_{t}}{\sum_{i=1}^Nw_i}.
	  \end{equation}
	  Hence, the solutions of the two parts of the operator splitting procedures of the lifted Heston model are given by
	  \begin{subequations}
	  	\begin{equation}\label{eq:lifted_heston_split_first_sol}
	  		\begin{cases}
	  			X^{(1)}_{t+h}=X^{(1)}_t-\frac12 I_t^{h,1}+N\left(0,(1-\rho^2)I_t^{h,1}\right),\\
	  			\mathbf{V}^{(1)}_{t+h}=e^{-Ah}\mathbf{V}^{(1)}_{t}+A^{-1}(I-e^{-Ah})\mathbf{b},\\
	  			V^{N,(1)}_{t+h}=\mathbf{W}\mathbf{V}^{(1)}_{t+h},\\
	  			I_t^{h,1}=\int_t^{t+h} V^{N,(1)}_s\,\text{d}s=\mathbf{W}A^{-1}\mathbf{b}h+\mathbf{W}A^{-1}\left[(I-e^{-Ah})\mathbf{V}_t^{(1)}-A^{-1}(I-e^{-Ah})\mathbf{b}\right],\\
	  		\end{cases} 
	  	\end{equation}
	  	\begin{equation}\label{eq:lifted_heston_split_second_sol}
	  		\begin{cases}
	  			X^{(2)}_{t+h}=X^{(2)}_t+\frac{\rho}{\nu\sum_{i=1}^N w_i}(V^{N,(2)}_{t+h}-V^{N,(2)}_{t}),\\
	  			V^{N,(2)}_{t+h}=CIR(0,0,\nu\sum_{i=1}^N w_i,h,V^{N,(2)}_{t}),\\
	  			V^{(2)}_{t+h,i}=V^{(2)}_{t,i}+\frac{V^{N,(2)}_{t+h}-V^{N,(2)}_{t}}{\sum_{i=1}^Nw_i}, i = 1,2,...,N.
	  		\end{cases} 
	  	\end{equation}
	  \end{subequations}
	  where $I$ is the $N\times N$ identity matrix.
	  We summarize the second-order operator splitting scheme of the lifted Heston model in Proposition~\ref{prop:lifted_hes}.
	  \begin{proposition}\label{prop:lifted_hes}
	  	The second-order operator splitting scheme of the lifted Heston model, given $X_t$ and $\mathbf{V}_t=[V_{t,1},...V_{t,i},...V_{t,N}]^T$, is given by
	  	\begin{equation}\label{eq:sol_split_lifted_heston}
	  		\begin{cases}
	  			X^{\text{split}}_{t+h}=X_t-\frac{I_t^h}2+N\left(0,(1-\rho^2)I_t^h\right)+\frac{\rho}{\nu\sum_{i=1}^N w_i}(\hat{V}^{N,(2)}_{t+h}-\tilde{V}^{N,(1)}_{t+h/2}),\\
	  			\mathbf{V}^{\text{split}}_{t+h}=e^{-Ah/2}\mathbf{\hat{V}}^{(2)}_{t+h}+A^{-1}(I-e^{-Ah/2})\mathbf{b},\\
	  			I_t^h=\mathbf{W}A^{-1}\mathbf{b}h+\mathbf{W}A^{-1}\left[(I-e^{-Ah/2})(\mathbf{V}_t+\mathbf{\hat{V}}^{(2)}_{t+h})-2A^{-1}(I-e^{-Ah/2})\mathbf{b}\right],\\
	  			V^{N,\text{split}}_{t+h}= \mathbf{W}\mathbf{V}^{\text{split}}_{t+h},
	  		\end{cases}
	  	\end{equation}
	  	where $\mathbf{\tilde{V}}^{(1)}_{t+h/2}=e^{-Ah/2}\mathbf{V}_t+A^{-1}(I-e^{-Ah/2})\mathbf{b}$,\, $\tilde{V}^{N,(1)}_{t+h/2}=\vec{W}\mathbf{\tilde{V}}^{(1)}_{t+h/2},$\\ $\hat{V}^{N,(2)}_{t+h}=CIR(0,0,\nu\sum_{i=1}^N w_i,h,\tilde{V}^{N,(1)}_{t+h/2})$,\\ $\mathbf{\hat{V}}^{(2)}_{t+h}=\mathbf{\tilde{V}}^{(1)}_{t+h/2}+\frac{\hat{V}^{N,(2)}_{t+h}-\tilde{V}^{N,(1)}_{t+h/2}}{\sum_{i=1}^N w_i}$,  and $\mathbf{V}^{\text{split}}_{t+h}=[V_{t+h,1}^{\text{split}},V_{t+h,2}^{\text{split}},...V_{t+h,N}^{\text{split}}]^T$.
	  \end{proposition}
	   \begin{proof}
	  	Once Eqs.~(\ref{eq:lifted_heston_split_first_sol},\,\ref{eq:lifted_heston_split_second_sol}) are established, we follow the same procedure as illustrated by the Heston model in \ref{ap:proof_heston}.
	  \end{proof}
	  \begin{remark}\label{re:lifted_heston}
	  	Since the second part in Eq.~\eqref{eq:lifted_heston_split_second_sol} is dictated to simulate $V^{N,(2)}$ from a CIR process, 
		this would guarantee non-negativity of $V^{N,(2)}$ in the above second-order operator splitting 
		scheme.
	  	We can establish non-negativity of $V_{t+h}^{N,(1)}$ in 
		Eq.~\eqref{eq:lifted_heston_split_first_sol} under the assumptions 
		of $w_i\geq0$, $x_i\geq0$ and $v_0^i\geq0$ for $i=1,2,...N$. 
		To prove the claim, it suffices to show $$e^{-Ah}\mathbf{V}^{(1)}_{t}+A^{-1}(I-e^{-Ah})\mathbf{b}\geq0.$$ 
		Observing $A^{-1}e^{Ah}=e^{Ah}A^{-1}$, it is equivalent to show
	  	$$\mathbf{V}^{(1)}_t\geq(I-e^{Ah})A^{-1}\mathbf{b}.$$ 
		Note that $e^{Ah}=\sum_{i=0}^\infty \frac{(Ah)^i}{i!}$, so 
		$$(I-e^{Ah})A^{-1}\mathbf{b}=-\sum_{i=1}^\infty \frac{h^i}{i!} A^{i-1} \mathbf{b}.$$ 
		As we assume $w_i\geq0$, $x_i\geq0$ and $v_0^i\geq0$, so all elements 
		in $A$ are non-negative and $\mathbf{b}\geq0$, therefore $(I-e^{Ah})A^{-1}\mathbf{b}\leq0$. 
		Note that the initial input value of $\mathbf{V}^{(1)}_t$ is non-negative, 
		hence $\mathbf{V}^{(1)}_t\geq(I-e^{Ah})A^{-1}\mathbf{b}$ holds. This would mean $\mathbf{V}_{t+h}^{(1)}$ computed by the first part in Eq.~\eqref{eq:lifted_heston_split_first_sol} observes non-negativity under the specified assumptions.
%
	  \end{remark}
	  \begin{remark}\label{re:lifted_hes_bayer}
	   Though \citet{bayer2024efficient} follow a similar approach of splitting of operators 
	   in constructing the simulation schemes for $V^N$ and $V^{(i)}$ in their second 
	   order operator splitting method, they choose to use the randomized 
	   leapfrog splitting method to simulate the asset price $S_t=e^{X_t}$. This requires 
	   one additional uniformly distributed random variable. Instead of employing exact simulation 
	   of the CIR process, they adopt the algorithm of \citet{lileika2021second} for 
	   simulating the CIR process, which approximates the CIR process as a three-valued 
	   discrete random variable. Our second-order operator splitting scheme can be implemented at better ease. 
	  \end{remark}
	  \section{Second-order operator splitting scheme for the Barndorff–Nielsen and Shephard model}\label{s:hull_white_bns_split}\noindent
	
	  \noindent Simulation schemes for option pricing under the BN-S model are relatively less explored in literature. 
	  In this section, we derive the second-order operator splitting schemes for pricing
	  VIX derivatives and European path dependent options under the BN-S model with various specifications of 
	  the Background Driving Lévy Process (BDLP). 
	  \subsection{Model dynamics}
	  \noindent Let $S_t$ denote the asset price at time $t$ and $X_t = \ln S_t$ be the log-asset price. Under a risk-neutral measure $\mathbb{Q}$, the dynamics of the BN-S model are governed by the following system of stochastic differential equations:
\begin{subequations}
\label{eq:bns_system} 
\begin{align}
    \mathrm{d}X_t &= \left[ r - \lambda \kappa(\rho) - \frac{1}{2}V_t \right] \mathrm{d}t + \sqrt{V_t} \, \mathrm{d}W_t + \rho \, \mathrm{d}Z_{\lambda t}, \label{eq:price_dynamics} \\
    \mathrm{d}V_t &= -\lambda V_t \, \mathrm{d}t + \mathrm{d}Z_{\lambda t}, \label{eq:vol_dynamics}
\end{align}
\end{subequations}

\noindent where $r$ represents the risk-free interest rate, $\lambda > 0$ is the 
mean reversion parameter governing the persistence of volatility shocks, 
$W_t$ is a standard Brownian motion, $Z = \{Z_{\lambda t}\}_{t \ge 0}$ 
is the time-scaled BDLP, which is a subordinator 
(Lévy process with non-negative increments and no drift). Usually, when we refer to $Z_t$, we imply that the
 increments of $Z_t$ follow a specific infinitely divisible distribution, such as the Gamma or Tempered Stable (TS) distribution.
When it takes the form $Z_{\lambda t}$, it refers to the time-scaled version of the BDLP, where the time parameter is scaled by the mean reversion rate $\lambda$. 
The term represents the random jumps injected into the variance process, 
$\rho \le 0$ is the leverage parameter. It captures the "leverage effect" (negative correlation between returns and volatility) by allowing the volatility jumps $\text{d}Z_{\lambda t}$ to impact the price process $X_t$ simultaneously. Lastly, $\kappa(\rho)$ is the cumulant transform of the BDLP. 
Specifically, we have  
\begin{equation}
\kappa(\rho) = \int_{\mathbb{R}^+} (e^{\rho x} - 1) \, \nu(\text{d}x),
\end{equation} 
which serves as a martingale correction term to ensure that the discounted asset price is a local martingale.

The variance process $V_t$ in Eq.~\eqref{eq:vol_dynamics} is a stationary process with 
its marginal distribution determined by the choice of the BDLP. The solution to the stochastic 
differential equations can be written in an integral form as follows:
\begin{subequations}
\label{eq:exact_simulation_step} 
\begin{align}
    X_{t+h} &= X_t + [r  - \lambda \kappa(\rho)]h - \frac{I_t^h}{2} + \rho \Delta Z_{\lambda t} + \sqrt{I_t^h}\epsilon, \label{eq:log_price_dynamics} \\
    V_{t+h} &= e^{-\lambda h} V_t + \int_t^{t+h} e^{-\lambda(t+h-s)} \, \mathrm{d}Z_{\lambda s}. \label{eq:vol_integral}
\end{align}
\end{subequations}
Here, $I_t^h = \int_t^{t+h} V_s \, \mathrm{d}s$ denotes the integrated variance over the interval $[t, t+h]$, $\Delta Z_{\lambda t} = Z_{\lambda(t+h)} - Z_{\lambda t}$ is the increment of the time-scaled BDLP over the same interval, and $\epsilon \sim N(0,1)$ is a standard normal random variable independent of $Z$.
This explicit structure is the key to our second-order operator splitting scheme, which allows for the exact simulation of the stochastic volatility process when the integral term can be sampled efficiently.


The flexibility of the BN-S model allows us to specify the model dynamics 
in two distinct classes: D-OU class and OU-D class, where D can be Gamma or Tempered Stable (TS).  
In the D-OU framework, the marginal density $u(x)$ of the variance process $V_t$
is prespecified. Crucially, 
the existence of such model is guaranteed by the self-decomposability property, 
which explicitly links the marginal density $u(x)$ to the Lévy density $w(x)$ of 
the driving BDLP via the differential relation:
\begin{equation}\label{eq:levy_tran}
	w(x) = -u(x) - x u'(x).
\end{equation}
Conversely, the OU-D framework directly specifies the BDLP increments, 
allowing the variance distribution evolving consequently.

In our study, we focus exclusively on the following four specific BN-S model 
specifications constructed by setting the distribution $D$ to be either Gamma or 
TS (namely, Gamma-OU, TS-OU, OU-Gamma, and OU-TS). 
We explicitly exclude the Normal Tempered Stable (NTS) process from our analysis. 
This is because the NTS process allows for negative jumps, which violates the 
strict non-negativity requirement of the variance process.

\subsection{Exact simulation schemes}\noindent
To contextualize the efficiency of our proposed second-order operator splitting scheme,  
we first present the exact simulation schemes required for various BN-S specifications. 
The exact simulation of the variance process $V_{t+h}$ relies on the explicit solution shown in Eq.~\eqref{eq:vol_integral}:
\begin{equation}
    V_{t+h} = e^{-\lambda h} V_t + U_h, \quad \text{where} \quad U_h = \int_t^{t+h} e^{-\lambda(t+h-s)} \, \mathrm{d}Z_{\lambda s}.
\end{equation}
\noindent
The simulation mechanisms for the stochastic innovation $U_h$ 
(also referred as the $a$-remainder in the context of self-decomposability) 
generally fall into two distinct theoretical frameworks proposed in the literature. They 
include the MGF factorization approach pioneered by \cite{Dassio2019Gamma} and \cite{Dassio2021TS}, 
and the Lévy measure decomposition approach developed by \cite{sabino2020gamma} and \cite{sabino2020tempered}.

\cite{Dassio2019Gamma} approach the simulation problem by analytically deriving 
the MGF or the Laplace transform of the integrated process 
$U_h$. The core idea is to manipulate the resulting algebraic expression of 
the MGF, factoring it into a product of recognizable transforms, like
\begin{equation*}
    \mathbb{E}[e^{-u U_h}] = \phi_{X}(u) \times \phi_{Y}(u),
\end{equation*}
where $\phi_{X}$ and $\phi_{Y}$ correspond to the transforms 
of some standard distributions.
 This method essentially recognizes the innovation as the 
 sum of independent random variables. 
In contrast, Sabino and his co-authors propose a direct probabilistic decomposition based on 
the properties of self-decomposable distributions and their associated 
$a$-remainders. Instead of working in the Laplace domain, they analyze the 
Lévy measure $\nu_a(\text{d}x)$ of the innovation $U_h$ directly. Their strategy involves 
splitting $\nu_a(\text{d}x)$ into a singular component (capturing small jumps) and a 
regular component with finite mass (capturing large jumps). 

Taking the TS-OU process as a prime example, \cite{sabino2020gamma} illustrate that the 
Lévy measure of the $a$-remainder can be exactly partitioned 
as $\nu_a = \nu_{1} + \nu_{2}$. Here, $\nu_{1}$ corresponds to a rescaled
 TS distribution. Despite originating 
 from an infinite activity process, $\nu_2$ is proven to be a finite 
 integral. Consequently, the second component is identified as a 
 Compound Poisson process, allowing $U_h$ to be simulated exactly via a 
 sum of a TS variable and a Compound Poisson variable.

 We summarize the exact simulation schemes for the four types of BN-S model specifications in Table \ref{tab:exact_methods_formulas}. 
 It is important to note that though the scheme in \citeauthor{sabino2020gamma} (\citeyear{sabino2020gamma}, \citeyear{sabino2020tempered}) avoids the acceptance-rejection sampling inherent 
	in the approach in \citet{Dassio2019Gamma, Dassio2021TS}, our numerical experiments reveal that it is highly sensitive to 
	discretization bias, necessitating the choice of small time step $h$. As a result, its efficiency gains are undermined. The hierarchical dependency 
	(say, jump sizes $S_k$ conditioning on the auxiliary variable $V$) leads to high computational cost of the exact simulation schemes when compared to our proposed second-order operator splitting scheme.

\begin{table}[htbp]
    \centering
    \renewcommand{\arraystretch}{1.9} 
    
    \resizebox{\textwidth}{!}{%
    \begin{tabular}{@{}c l l l l@{}}
        \toprule
        \textbf{Class} & \textbf{Model} & \textbf{Method} & \textbf{Decomposition ($U_h$)} & \textbf{Component Specifications} \\ \midrule
        
        \multirow{3}{*}{\textbf{OU-D}} 
        & \multirow{2}{*}{\textbf{OU-Gamma}} 
        & \cite{James2013}
        & $\displaystyle e^{-\lambda h}G_1 + G_2$ 
        & \makecell[l]{$G_1, G_2 \sim \text{Gamma}(\cdot, \beta)$ \\ $I_t^h \sim \text{Dirichlet Mean}$ (Double CFTP)} \\ 
        
        & 
        & \cite{Dassio2019Gamma}
        & $\displaystyle \tilde{\Gamma} + \sum_{k=1}^{\tilde{N}} S_k$ 
        & \makecell[l]{$\tilde{\Gamma} \sim \Gamma(\alpha h, \beta e^{\lambda h}), \ \tilde{N} \sim \text{Pois}(\frac{1}{2}\alpha \lambda h^2)$ \\ $S_k \sim \text{Exp}(\beta e^{\lambda h \sqrt{U}})$} \\ \cmidrule(l){2-5} 
        
        & \textbf{OU-TS} 
        & \cite{Dassio2021TS} and \cite{sabino2020tempered}
        & $\displaystyle \tilde{Y} + \sum_{k=1}^{N} S_k$ 
        & \makecell[l]{$\tilde{Y} \sim \mathcal{TS}\left(\alpha, \frac{\beta}{a}, \frac{\theta}{\alpha\lambda}(1-a^\alpha)\right), \ a = e^{-\lambda h}$ \\ 
                       $N \sim \text{Pois}\left(\frac{\theta}{\alpha\lambda}\beta^\alpha\Gamma(1-\alpha) D_a\right), \ D_a = \frac{1 - a^\alpha + a^\alpha \log a^\alpha}{\alpha a^\alpha}$ \\ 
                       $S_k|V \sim \Gamma(1-\alpha, \beta V)$ \\ 
                       $V \sim f_V$ via Rejection Sampling \\ 
                       $f_V(v) = \frac{1}{D_a} \frac{v^\alpha - 1}{v} \qquad 1 \leq v \leq 1/a$} \\ \midrule \midrule
        
        \multirow{5}{*}{\textbf{D-OU}} 
        & \multirow{3}{*}{\textbf{Gamma-OU}} 
        & \cite{Dassio2019Gamma}
        & $\displaystyle \sum_{k=1}^{N} S_k$ 
        & \makecell[l]{$N \sim \text{Pois}(\lambda \alpha h)$ \\ $S_k \sim \text{Exp}(\beta e^{\lambda h U}), \ U \sim \mathcal{U}[0,1]$} \\ 
        
        & 
        & \cite{sabino2020gamma}
        & $\displaystyle G_h$ 
        & \makecell[l]{$N \sim \text{Polya}(\alpha, 1-e^{-\lambda h})$ \\ $G_h | N \sim \text{Erlang}(N, \beta e^{\lambda h})$} \\ 
        
        & 
        & {Common Method in} \cite{sabino2020gamma} 
        & $\displaystyle \sum_{k=1}^{N} e^{-(h-\tau_k)} Y_k$ 
        & \makecell[l]{$N \sim \text{Pois}(\lambda \alpha h)$ \\ $\tau_k \sim \mathcal{U}[0,h], \ Y_k \sim \text{Exp}(\beta)$} \\ \cmidrule(l){2-5} 
        
        & \multirow{2}{*}{\textbf{TS-OU}} 
        & \cite{Dassio2021TS} 
        & $\displaystyle \tilde{Y} + \sum_{k=1}^{N_1} S_k + \sum_{j=1}^{N_2} \tilde{S}_j$ 
        & \makecell[l]{$\tilde{Y} \sim \mathcal{TS}\left(\alpha, \frac{\beta}{a}, \theta(1-a^\alpha)\right)$ \\ 
                       $N_1 \sim \text{Pois}\left(\theta\beta^\alpha\Gamma(1-\alpha)D_a\right)$ \\ 
                       $N_2 \sim \text{Pois}\left(\theta\beta^\alpha\Gamma(1-\alpha)\ln(1/a)\right)$ \\ 
                       $S_k|V \sim \Gamma(1-\alpha, \beta V), \ V \sim f_V \text{ (see OU-TS)}$ \\ 
                       $\tilde{S}_j|V^* \sim \Gamma(1-\alpha, \beta V^*), \ V^* = e^{\lambda h U}, \ U \sim \mathcal{U}[0,1]$} \\ 
        
        & 
        & \cite{sabino2020tempered} 
        & $\displaystyle \tilde{Y} + \sum_{k=1}^{N} S_k$ 
        & \makecell[l]{$\tilde{Y} \sim \mathcal{TS}(\alpha, \beta, c(1-a^\alpha)), \ a = e^{-\lambda h}$ \\ 
                       $N \sim \text{Pois}(\Lambda_a), \Lambda_a = \theta \Gamma(1-\alpha)\frac{\beta^\alpha}{\alpha}(1-a^\alpha)$ \\ 
                       $S_k|V_k \sim \Gamma(1-\alpha, \beta V_k)$ \\ 
                       $V_k = (1+\frac{e^{\alpha \lambda h}-1}{\alpha}U_k)^{1/\alpha}, \ U_k \sim \mathcal{U}[0,1]$} \\
        \bottomrule
    \end{tabular}%
    }
    \caption{Summary of various exact simulation schemes for the four BN-S model specifications.}
    \label{tab:exact_methods_formulas}
\end{table}

\subsection{Second-order operator splitting scheme}


\noindent Unlike pricing VIX options, which depend solely on the future state of the instantaneous volatility,
 pricing of path dependent equity options requires simulation of the underlying asset 
price $S_t$, say discrete monitoring 
value $S_{t_m}$ at time point $t_m$. The asset price trajectory is generated by 
simulating the log-asset price dynamics $X_t = \ln S_t$ recursively. To update the asset price $X_{t+h}$ over a time step, one requires 
the joint realization of the spot volatility $V_{t+h}$ (propagating the variance process) and the 
integrated volatility $I_t^h$ (computing the current log-asset price increment). This highlights the fundamental 
challenge in the joint simulation of $X_t$ and $V_t$ .

The methods proposed by \cite{Dassio2019Gamma} and \cite{sabino2020gamma} 
focus on generating the exact transition of the spot volatility 
from $V_t$ to $V_{t+h}$. While they manage to sample the terminal value 
$V_{t+h}$ from its exact distribution, they {fail to provide the exact 
conditional distribution of the integrated variance $I_t^h$} given $V_{t+h}$. 
As a result, the exact simulation method of \cite{Dassio2019Gamma} and \cite{sabino2020gamma} cannot be 
used for pricing standard European options, let alone options with path dependent features.

\cite{James2013} propose a specialized exact simulation algorithm for the OU-Gamma BN-S model 
that addresses the joint simulation of $\Delta Z_{\lambda t}$ together with $\int_t^{t+h} e^{-\lambda(t+h-s)} dZ_{\lambda s}$ 
by transforming the stochastic path integration problem into a distributional sampling problem. 
Specifically, they explicitly decompose the joint distribution into independent 
components comprising the Gamma random variable and Dirichlet mean random variable. 
However, simulating the Dirichlet mean component is non-trivial, which necessitates the implementation of 
the so called Double Coupling From The Past algorithm. This sophisticated perfect simulation technique 
involves the complex retrospective coupling and rejection steps, leading to high 
computational costs and significant implementation complexity. 
To circumvent the computational bottleneck of such exact joint simulation, 
the core idea of our second-order operator splitting scheme is to decompose the infinitesimal generator of the 
BN-S dynamics into two analytically solvable 
sub-operators: {continuous part} and {jump part}.

\par\medskip

\noindent{\textit{Sub-operator 1: Continuous dynamics}}
\mbox{}\\
In this operator splitting step, we isolate the continuous components.
Specifically, the jump term $dZ_{\lambda t}$ is removed from both of the dynamic equations of $X_t$ and $V_t$, rendering 
the variance process deterministic:

\begin{equation}
\label{eq:split_step1} 
\begin{aligned}
    \mathrm{d}X_t^{(1)} &= \left[ r - \lambda \kappa(\rho) - \frac{1}{2}V_t^{(1)} \right] \mathrm{d}t + \sqrt{V_t^{(1)}} \, \mathrm{d}W_t, \\
    \mathrm{d}V_t^{(1)} &= -\lambda V_t^{(1)} \, \mathrm{d}t.
\end{aligned}
\end{equation}
Since $V_t^{(1)}$ follows the deterministic decay at the rate $\lambda$, its solution is 
$V_{t+h}^{(1)} = e^{-\lambda h}V_t^{(1)}$. The integrated variance ${I}_t^h$
over one time step can be found {explicitly} as follows: 
\begin{equation*}
    {I}_t^h = \int_t^{t+h} V_s^{(1)} \text{d}s = V_t^{(1)} \left( \frac{1 - e^{-\lambda h}}{\lambda} \right).
\end{equation*}
Subsequently, the solution to Eq.~\eqref{eq:split_step1} is explicitly given by
\begin{equation*}
    \begin{cases}
        X_{t+h}^{(1)} = X_t^{(1)} + [r - \lambda \kappa(\rho)]h - \frac{1}{2}{I}_t^h + \sqrt{{I}_t^h} \epsilon, \\
        V_{t+h}^{(1)} = e^{-\lambda h} V_t^{(1)},
    \end{cases}
\end{equation*}
where $\epsilon \sim N(0,1)$. The first step of operator splitting effectively solves for the 
integrated variance explicitly, 
avoiding the discretization bias associated with stochastic integrals.
\par\medskip
\noindent{\textit{{Sub-operator 2: Pure jump dynamics}}}
\mbox{}\\
The second step of operator splitting isolates the contribution of the Lévy jumps to both 
the log-asset price and volatility:
\begin{equation}
\label{eq:split_step2} 
\begin{aligned}
    \mathrm{d}X_t^{(2)} &= \rho \, \mathrm{d}Z_{\lambda t}, \\
    \mathrm{d}V_t^{(2)} &= \mathrm{d}Z_{\lambda t}.
\end{aligned}
\end{equation}
The solution is simply the addition of the respective Lévy increment:
\begin{equation*}
    \begin{cases}
        X_{t+h}^{(2)} = X_t^{(2)} + \rho \, \Delta Z_{\lambda t}^h, \\
        V_{t+h}^{(2)} = V_t^{(2)} + \Delta Z_{\lambda t}^h,
    \end{cases}
\end{equation*}
where $\Delta Z_{\lambda t}^h = Z_{\lambda(t+h)} - Z_{\lambda t}$ is simulated 
directly according to the specific BDLP 
distribution.

\par\medskip
\noindent{\textit{{Composition of operator splitting}}}
\mbox{}\\
Next, we combine the solutions of the two sub-operators. 
Let $\Phi_h^{(1)}$ and $\Phi_h^{(2)}$ denote the solution operators 
for Sub-operator 1 and Sub-operator 2 over a time step $h$, respectively. 
Following the Strang second-order operator splitting procedure as exemplified in Eq.~\eqref{eq:split}, 
the full step from $t$ to $t+h$ is approximated by
\begin{equation}
    (X_{t+h}, V_{t+h}) \approx \Phi_{h/2}^{(1)} \circ \Phi_h^{(2)} \circ \Phi_{h/2}^{(1)} (X_t, V_t).
\end{equation}

To illustrate this composition explicitly, we unravel the splitting procedure by formalizing the updates of the state variables step-by-step. 
By following the two-step splitting process for the variance process $V_t$, the continuous exponential decay factor is applied to the jump term to give
$$ V_{t+h} \approx \Big( {V_t e^{-\lambda h/2}} + {\Delta Z_{\lambda t}^h} \Big) {e^{-\lambda h/2}} = V_t e^{-\lambda h} + e^{-\lambda h/2}\Delta Z_{\lambda t}^h. $$

\noindent Consequently, the total integrated variance $I_t^h = I_1 + I_2$ over the full time step is derived by substituting the intermediate variance from Step 2 into the second integral $I_2$, where $I_1 = V_t \left( \frac{1 - e^{-\lambda h/2}}{\lambda} \right)$ and $I_2 = \left( V_t e^{-\lambda h/2} + \Delta Z_{\lambda t}^h \right) \left( \frac{1 - e^{-\lambda h/2}}{\lambda} \right)$ represent the integrated variance obtained in the first and second half-steps, respectively. This yields
$$ I_t^h = {\frac{1-e^{-\lambda h/2}}{\lambda}V_t} + {\frac{1-e^{-\lambda h/2}}{\lambda}\Big( V_t e^{-\lambda h/2} + \Delta Z_{\lambda t}^h \Big)} = \frac{1-e^{-\lambda h}}{\lambda}V_t + \frac{1-e^{-\lambda h/2}}{\lambda}\Delta Z_{\lambda t}^h. $$

\noindent Let $\mu = r - \lambda\kappa(\rho)$ denote the constant drift part. For the log-asset price process $X_t$, the total update is simply the sum of increments from the continuous and jump sub-operators:
$$ X_{t+h} \approx X_t + {\left[ \mu \frac{h}{2} - \frac{1}{2}I_1 + N(0, I_1) \right]} + {\rho \Delta Z_{\lambda t}^h} + {\left[ \mu \frac{h}{2} - \frac{1}{2}I_2 + N(0, I_2) \right]}. $$

\noindent This specific sequence of composition reveals that the jump innovation 
$\Delta Z_{\lambda t}^h$ is effectively incorporated at 
the midpoint of the exponential decay path of volatility. 
The resulting second-order operator splitting scheme is summarized as follows:
\begin{equation}
\begin{cases}
X_{t+h}^{\text{split}} = X_t + \mu h - \frac{1}{2} I_t^h + N(0, I_t^h) + \rho \Delta {Z}_{\lambda t}^h, \\
V_{t+h}^{\text{split}} = e^{-\lambda h} V_t + e^{-\lambda h / 2} \Delta {Z}_{\lambda t}^h,\\
I_t^h = \frac{1 - e^{-\lambda h}}{\lambda} V_t + \frac{1 - e^{-\lambda h / 2}}{\lambda} \Delta {Z}_{\lambda t}^h.
\end{cases}
\label{eq:bns_split_final}
\end{equation}


Comparing to the exact integral representation 
$V_{t+h} = e^{-\lambda h} V_t + \int_t^{t+h} e^{-\lambda(t+h-s)} \, \mathrm{d} Z_{\lambda s}$, we observe that the second-order operator splitting 
scheme implicitly approximates the stochastic integral 
using the {midpoint rule} in numerical quadrature. 
Specifically, evaluating the deterministic integrand $e^{-\lambda(t+h-s)}$ at the midpoint of the time interval $s = t + h/2$ 
yields the constant $e^{-\lambda h/2}$. Multiplying this midpoint evaluation by the total jump increment 
$\int_t^{t+h} \mathrm{d} Z_{\lambda s} = \Delta Z_{\lambda t}^h$ exactly recovers the term $e^{-\lambda h/2}\Delta Z_{\lambda t}^h$ 
as appearing in our second-order operator splitting scheme.

Since the pricing of VIX options relies critically on the accurate evaluation of the expected future variance $\mathbb{E}[V_T]$, it would be desirable to quantify the numerical bias introduced by the above midpoint approximation of the volatility process $V_t$. We formalize the analysis of bias in \Cref{prop:second_order_convergence}.



\begin{proposition}\label{prop:second_order_convergence}
Let $V_t$ be an Ornstein-Uhlenbeck (OU) type process driven by the Lévy process $Z_t$, as governed by the stochastic differential equation:
\begin{equation*}
    \mathrm{d}V_t = -\lambda V_t \, \mathrm{d}t + \mathrm{d}Z_t,
\end{equation*}
where $\lambda > 0$ is the mean reversion speed. Let $\bar{V}_{T}$ denote the approximation of $V_T$ at the terminal time $T = N h$ obtained using the second-order operator splitting scheme. The bias of the expected value satisfies
\begin{equation}
    \left| \mathbb{E}[V_{T}] - \mathbb{E}[\bar{V}_{T}] \right| = \mathcal{O}(h^2).
\end{equation}
\end{proposition}
We refer readers to \ref{appendix: proof_second_order_vix} for the proof of \Cref{prop:second_order_convergence}. 

\subsection{Fourier-cosine method}

\noindent Pricing VIX options serves as an excellent test case to assess accuracy of the simulation of the variance process. Unlike the standard equity options whose payoffs depend on the log-asset price process $X$, the VIX options are derivatives written exclusively on the variance process $V$. By performing Monte Carlo simulation of the variance paths using our second-order operator splitting scheme, we can numerically evaluate these VIX option prices. Since our second-order operator splitting scheme is globally second-order consistent for the joint infinitesimal generator of $(X, V)$, the second-order operator splitting scheme would be highly effective for simulating the VIX dynamics as well.

To perform an assessment on numerical accuracy of our second-order operator splitting scheme, an accurate benchmark VIX option price is required. Given the affine structure of the BN-S model, the characteristic function of $\text{VIX}^2$ can be derived in closed form. 
The Fourier-cosine method of \cite{Fang2009fourierCos} is renowned for its exponential convergence rate and numerical 
robustness when pricing European derivatives with known characteristic 
functions, making it a good choice to provide the benchmark VIX
option prices for numerical comparison.

The affine structure of $\text{VIX}^2$ with respect to the spot variance $V_T$ under the BN-S framework provides the mathematical tractability required to obtain an analytical characteristic function, which is the fundamental prerequisite for applying the Fourier-cosine pricing technique. 
The affine structure and conditional characteristic function of $\text{VIX}^2$ are summarized in \Cref{prop:vix_affine_cf}.


\begin{proposition}
    \label{prop:vix_affine_cf}
    Let $\tau_{vix}$ be the time horizon of the VIX. The squared VIX admits an affine structure with respect to the spot variance $V_T$:
    \begin{equation}
        \mathrm{VIX}_T^2 = \mathcal{A}(\tau_{vix}) V_T + \mathcal{B}(\tau_{vix}),
    \end{equation}
    where the coefficient functions are given by:
	\begin{subequations}
\label{eq:vix_coefficients} 
\begin{align}
    \mathcal{A}(\tau_{vix}) &= \frac{1-e^{-\lambda \tau_{vix}}}{\lambda \tau_{vix}}, \label{eq:vix_coeff_A} \\
    \mathcal{B}(\tau_{vix}) &= \frac{1}{\lambda} [1 - \mathcal{A}(\tau_{vix})] \int_0^\infty x \, \nu(\mathrm{d}x) - 2 \int_0^\infty (1 + \rho x - e^{\rho x}) \, \nu(\mathrm{d}x). \label{eq:vix_coeff_B}
\end{align}
\end{subequations}
    Consequently, the conditional characteristic function of $\mathrm{VIX}_T^2$, denoted by $\Phi(\phi; V_t, \tau) = \mathbb{E}^{\mathbb{Q}}\left[ e^{i\phi \mathrm{VIX}_T^2} \mid V_t \right]$, preserves the exponential affine form:
    \begin{equation}\label{eq:exp_chf}
        \Phi(\phi; V_t, \tau) = \exp\left( A(\phi; \tau)V_t + C(\phi; \tau) \right),
    \end{equation}
    where the coefficient functions $A(\phi; \tau)$ and $C(\phi; \tau)$ are explicitly given by
\begin{subequations}\label{eq:cf_coefficients}
\begin{align}
A(\phi;\tau)
&=
i\phi\,\mathcal{A}(\tau_{\mathrm{vix}})
e^{-\lambda\tau},
\label{eq:cf_coefficient_A}
\\
C(\phi;\tau)
&=
i\phi\,\mathcal{B}(\tau_{\mathrm{vix}})
+
\int_0^\tau
\int_0^\infty
\left(
e^{A(\phi;s)y}-1
\right)
\nu(\mathrm{d}y)\,\mathrm{d}s.
\label{eq:cf_coefficient_C}
\end{align}
\end{subequations}
    Here, $\nu(\text{d}y)$ denotes the Lévy measure of the BDLP.
\end{proposition}

We refer readers to \ref{appendix:proof_vix_affine_cf} for the detailed proof of Proposition \ref{prop:vix_affine_cf}.


  
  
\begin{corollary}\label{coll:1}
The integral quantities $\int_0^{\infty} x \, \nu (\mathrm{d}x)$ and $\int_0^{\infty} (1 + \rho x - e^{\rho x}) \, \nu(\mathrm{d}x)$ in Eq.~\eqref{eq:vix_coeff_B} for different BDLPs are summarized in the following table. 

\begin{center}
\resizebox{\textwidth}{!}{%
$
\begin{array}{ccccc}
\hline
\text{Model} & \text{OU-}\Gamma & \Gamma\text{-OU} & \text{OU-TS} & \text{TS-OU} \\
\hline
\text{L\'evy measure} & 
\dfrac{\alpha e^{-\beta x}}{x} \, \mathrm{d}x & 
\alpha \beta e^{-\beta x} \, \mathrm{d}x & 
\dfrac{\alpha e^{-\beta x}}{x^{\theta+1}} \, \mathrm{d}x & 
\dfrac{\alpha e^{-\beta x}}{x^{\theta+1}} (\beta x + \theta) \, \mathrm{d}x \\[15pt]

\displaystyle \int_0^{\infty} x \, \nu (\mathrm{d}x) & 
\dfrac{\alpha}{\beta} & 
\dfrac{\alpha}{\beta} & 
\alpha \beta^{\theta-1} \Gamma(1-\theta) & 
\alpha \beta^{\theta-1} \Gamma(1-\theta) \\[15pt]

\displaystyle \int_{0}^{\infty} (1 + \rho x - e^{\rho x}) \, \nu(\mathrm{d}x) & 
-\dfrac{\rho \alpha }{\beta} + \alpha \ln \dfrac{\beta}{\beta-\rho} & 
\dfrac{\alpha \rho^2}{\beta (\beta-1)} & 
\alpha \beta^\theta \Gamma(-\theta) \left\{(1-\dfrac{\rho}{\beta})^\theta -1 + \dfrac{\rho \theta}{\beta}\right\} & 
\alpha \rho \Gamma(1-\theta) \left(-\beta^{\theta-1} + (\beta-\rho)^{\theta-1}\right) \\[15pt]
\hline
\end{array}
$
}
\end{center}
\end{corollary}

With the integral quantities explicitly evaluated for various Lévy processes in \Cref{coll:1}, the conditional characteristic function $\Phi(\phi; V_t, \tau)$ in Eq.~\ref{eq:exp_chf} is fully specified in closed form. Having obtained the exact analytical characteristic function, we can apply the Fourier-cosine method to compute the benchmark price of the VIX option. The Fourier-cosine algorithm for pricing the VIX call option is presented in \Cref{prop:vix_cf}.

\begin{proposition}
\label{prop:vix_cf}
Given the characteristic function $\Phi(\phi; V_t, \tau)$ of $\text{VIX}^2_T$, the price of the VIX call option with strike $K$ and maturity $T$ can be approximated using the COS series expansion:
\begin{equation}
    C_t \approx e^{-r(T-t)} \sideset{}{'}\sum_{n=0}^{N-1} \Re \left[ \Phi\left( \frac{n\pi}{b-a}; V_t, \tau \right) e^{-i \frac{n a \pi}{b-a}} \right] U_n,
\end{equation}
where $[a, b]$ is the truncation domain, $[a, b] = [\max(c_1 - L \sqrt{c_2+\sqrt{c_4}}, 0), c_1 + L\sqrt{c_2+\sqrt{c_4}}]$, where $c_j$ represents the $j$th cumulant of $\text{VIX}^2_T$ defined as 
$$ c_j = \frac{1}{i^j} \frac{\partial^j}{\partial \phi^j} \log \Phi(\phi; V_t, \tau) \Big|_{\phi=0}. $$
The coefficients $U_n$, corresponding to the payoff $(\sqrt{\text{VIX}_T^2} - K)^+$, are derived analytically as follows:
\begin{equation*}
U_n = 
\begin{cases}
  \displaystyle \frac{2}{b-a} \left( \frac{2}{3} b^{\frac{3}{2}} + \frac{1}{3} K^3 - K b \right) & \text{if } n = 0, \ a < K^2 \\[10pt]
  
  \displaystyle \frac{2}{b-a} \left( \frac{2}{3} b^{\frac{3}{2}} - \frac{2}{3} a^{\frac{3}{2}} - K b + K a \right) & \text{if } n = 0, \ a \ge K^2 \\[10pt]
  
  \displaystyle \frac{2}{b-a} \Re \Bigg\{ e^{-iw_na} \Bigg[
    \frac{\sqrt{b}-K}{iw_n} e^{iw_nb} - \mathbf{1}_{\{a \ge K^2\}}\frac{\sqrt{a} - K}{iw_n} e^{iw_na} + \Lambda_n
  \Bigg] \Bigg\} & \text{if } n \neq 0
\end{cases}
\end{equation*}
with $w_n = \frac{n\pi}{b-a}$ and the auxiliary term $\Lambda_n$ is given by
\begin{equation*}
    \Lambda_n = \frac{\sqrt{\pi}}{2(-iw_n)^{3/2}} \left( \operatorname{erfz}(\sqrt{-iw_n b}) - \operatorname{erfz}(\sqrt{-iw_n \max(a, K^2)}) \right).
\end{equation*}
Here, $\operatorname{erfz}(z) = \frac{2}{\sqrt{\pi}} \int_{0}^{z} e^{-t^2} \mathrm{d}t$ is the complex error function.
\end{proposition}

\section{Second-order convergence of the operator splitting schemes}
\label{sec:error-analysis}

\noindent
The operator splitting schemes derived for various stochastic volatility models share the same prototype
 structure. We now formulate this structure using the notion of transition semigroups and establish
a unified second-order convergence of the operator splitting schemes.
The proof of the second-order convergence of the operator splitting schemes consists of two steps. First, 
we show that the prototype structure 
has the local weak error of order $O(h^3)$, where $h$ is the time step. Let $T$ denote the time horizon 
of the computational problem. Second, a telescoping argument over
$N=T/h$ time steps gives an accumulated global weak error of order
\[
    N O(h^3)
    =
    \frac{T}{h} \, O(h^3)
    =
    O(h^2).
\]

Let $(X_t,V_t)$ denote the original Markov process on the state space
$\mathcal X\times\mathcal V$, where $X_t$ is the log-asset price and $V_t$ denotes the
variance state variable. For the Heston model, $4/2$ model, Hull-White model, and BN-S model, $V_t$ is
one-dimensional. For the lifted Heston model, $V_t$ is taken to be the vector
\[
    V_t=(V_{t,1},\ldots,V_{t,N})\in\mathcal V\subseteq\mathbb R_+^N .
\]
Let $\{\mathcal{P}_t:t\geq0\}$ be the transition semigroup of the original process underlying the 
stochastic volatility model as defined by
\begin{equation}
    \mathcal{P}_t f(x,v)
    =
    \mathbb E\!\left[f(X_t,V_t)\mid X_0=x,V_0=v\right],
    \label{eq:original-semigroup}
\end{equation}
associated with the infinitesimal generator $\mathcal G$.
Here, $f$ is the test function.
In each stochastic volatility model considered in this paper, the operator 
splitting procedure amounts to the splitting of the infinitesimal generator into
\begin{equation*}
    \mathcal G=\mathcal G^{(1)}+\mathcal G^{(2)},
    \label{eq:generator-decomposition}
\end{equation*}
where $\mathcal G^{(1)}$ and $\mathcal G^{(2)}$ are the infinitesimal generators of the two
sub-operators with transition semigroups $\{\mathcal{P}_t^{(1)}:t\geq0\}$ and
$\{\mathcal{P}_t^{(2)}:t\geq0\}$, respectively. Over one time step $h>0$, we define the composite transition semigroup associated
with the Strang operator splitting procedure by
\begin{equation}
    \mathcal{P}_h^{\text{split}}
    =
    \mathcal{P}_{h/2}^{(1)}
    \mathcal{P}_h^{(2)}
    \mathcal{P}_{h/2}^{(1)} .
    \label{eq:strang-one-step}
\end{equation}
For $T=Nh$, the corresponding composite $N$-step operator splitting approximation is $\mathcal{P}_T^{\text{split}} = \left(\mathcal{P}_h^{\mathrm{split}}\right)^N$. 
Accordingly, we have
\begin{equation}
    \mathcal{P}_T^{\text{split}} f(x,v)
	:=
\left(P_h^{\mathrm{split}}\right)^N f
    =
    \mathbb E\!\left[
    f(X_T^{\mathrm{split}},V_T^{\mathrm{split}})
    \mid X_0^{\mathrm{split}}=x,V_0^{\mathrm{split}}=v
    \right].
    \label{eq:split-semigroup}
\end{equation}

\noindent
In our later exposition, the regularity assumptions in \Cref{prop:terminal-weak-error} involve the notion of commutator. For two given operators $A$ and $B$, we define their commutator by
\begin{equation*}
    [A,B]=AB-BA.
\end{equation*}

\subsection{Connection with the model specific operator splitting schemes}
\label{subsec:model-applicability}

\noindent
The splitting infinitesimal generators $\mathcal G^{(1)}$ and $\mathcal G^{(2)}$ introduced above
are associated with the sub-generators in each specific stochastic volatility model. For the Heston, $4/2$, and Hull--White models, the state variables are $(X,V)$. 
The sub-operator 1 is chosen so that the integrated variance appearing in the provisional log-price
 can be derived in closed form. The sub-operator 2 contains the correlated
stochastic variance component, which is computed via simulation calculations. 
For the lifted Heston model, the same interpretation applies after replacing the scalar
variance state $V$ by the vector of variance factors $(V_1,\ldots,V_N)$.
For the BN-S model, the sub-operator 1 corresponds to the continuous OU decay component and
Brownian innovations, while the sub-operator 2 corresponds to the pure BDLP.

Therefore, the simulation schemes derived in the preceding sections dictate the model specific
choices of $\mathcal P_t^{(1)}$ and $\mathcal P_t^{(2)}$. Although the explicit
transition semigroups may differ across different stochastic volatility models, all operator splitting 
schemes share the common prototype structure as defined in Eq.~\eqref{eq:strang-one-step}. 
The convergence proof presented below can be applied without distinction
 to each of these model specific simulation schemes.

\subsection{Second-order convergence for the N-step operator splitting approximation}
\label{subsec:terminal-weak-error}
\noindent
We write $P_T=(P_h)^N$, $T=Nh$, as the $N$-step transition semigroup
of the original process. We would like to establish the second-order convergence
of $P_T^{\mathrm{split}}$ [see Eq.~\eqref{eq:split-semigroup}] to $P_T$.
Specifically, we examine the order of the weak error $ \|P_T^{\mathrm{split}}f-P_Tf\|_\infty,$
where $\|\cdot\|_\infty$ is the maximum norm and $f$ is a test function.
In the weak convergence analysis, $f$ is a test function applied to the
terminal state. Proposition~\ref{prop:terminal-weak-error} examines 
the discrepancy of the distribution of the operator splitting scheme at some 
terminal time to the exact solution. In the context of pricing a European option 
under stochastic volatility model, Proposition~\ref{prop:terminal-weak-error} establishes
that the order of operator splitting approximation in the corresponding option price is $O(h^2)$.

\begin{proposition}
\label{prop:terminal-weak-error}

Let $f:\mathcal X\times\mathcal V\to\mathbb R$ be a test function such that, for any
$s,r,u\in[0,T]$, the functions
\[
    \mathcal{P}_r^{(1)}
    \mathcal{P}_s^{(2)}
    \mathcal{P}_s^{(1)}
    \mathcal{P}_u f
    \quad\text{and}\quad
    \mathcal{P}_r^{(2)}
    \mathcal{P}_s^{(1)}
    \mathcal{P}_u f
\]
belong to the domains of the double commutators
\[
    [\mathcal G^{(1)},[\mathcal G^{(1)},\mathcal G^{(2)}]]
    \quad\text{and}\quad
    [\mathcal G^{(2)},[\mathcal G^{(2)},\mathcal G^{(1)}]],
\]
respectively. Suppose that $f$ satisfies the following boundness conditions 
\begin{subequations}\label{eq:commutator-bounds}
\begin{align}
\sup_{s,r,u\in[0,T]}
\left\|
[\mathcal G^{(1)},[\mathcal G^{(1)},\mathcal G^{(2)}]]
\mathcal{P}_r^{(1)}
\mathcal{P}_s^{(2)}
\mathcal{P}_s^{(1)}
\mathcal{P}_u f
\right\|_\infty
&< \infty,
\label{eq:commutator-bound-1}
\\[4pt]
\sup_{s,r,u\in[0,T]}
\left\|
[\mathcal G^{(2)},[\mathcal G^{(2)},\mathcal G^{(1)}]]
\mathcal{P}_r^{(2)}
\mathcal{P}_s^{(1)}
\mathcal{P}_u f
\right\|_\infty
&< \infty.
\label{eq:commutator-bound-2}
\end{align}
\end{subequations}
Then there exists a constant $C>0$, independent of $h$, such that
\begin{equation}
    \sup_{(x,v)\in\mathcal X\times\mathcal V}
    \left|
    \mathbb E\!\left[
        f(X_T^{\mathrm{split}},V_T^{\mathrm{split}})
        \mid
        X_0^{\mathrm{split}}=x,
        V_0^{\mathrm{split}}=v
    \right]
    -
    \mathbb E\!\left[
        f(X_T,V_T)
        \mid
        X_0=x,
        V_0=v
    \right]
    \right|
    \leq Ch^2.
    \label{eq:terminal-weak-error}
\end{equation}
That is,
$\|P_T^{\mathrm{split}}f-P_Tf\|_\infty=O(h^2)$.
\end{proposition}

\begin{proof}

By Eqs.~\eqref{eq:original-semigroup} and \eqref{eq:split-semigroup}, we have
\begin{align}
    &\sup_{(x,v)\in\mathcal X\times\mathcal V}
    \left|
    \mathbb E\!\left[
        f(X_T^{\mathrm{split}},V_T^{\mathrm{split}})
        \mid
        X_0^{\mathrm{split}}=x,
        V_0^{\mathrm{split}}=v
    \right]
    -
    \mathbb E\!\left[
        f(X_T,V_T)
        \mid
        X_0=x,
        V_0=v
    \right]
    \right|
    \notag\\
    &=
    \| \mathcal{P}_T^{\text{split}} f-\mathcal P_Tf\|_\infty
    \notag =
    \|(\mathcal{P}_h^{\text{split}})^N f-\mathcal (P_h)^Nf\|_\infty
    \notag =
    \left\|
    \sum_{i=0}^{N-1}
    \mathcal{P}_{(N-i-1)h}^{\text{split}}
    (\mathcal{P}_h^{\text{split}}-\mathcal P_h)
    \mathcal P_{ih} f
    \right\|_\infty
    \notag\\
    &\leq
    \sum_{i=0}^{N-1}
    \left\|
    (\mathcal{P}_h^{\text{split}}-\mathcal P_h)
    \mathcal P_{ih}f
    \right\|_\infty,
    \label{eq:weak-error-decomposition}
\end{align}
where we have used the contraction property of $\mathcal{P}_h^{\text{split}}$.

We next estimate the local weak error. By the Strang splitting residual formula in
\citet[Lemma~2.4]{DescombesSchatzman2002}, we obtain
\begin{align}
    &\left\|
    (\mathcal{P}_h^{\text{split}}-\mathcal P_h)
    \mathcal P_{ih}f
    \right\|_\infty
    \notag\\
    &\leq
    \int_0^h\!\!\int_0^s
    (s-r)
    \left\|
    \mathcal P_{h-s}
    \mathcal P_{s-r}^{(1)}
    [\mathcal G^{(1)},[\mathcal G^{(1)},\mathcal G^{(2)}]]
    \mathcal P_r^{(1)}
    \mathcal P_s^{(2)}
    \mathcal P_s^{(1)}
    \mathcal P_{ih}f
    \right\|_\infty
    \,\mathrm{d}r\,\mathrm{d}s
    \notag\\[4pt]
    &\quad+
    \int_0^h\!\!\int_0^s
    (s-r)
    \left\|
    \mathcal P_{h-s}
    \mathcal P_s^{(1)}
    \mathcal P_r^{(2)}
    [\mathcal G^{(2)},[\mathcal G^{(2)},\mathcal G^{(1)}]]
    \mathcal P_{s-r}^{(2)}
    \mathcal P_s^{(1)}
    \mathcal P_{ih}f
    \right\|_\infty
    \,\mathrm{d}r\,\mathrm{d}s
    \notag\\
    &\leq
    \int_0^h\!\!\int_0^s
    (s-r)
    \left\|
    [\mathcal G^{(1)},[\mathcal G^{(1)},\mathcal G^{(2)}]]
    \mathcal P_r^{(1)}
    \mathcal P_s^{(2)}
    \mathcal P_s^{(1)}
    \mathcal P_{ih}f
    \right\|_\infty
    \,\mathrm{d}r\,\mathrm{d}s
    \notag\\[4pt]
    &\quad+
    \int_0^h\!\!\int_0^s
    (s-r)
    \left\|
    [\mathcal G^{(2)},[\mathcal G^{(2)},\mathcal G^{(1)}]]
    \mathcal P_{s-r}^{(2)}
    \mathcal P_s^{(1)}
    \mathcal P_{ih}f
    \right\|_\infty
    \,\mathrm{d}r\,\mathrm{d}s
    \notag\\
    &\leq
    C_{\mathrm{loc}}
    \int_0^h\!\!\int_0^s
    (s-r)\,\mathrm{d}r\,\mathrm{d}s
    =
    \frac{1}{6}C_{\mathrm{loc}}h^3,
    \label{eq:local-third-order-error}
\end{align}
where we have used the contraction properties of $\mathcal P_t$, $\mathcal P_t^{(1)}$,
and $\mathcal P_t^{(2)}$, and $C_{\mathrm{loc}}>0$ is independent of $h$. Thus, the prototype Strang splitting has local weak error of order $O(h^3)$.
Substituting Eq.~\eqref{eq:local-third-order-error} into
Eq.~\eqref{eq:weak-error-decomposition}, we obtain the weak error bound of $O(h^2)$, where
\begin{align}
    \|\mathcal{P}_T^{\text{split}}f-\mathcal P_Tf\|_\infty
    &\leq
    \sum_{i=0}^{N-1}
    \frac{1}{6}C_{\mathrm{loc}}h^3
    \notag =
    \frac{1}{6}C_{\mathrm{loc}}Nh^3
    \notag =
    \frac{1}{6}C_{\mathrm{loc}}Th^2.
	\qedhere
\end{align}
\end{proof}

\subsection{Second-order convergence for pricing options with discretely monitored path dependent payoffs}
\label{subsec:path-dependent-weak-error}


\noindent
For a discretely monitored path dependent option, the payoff depends on the sequence of 
state variables on a finite collection of monitoring dates. Accordingly, we consider a
multivariate test function $g:(\mathbb X\times\mathbb V)^M\to\mathbb R$.
Proposition \ref{prop:second_order_asian} compares the expectation of $g$ under the operator splitting
approximation with its expectation under the original process. If $g$ is
chosen as a path dependent payoff function satisfying the assumptions stated below,
\Cref{prop:second_order_asian} shows that the order of operator splitting approximation in
the corresponding option price is $O(h^2)$.
The following backward recursion is introduced to reduce the
multi-step weak error analysis to a sequence of single-step
estimates under which the results in Proposition~\ref{prop:terminal-weak-error}
can be applied.

Let \(g_M=g_M^{\mathrm{split}}=g\), and for \(i=0,1,\ldots,M-1\), we define recursively
\begin{align*}
g_i(x_1,v_1,\ldots,x_i,v_i)
&=
\mathbb{E}\!\left[
g_{i+1}
\left(
x_1,v_1,\ldots,x_i,v_i,
X_{t_{i+1}},V_{t_{i+1}}
\right)
\,\middle|\,
(X_{t_i},V_{t_i})=(x_i,v_i)
\right],
\\
g_i^{\mathrm{split}}(x_1,v_1,\ldots,x_i,v_i)
&=
\mathbb{E}\!\left[
g_{i+1}^{\mathrm{split}}
\left(
x_1,v_1,\ldots,x_i,v_i,
X_{t_{i+1}}^{\mathrm{split}},
V_{t_{i+1}}^{\mathrm{split}}
\right)
\,\middle|\,
\left(
X_{t_i}^{\mathrm{split}},
V_{t_i}^{\mathrm{split}}
\right)
=
(x_i,v_i)
\right].
\end{align*}

\begin{proposition}\label{prop:second_order_asian}
	Suppose $g\!:(\mathbb{X} \times \mathbb{V})^M \to \mathbb{R}$ is a function such that for any $i = 1, 2, \ldots, M$ and $x_1, v_1, \ldots, x_{i-1}, v_{i-1}$, $g_i(x_1, v_1, \ldots, x_{i-1}, v_{i-1}, \cdot, \cdot)$ satisfies the assumption imposed in Proposition \ref{prop:terminal-weak-error} on $f$ with finite bounds as stated in Eqs.~(\ref{eq:commutator-bounds}a, b)
	uniformly for $(x_1, v_1, \ldots, x_{i-1}, v_{i-1} ) \in (\mathbb{X} \times \mathbb{V})^{i-1}$. We then have 
	\begin{equation}
	\begin{aligned}
		&\sup_{(x, v) \in \mathbb{X} \times \mathbb{V}}| \mathbb{E}[g(X_{t_1}^{split}, V_{t_1}^{split}, \ldots, X_{t_M}^{split}, V_{t_M}^{split})|X_0^{split}=x, V_0^{split}=v]  \\
		&\qquad\qquad\qquad\qquad\qquad- \mathbb{E}[g(X_{t_1}, V_{t_1}, \ldots,X_{t_M}, V_{t_M})|X_0=x, V_0=v]| \le Ch^2,
	\end{aligned}
	\label{eq:second_order_asian_weak_error}
	\end{equation}
	for some constant $C>0$ independent of $h$. 
\end{proposition}

\begin{proof}
We prove the second-order bound by backward recursion on the 
index $i$, starting from $i=M$ and proceeding to $i=0$.
First, since \(g_M^{\mathrm{split}}=g_M=g\), we have
\begin{align*}
\sup_{
(x_1,v_1,\ldots,x_M,v_M)
\in
(\mathbb{X}\times\mathbb{V})^M
}
\left|
g_M^{\mathrm{split}}(x_1,v_1,\ldots,x_M,v_M)
-
g_M(x_1,v_1,\ldots,x_M,v_M)
\right|
=0.
\end{align*}
For \(i=0,1,2,\ldots,M-1\), suppose that there exists a constant
\(C_{i+1}>0\), independent of \(h\), such that
\begin{align*}
\sup_{
(x_1,v_1,\ldots,x_{i+1},v_{i+1})
\in
(\mathbb{X}\times\mathbb{V})^{i+1}
}
\Big|
&
g_{i+1}^{\mathrm{split}}
(x_1,v_1,\ldots,x_{i+1},v_{i+1}) -
g_{i+1}
(x_1,v_1,\ldots,x_{i+1},v_{i+1})
\Big|
\leq C_{i+1}h^2.
\end{align*}
Consider the backward recursion, where
{
	\allowdisplaybreaks
	\begin{align*}
&
\sup_{
(x_1,v_1,\ldots,x_i,v_i)
\in
(\mathbb{X}\times\mathbb{V})^i
}
\left|
g_i^{\mathrm{split}}(x_1,v_1,\ldots,x_i,v_i)
-
g_i(x_1,v_1,\ldots,x_i,v_i)
\right|
\nonumber\\
={}&
\sup_{
(x_1,v_1,\ldots,x_i,v_i)
\in
(\mathbb{X}\times\mathbb{V})^i
}
\Bigg|
\mathbb{E}\!\left[
g_{i+1}^{\mathrm{split}}
\left(
x_1,v_1,\ldots,x_i,v_i,
X_{t_{i+1}}^{\mathrm{split}},
V_{t_{i+1}}^{\mathrm{split}}
\right)
\right.
\left.
\middle|\,
\left(
X_{t_i}^{\mathrm{split}},
V_{t_i}^{\mathrm{split}}
\right)
=(x_i,v_i)
\right]
\nonumber\\
&\hspace{1.2cm}
-
\mathbb{E}\!\left[
g_{i+1}
\left(
x_1,v_1,\ldots,x_i,v_i,
X_{t_{i+1}},
V_{t_{i+1}}
\right)
\,\middle|\,
(X_{t_i},V_{t_i})=(x_i,v_i)
\right]
\Bigg|
\nonumber\\
\leq{}&
\sup_{
(x_1,v_1,\ldots,x_i,v_i)
\in
(\mathbb{X}\times\mathbb{V})^i
}
\Bigg|
\mathbb{E}\!\left[
g_{i+1}
\left(
x_1,v_1,\ldots,x_i,v_i,
X_{t_{i+1}}^{\mathrm{split}},
V_{t_{i+1}}^{\mathrm{split}}
\right)
\right.
\left.
\middle|\,
\left(
X_{t_i}^{\mathrm{split}},
V_{t_i}^{\mathrm{split}}
\right)
=(x_i,v_i)
\right]
\nonumber\\
&\hspace{1.2cm}
-
\mathbb{E}\!\left[
g_{i+1}
\left(
x_1,v_1,\ldots,x_i,v_i,
X_{t_{i+1}},
V_{t_{i+1}}
\right)
\,\middle|\,
(X_{t_i},V_{t_i})=(x_i,v_i)
\right]
\Bigg|
+
C_{i+1}h^2
\nonumber\\
\leq{}&
C_{i+1}'h^2+C_{i+1}h^2
=
C_i h^2,
\end{align*}
}

\noindent 
where the result in Proposition~\ref{prop:terminal-weak-error} is used in the last inequality,
\(C_{i+1}'>0\) is a constant independent of \(h\), and
$C_i=C_{i+1}'+C_{i+1}.$ Backward recursion implies that the error bound holds for \(i=0\), that is, 
$\left|g_0^{\mathrm{split}}-g_0\right|
\leq Ch^2.$ This is essentially the result in Eq.~\eqref{eq:second_order_asian_weak_error}.
\end{proof}

	 \section{Numerical experiments}\label{s:num}\noindent
	 In this section, we present the numerical experiments that were performed to 
	 illustrate both efficiency and robustness of our second-order operator 
	 splitting schemes under the Heston model, 4/2 model, HW model, 
	 lifted Heston model, and BN-S model. 

	In our numerical tests, unless otherwise specified, we run $N = 10^5$ paths for each simulation run and repeated the process $m = 50$ times. For each simulation run, we compute the mean of the option prices and the corresponding CPU time, resulting in $m$ estimators for the option price and $m$ CPU time values. From these $m$ samples, we calculate the average option price and average CPU time. We then compute the root mean square error (RMSE) by
	 \begin{equation}
		\text{RMSE}=\sqrt{\text{StDev}^2+\text{Bias}^2},
	 \end{equation}
	where $\text{Bias} = \text{simulation value} - \text{benchmark value}$ and $\text{StDev}$ denotes the standard deviation of the simulation estimator.

	\subsection{Heston model, 4/2 model and Hull-White model}\label{subsection:heston_4_2_hull_white}\noindent

	\noindent We first show the numerical results for the Heston type models, including the Heston model, 4/2 model and lifted Heston model. 
	We compare our second-order operator splitting scheme with several established methods, including the Euler discretization scheme, trapezoidal rule, and Hilbert interpolation scheme in \citet{Zeng2023}. 

We consider pricing of the monthly monitored European arithmetic Asian call option with one-year maturity under the Heston model, where the option payoff depends on the discrete average of the asset prices. Specifically, we take the average of the prices observed at the end of each month alongside the initial price $S_0$. This results in a total of 13 prices used for calculating the average, which is $\frac{S_0 + S_1 + \dots + S_{12}}{13}$. 
In Figure~\ref{fig:heston_asian}, we show the plots of RMSE versus CPU time for pricing the Asian option using various simulation schemes. 
 The parameter values of the Asian call option are taken from \citet{Zeng2023}: $S_0=100,\,V_0=0.010201,\,\theta=0.019,\,\sigma=0.61,\,\rho=-0.7,\,k=6.21,\,T=1,\,r=0.0319\text{ and strike }K=100.$ The benchmark option value is 3.5665, obtained
 by \cite{Zeng2023} using very large number of simulation paths.
 From the plots, we conclude that the second-order operator splitting scheme is the most efficient one in terms of accuracy-speed tradeoff when pricing
 the Asian call option under the Heston model. The trapezoidal scheme shows relatively better performance when compared with the Hilbert 
 interpolation scheme. The Euler scheme exhibits the weakest performance when compared with the other three simulation schemes. 

	\begin{figure}[H]
		\centering
		\includegraphics[width=0.6\linewidth]{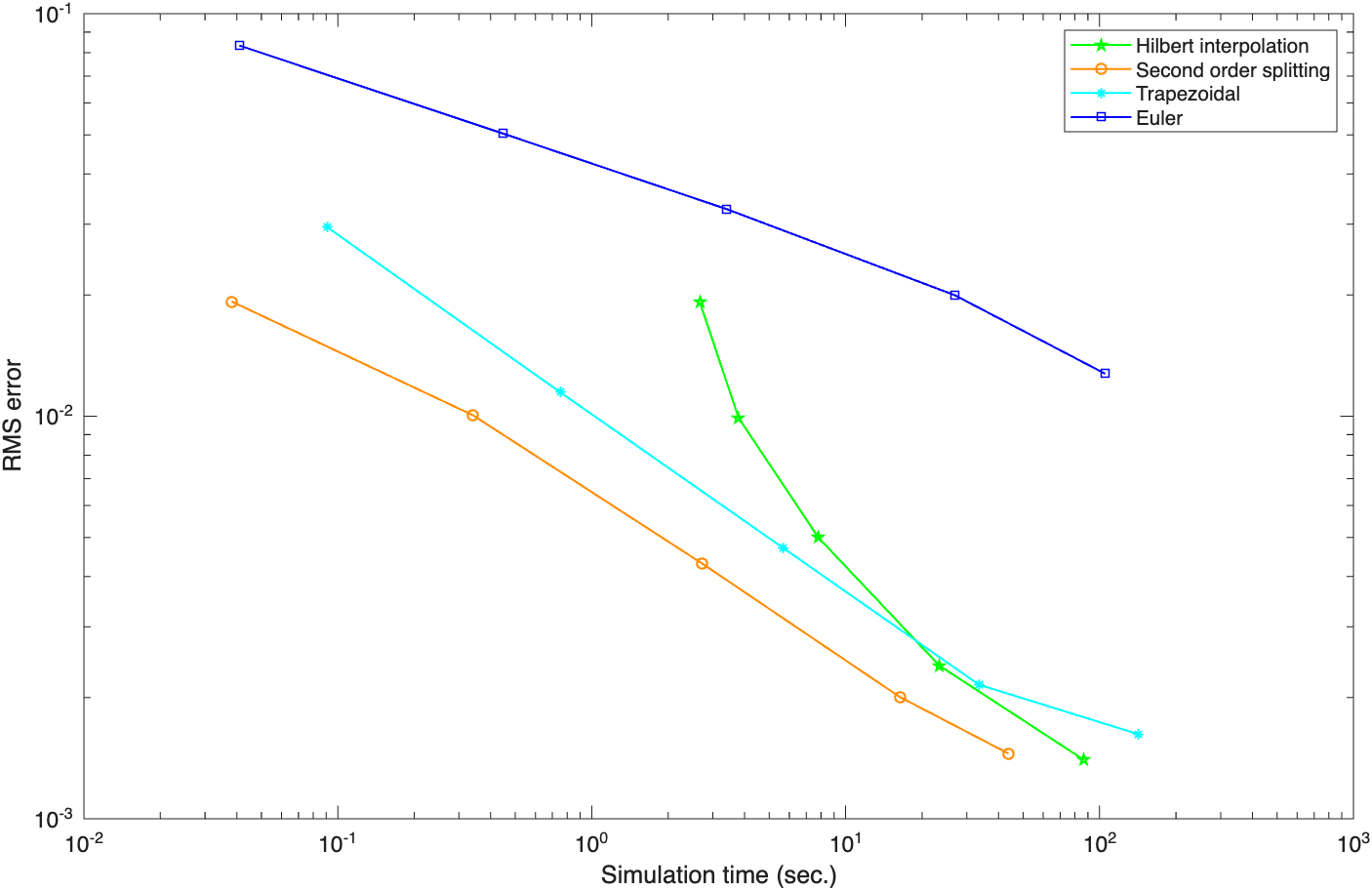}
		\caption{Plots of RMSE versus CPU time for various simulation schemes for pricing monthly monitored arithmetic Asian call option under the Heston model. The sample sizes and time steps for the second-order operator splitting scheme and trapezoidal rule are $(5*10^4,1/12)$, $(2*10^5,1/24)$, $(8*10^5,1/48)$, $(32*10^5,1/72)$ and $(64*10^5,1/96)$, and those of the Euler scheme are $(5*10^4,1/96)$, $(2*10^5,1/192)$, $(8*10^5,1/384)$, $(32*10^5,1/768)$ and $(64*10^5,1/1536)$.}
		\label{fig:heston_asian}
	\end{figure}

	In Figure~\ref{fig:32_barrier}, we show the plots of RMSE versus CPU time for pricing the monthly monitored European knock-out barrier put option under the 3/2 model.
	 The parameter values of the barrier put option are taken from \citet{Zeng2023}: $S_0=100,\,V_0=16.6597,\,\theta=19.3038,\,\sigma=-8.56,\,\rho=-0.99,\,k=4.9790,\,T=1,\,r=0.0048\text{ and strike }K=100.$ The benchmark option value is 1.9580.
	 Again, the second-order operator splitting scheme is shown to be more efficient than the Hilbert interpolation scheme and Euler 
	 scheme by a wide margin. Under the scenario of taking larger number of simulation paths and finer time steps, the Hilbert interpolation
	 scheme shows significant drop in RMSE, surpassing the performance of the Euler scheme. 
	In conclusion, these plots for the RMSE versus CPU time of the path dependent barrier put option reveals that the second-order operator splitting scheme shows better 
	performance in terms of accuracy-speed tradeoff than the Hilbert interpolation scheme of \citet{Zeng2023} under the 3/2 model.\\
	
	\begin{figure}[H]
	\centering
	\includegraphics[width=0.6\linewidth]{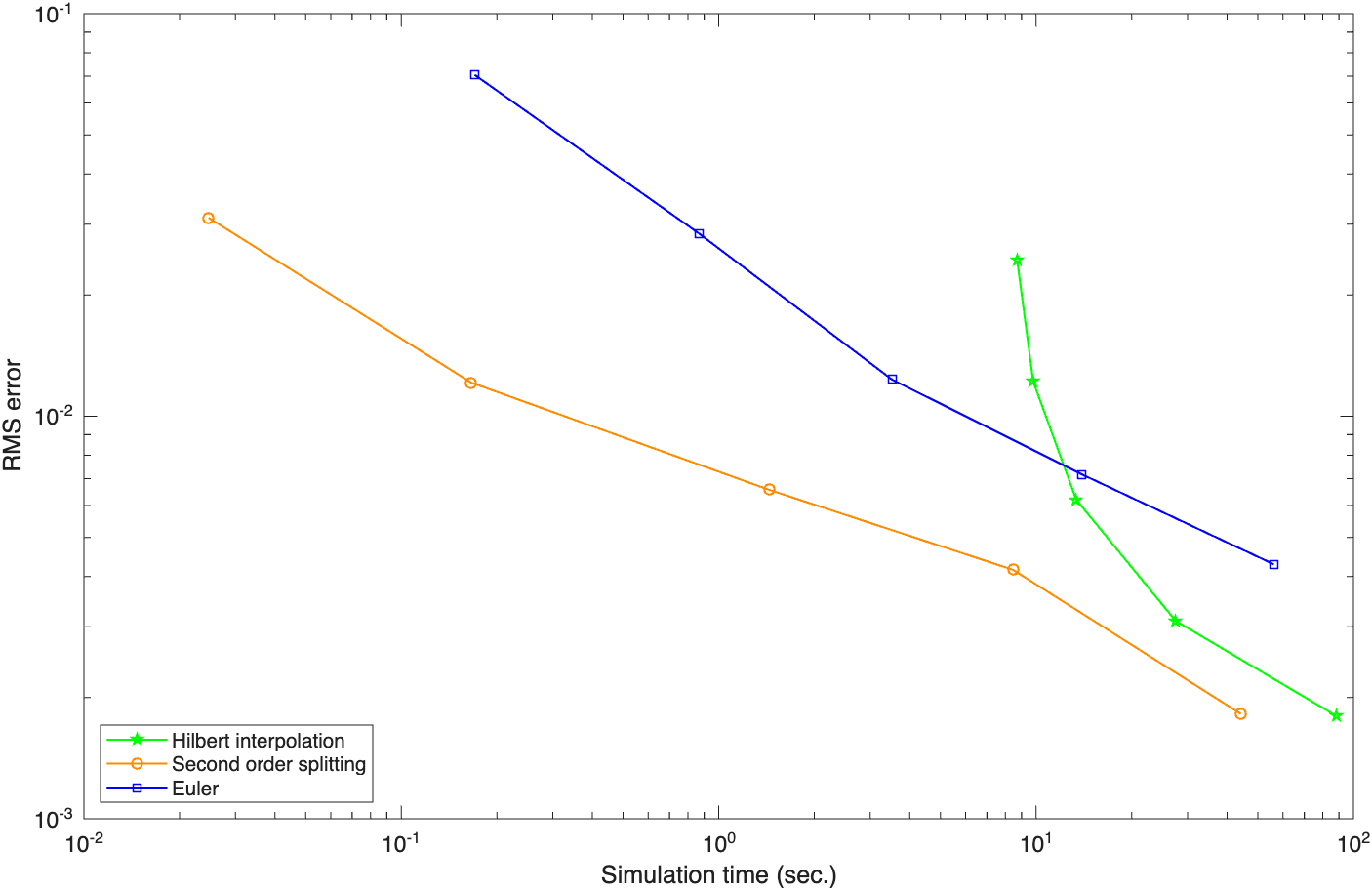}
	\caption{Plots of RMSE versus CPU time for various simulation schemes for pricing monthly monitored knock-out barrier put option under the 3/2 model. The sample sizes and time steps for the second-order operator splitting scheme are $(5*10^4,1/12)$, $(2*10^5,1/24)$, $(8*10^5,1/48)$, $(16*10^5,1/72)$ and $(32*10^5,1/96)$, and those of the Euler scheme are $(10^5,1/96)$, $(2*10^5,1/192)$, $(4*10^5,1/384)$, $(8*10^5,1/768)$ and $(16*10^5,1/1536)$.}
	\label{fig:32_barrier}
\end{figure}

	 Next, we compare the performance in terms of accuracy-speed tradeoff of our second-order operator splitting scheme, Euler scheme, trapezoidal scheme and Algorithm~2 
	 in \citet{brigone2024104861} for pricing European call option under the HW model. The parameter values are taken from \citet{brigone2024104861}: $S_0=100,\,V_0=0.04,\,\mu=0.2,\,\xi=2,\,\rho=-0.2,\,T=1,\,r=0.02,\,\text{strike }K=100$, corresponding to set B for Brignone-Gonzato Algorithm~2. The benchmark option value is 8.0361.
	The plots of RMSE versus CPU time for various simulation schemes are shown in Figure~\ref{fig:hull_white_rms_set_B}. We observe that both the trapezoidal rule and second-order operator splitting scheme perform much better than Algorithm~2 in \citet{brigone2024104861},
	while the second-order operator splitting scheme is slightly better than the trapezoidal rule. For the same level of RMSE, 
	the CPU time required for the second-order operator splitting scheme is the lowest among all simulation schemes. 
	Unfortunately, the sophisticated Brignone-Gonzato algorithm even does not perform better than the Euler scheme. 
	\begin{figure}[H]
		\centering
		\includegraphics[width=0.6\linewidth]{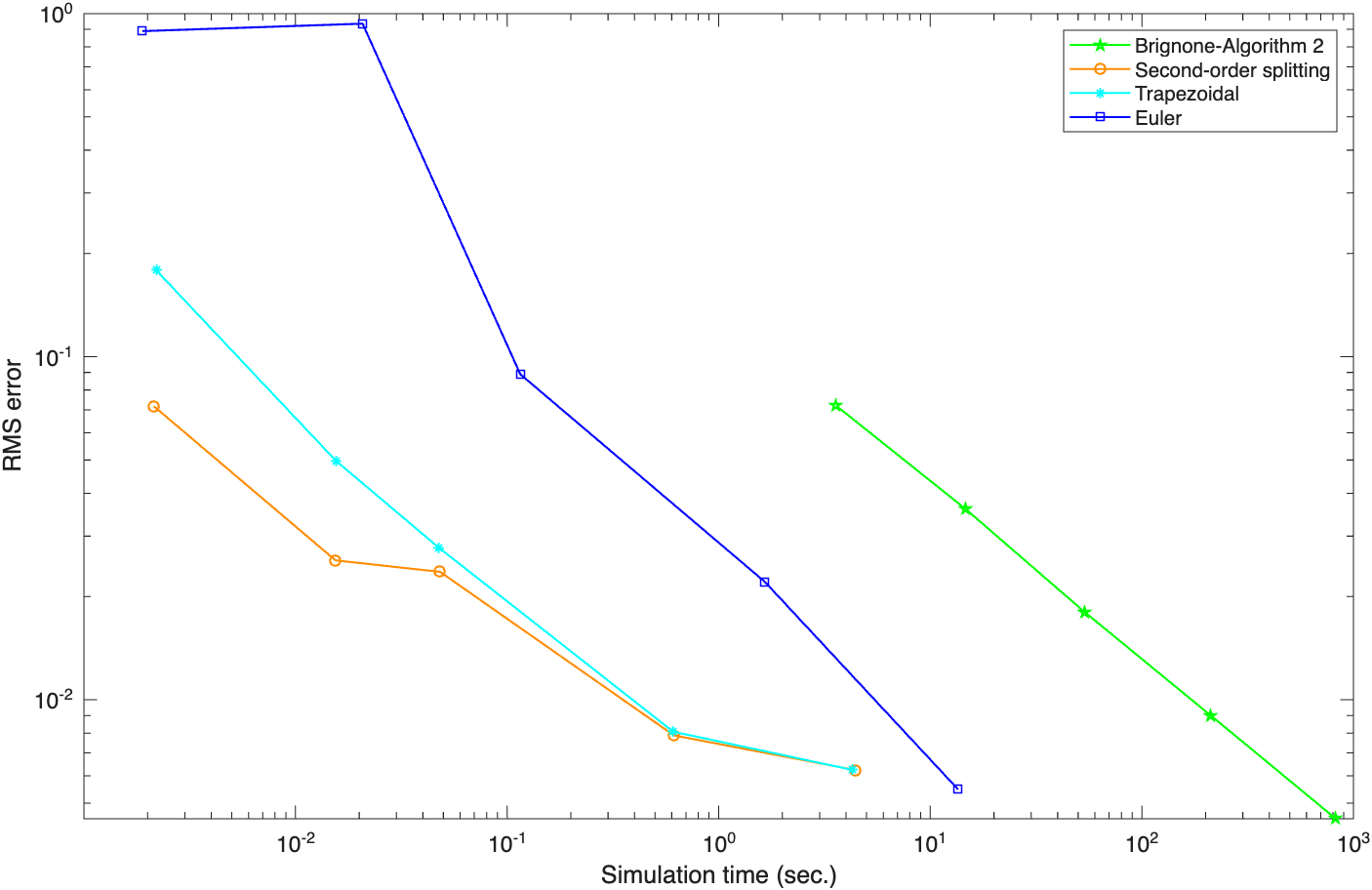}
		\caption{Plots of RMSE versus CPU time for various simulation schemes of pricing European call option for the Hull-White model under parameter set B in \citet{brigone2024104861}. The sample sizes and size of time steps for the second-order operator splitting scheme and trapezoidal rule are $(10^5,1)$, $(4*10^5,1/2)$, $(8*10^5,1/4)$, $(32*10^5,1/16)$ and $(64*10^5,1/64)$, and those of the Euler scheme are $(10^5,1)$, $(4*10^5,1/4)$, $(8*10^5,1/16)$, $(32*10^5,1/64)$ and $(64*10^5,1/256)$.}
		\label{fig:hull_white_rms_set_B}
	\end{figure}
	
	\noindent
	
	\subsection{Lifted Heston model}
	\label{subsection:liftedHeston}
	\noindent Among the simulation schemes for pricing options under the lifted Heston model, we compare our second-order operator splitting scheme with the weak simulation scheme proposed by \citet{bayer2024efficient}, and the drift-implicit Euler scheme. The formulation of the drift-implicit Euler scheme as recommended by \citet{bayer2024efficient} is given by
	\begin{equation}\label{eq:euler_lifted_implicit}
		\begin{cases}
			X_{t+h}=X_t-\frac{(V^N_t)^{+} h}{2}+\sqrt{(V^N_t)^{+}}(\rho \Delta W_t^h+\sqrt{1-\rho^2} \Delta W_t^{\perp,h}),\\
			V^{(i)}_{t+h}=V^{(i)}_t-x_i(V_{t+h}^{(i)}-v_0^i)h+(\theta-\lambda V^{N}_{t+h})h+\nu\sqrt{(V_t^{N})^{+}}\Delta W_t^h,
		\end{cases}
	\end{equation}
	where $(V_t^{N})^{+}=\text{max}(V_t^N,0),$ $\Delta W_t^h=W_{t+h}-W_t$, $\Delta W_t^{\perp,h}=W^{\perp}_{t+h}-W^{\perp}_t$ and $V^N_{t+h}=\sum_{i=1}^Nw_iV^{(i)}_{t+h}$. Here, $W_t$ and $W^{\perp}_t$ are uncorrelated Brownian motions. Hence, the solution of $V^{(i)}_{t+h}$ can be solved by an $N$-dimensional linear system. 
	
	In our numerical experiments, we price out-of-money and at-the-money European call options 
	under the lifted Heston model. 
	The parameter values 
	are taken from \citet{bayer2024efficient}: 
	$\lambda=0.3,\,\nu=0.3,\,\theta=0.02,\,V^N_0=0.02,\,\rho=-0.7,\,S_0=1,\,[x_1,x_2,x_3]=[0.63781,9.6554,681.37]\text{ and }[w_1,w_2,w_3]=[0.66909,3.3694,184.50].$ 
	The initial value of $v^i_0$ is calculated as $\frac{V_0^N}{w_i}$. The benchmark option 
	values for $K = 0.6$ and $K = 1$, obtained from the Fourier-cosine method, are 0.40291 and 
	0.056555, respectively. As revealed from the results for the comparison of biases among various simulation schemes in Table~\ref{t:otm_lifted_bias}, 
	for the same time step $h$, the biases for the second-order operator splitting scheme and weak simulation 
	scheme in \citet{bayer2024efficient} are very small (order of $10^{-5}$) and do not have too much 
	difference. 
	The bias of the drift-implicit Euler scheme appears 
	to increase when the time step $h$ is less than $\frac{1}{x_3}$ 
	(around $\frac1{650}$), and it decreases beyond that threshold. A similar phenomenon is 
	also observed in \citet{bayer2024efficient}. Notably, the bias for out-of-money call 
	options using the drift-implicit Euler scheme is quite large when compared with those of the second-order
	operator splitting scheme and weak simulation scheme.
	The CPU times reported in Table~\ref{t:cpu_lifted_bias} indicate that 
	the second-order operator splitting scheme is more efficient than the weak simulation method. For the same 
	time step, the CPU time required for the drift-implicit Euler scheme may be about one third that of the 
	second-order operator splitting scheme. However, biases of the drift-implicit Euler scheme are much more significant. 
	
	In \Cref{fig:rms_lifted_heston}, we show the plots of RMSE against CPU time 
	for pricing European call options under the lifted Heston model. We find that the second-order operator splitting scheme outperforms the weak 
	simulation scheme for pricing in-the-money and at-the-money call options. Surprisingly, 
	the RMSE for the drift-implicit Euler scheme is seen to increase with CPU time before 
	eventually decrease.
	In conclusion, the second-order operator splitting scheme is the most efficient simulation method for pricing European call options under the lifted Heston model.
	 
    

\begin{table}[H]
\centering
\resizebox{0.55\textwidth}{!}{%
\begin{tabular}{lcccc}
\toprule
 & \textbf{Splitting} & \textbf{Weak} & & \textbf{Drift-Implicit Euler} \\
\cmidrule(lr){2-3} \cmidrule(lr){5-5}
$h$ & Bias & Bias & $h$ & Bias \\
\midrule
1/5     & 0.95  & 0.69 & 1/400   & 248.92 \\
1/10    & 0.61  & 0.97 & 1/650   & 256.84 \\
1/500   & -0.55 & 1.32 & 1/5000  & 111.48 \\
1/800   & -0.21 & 0.90 & 1/10000 & 63.79 \\
\bottomrule
\end{tabular}%
}
\caption{Bias ($\text{in units of } 10^{-5}$) for pricing out-of-the-money European call option ($K = 1.4$) under the lifted Heston model. `Splitting' refers to the second-order operator splitting scheme, and `Weak' refers to the weak simulation scheme in \citet{bayer2024efficient}.}
\label{t:otm_lifted_bias}
\end{table}

\begin{table}[H]
\centering
\resizebox{0.55\textwidth}{!}{%
\begin{tabular}{lcccc}
\toprule
 & \textbf{Splitting} & \textbf{Weak} & & \textbf{Drift-Implicit Euler} \\
\cmidrule(lr){2-3} \cmidrule(lr){5-5}
$h$ & CPU Time & CPU Time & $h$ & CPU Time \\
\midrule
1       & 0.006 & 0.011 & 1       & 0.004  \\
1/5     & 0.020 & 0.044 & 1/400   & 0.557  \\
1/10    & 0.037 & 0.084 & 1/650   & 0.989  \\
1/500   & 1.828 & 4.327 & 1/5000  & 7.614  \\
1/800   & 3.225 & 6.920 & 1/10000 & 15.222 \\
\bottomrule
\end{tabular}%
}
\caption{Comparison of CPU Time (s) among various simulation schemes for computing the at-the-money European call option under the lifted Heston model. `Splitting' refers to the second-order operator splitting scheme, and `Weak' refers to the weak simulation scheme in \citet{bayer2024efficient}.}
\label{t:cpu_lifted_bias}
\end{table}

    

\begin{figure}[H]
    \centering
    \begin{minipage}[t]{0.47\textwidth}
        \centering
        {\small (a) strike $K = 0.6$: in-the-money} \vspace{0.5em} \\ 
        \includegraphics[width=\textwidth]{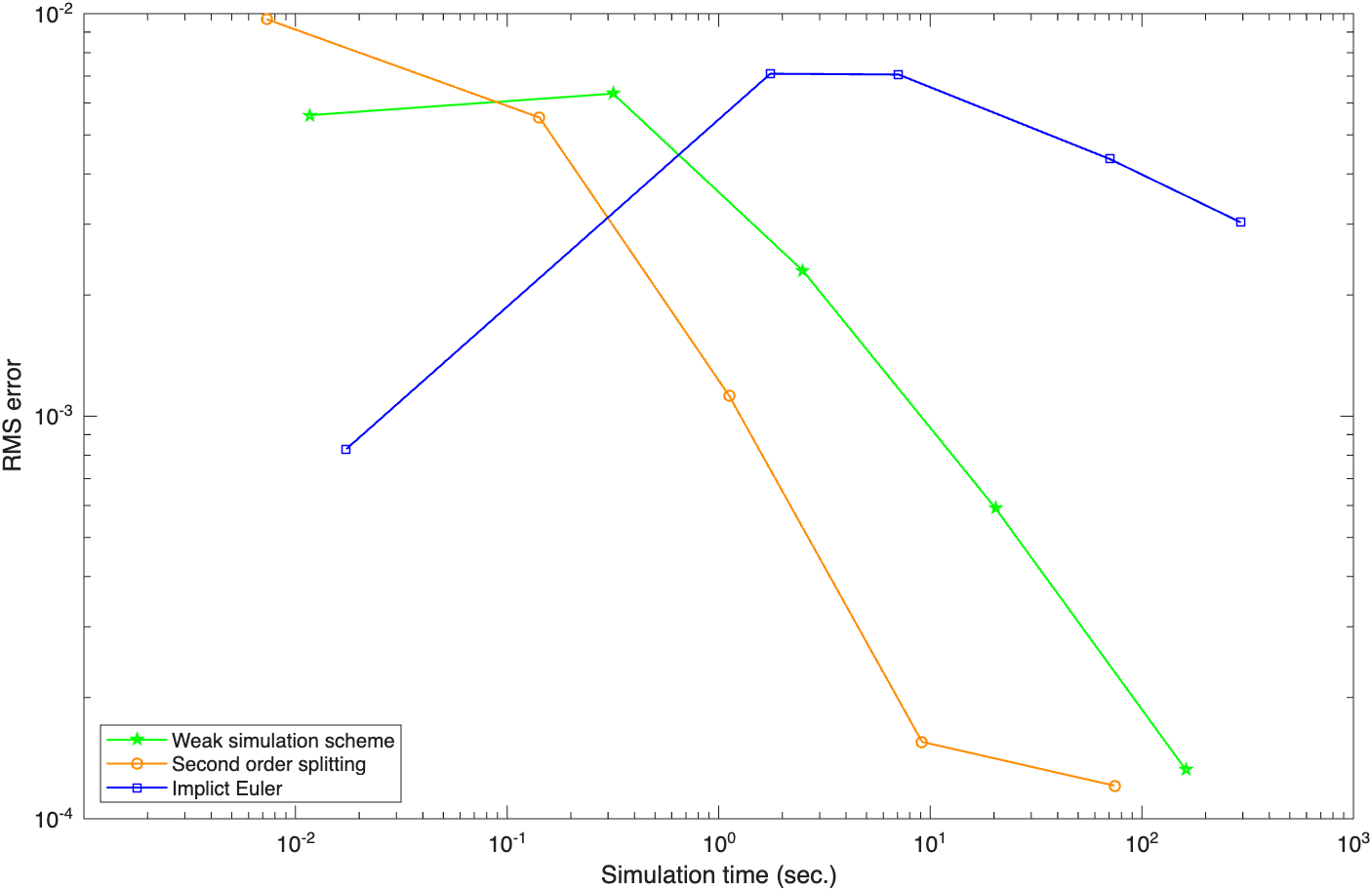}
    \end{minipage}
    \hfill 
    \begin{minipage}[t]{0.47\textwidth}
        \centering
        {\small (b) strike $K = 1$: at-the-money} \vspace{0.5em} \\ 
        \includegraphics[width=\textwidth]{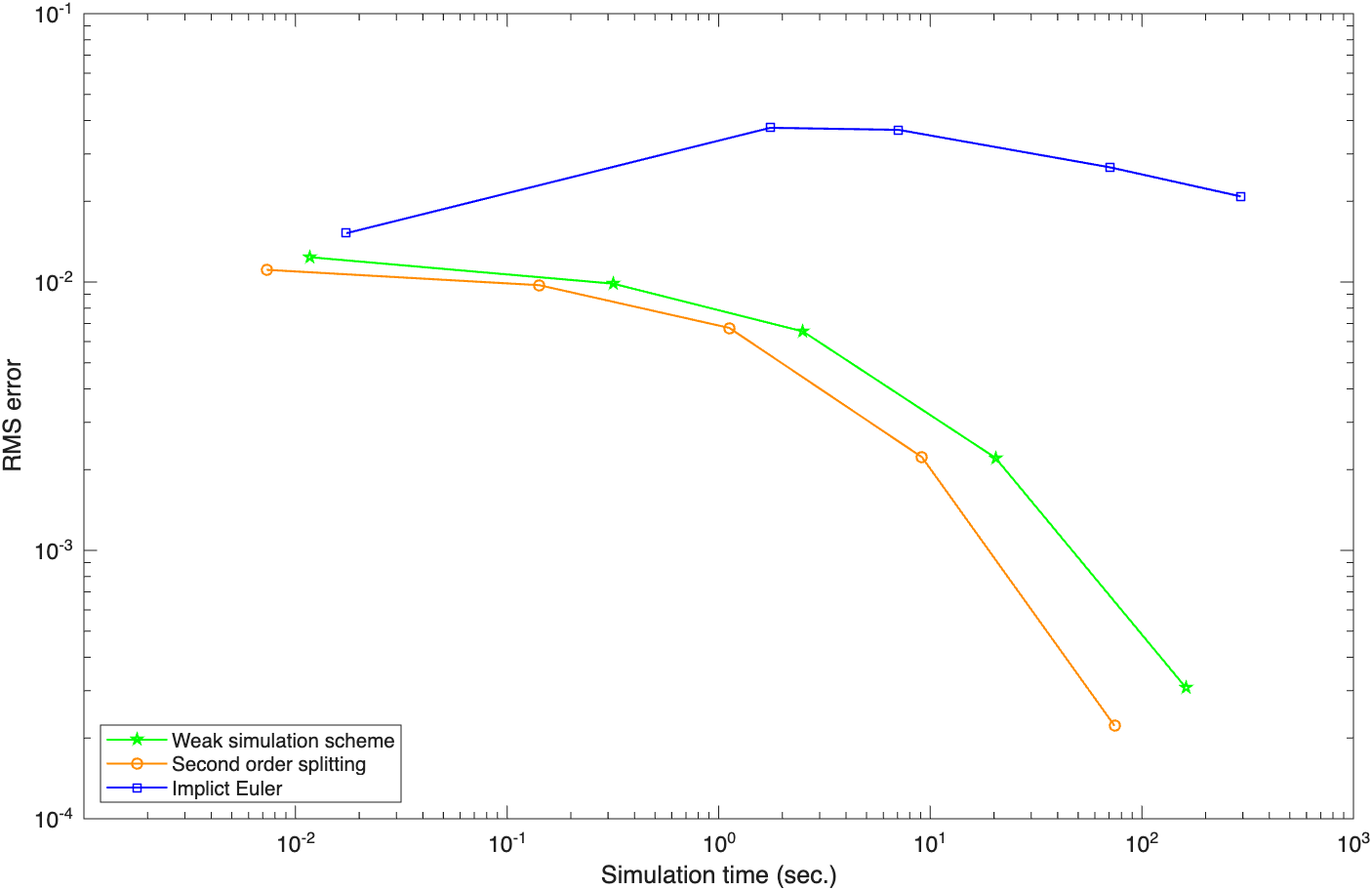}
    \end{minipage}
    
    \caption{Plots of RMSE versus CPU time for European call options under the lifted Heston model with strike prices $K=0.6$ (in-the-money) and $K=1$ (at-the-money). The sample sizes and time steps for the second-order splitting scheme and weak simulation scheme are $(10^5,1)$, $(2*10^5,1/16)$, $(4*10^5,1/64)$, $(8*10^5,1/256)$ and $(16*10^5,1/1024)$, and those of the implicit Euler scheme are $(10^5,1/10)$, $(2*10^5,1/500)$, $(4*10^5,1/1000)$, $(8*10^5,1/5000)$ and $(16*10^5,1/10000)$.}
    \label{fig:rms_lifted_heston}
\end{figure}

	\subsection{Barndorff-Nielsen and Shephard model}
	\label{subsection:BNS}
	\noindent
	In our numerical experiments for the BN-S model, we price the VIX call options and Asian call options
	under the D-OU type BN-S models (Gamma-OU and TS-OU) and OU-D type BN-S models (OU-Gamma and OU-TS). 
	To examine the tradeoff between RMSE and CPU time, we conducted simulation runs using the following configurations: ($N=10^4,\,h=1$), ($N=4*10^4,\,h=1/2$), ($N=16*10^4,\,h=1/4$), ($N=64*10^4,\,h=1/8$) and ($N=256*10^4,\,h=1/16$). These configurations yield 5 values for each of RMSE and CPU time, which are used to generate the 5 points in the corresponding figures.
	The formula for calculating the percentage relative error is given by
	\begin{equation}
		\text{percentage relative error}=\frac{\text{simulation value}-\text{benchmark value}}{\text{benchmark value}} \times 100\%.
	\end{equation} 
	
	 To implement the second-order operator splitting scheme, 
	 it is necessary to have a simulation scheme for the BDLP $Z_t$, 
	 which is relatively straightforward for the OU-D process. This can be easily achieved by simulating 
	 a gamma random variable for the OU Gamma process. For the D-OU processes, 
	 we need the Lévy density of the Lévy process $Z$ to simulate the BDLP. 
	 According to \citet{barndorff-nielsen2001non}, we can calculate the Lévy density of the BDLP 
	 by Eq.~\eqref{eq:levy_tran}. For convenience, we provide the L\'{e}vy densities of the BDLPs for the four BN-S model specifications in the table below.

\begin{table}[H]
    \centering
    \renewcommand{\arraystretch}{0.9} 
    \resizebox{0.8\textwidth}{!}{%
    \begin{tabular}{ll}
    \toprule    
    \textbf{BN-S model specification} & \textbf{BDLP description} \\
    \midrule
    OU - Gamma & $\Gamma(\alpha t , \beta)$ \\
    Gamma - OU & $L_t = \sum_{n=1}^{N_t} x_n, \quad x_n \sim \exp(\beta), \quad N_t \sim \text{Poisson}(\lambda \alpha t)$ \\
    OU - TS & $\mathcal{TS}(\alpha, \beta, \theta t)$ \\
    TS - OU & $\begin{aligned} 
                 &\mathcal{TS}(\alpha, \beta, \lambda \alpha \theta t) + \sum_{n=1}^{\tilde{N}_t}\tilde{x}_n, \\
                 &\tilde{N}_t \sim \text{Poisson}(\lambda \theta \beta^{\alpha}\Gamma(1-\alpha)t), \quad \tilde{x}_n \sim \Gamma(1-\alpha, \beta)
               \end{aligned}$ \\
    \bottomrule 
    \end{tabular}%
    }
\end{table}


\subsubsection{VIX call option}
\noindent For pricing of VIX call options, we use the Fourier-cosine method to obtain the benchmark value, using the knowledge of the characteristic function. The implementation details of the Fourier-cosine method are shown in Proposition \ref{prop:vix_cf}. Here, the model parameters are set as follows: 
$\lambda=0.5, \rho=-10$, risk-free rate $r=0.03$, strike price $K=0.15$, initial volatility $V_0=0.04$, maturity $T=1$. 
    The other parameters in the BDLPs are chosen as follows: 
    (a) OU-Gamma: $(\alpha, \beta) = (2, 2)$; 
    (b) Gamma-OU: $(\alpha, \beta) = (2, 2)$; 
    (c) OU-TS: $(\alpha, \beta, \theta) = (0.5, 2, 0.25)$; 
    (d) TS-OU: $(\alpha, \beta, \theta) = (0.25, 2, 0.25)$.

In the simulation calculations of the moments, we choose not to use the standard 
centered difference scheme for approximating the successive order of derivatives 
of the characteristic function, since numerical differentiation is prone to significant numerical instability and 
precision loss. Instead, we adopt the analytical algorithm proposed 
in \cite{kyriakou2023unified} by numerical inversion of an adaptively modified MGF. 
The core of this methodology lies in the direct application of numerical 
integration of the Laplace transform of the stochastic process, 
thereby circumventing the laborious procedure of performing high-order 
analytical differentiation of the MGF. 
In practical implementation, by recursively adjusting the scaling parameter $\alpha_n$ and 
the pre-determined error control parameter $\gamma$, the computational error 
can be well maintained at an extremely high precision 
level of approximately $10^{-11}$. Specifically, the moments $\mu_n$ are sequentially evaluated via 
 the numerical inversion of the Laplace transform, while $\alpha_n$ is dynamically updated utilizing the preceding 
 lower-order moments through $\alpha_n = (n-1)\mu_{n-2}/\mu_{n-1}$ for $n \ge 3$.

In all cases, the second-order operator splitting scheme demonstrates 
better performance. As revealed from Table~\ref{tab:vix_pricing_all},
the second-order operator splitting scheme mitigates the systematic discretization errors inherent in the Euler scheme, particularly on coarse grids. For instance, in the OU-Gamma process at $h=1$ (Table~\hyperref[tab:ou_gamma]{\ref{tab:vix_pricing_all}(a)}), the Euler method yields a large percentage relative error of 0.7657\%, whereas the second-order operator splitting scheme yields a much lower level of -0.0294\%. Importantly, such accuracy improvement does not incur higher computational cost, with both methods requiring nearly identical CPU time (0.0031s versus 0.0030s). This demonstrates that the second-order operator splitting scheme accurately captures the intrinsic dynamics of the process without relying on the computationally intensive fine grids demanded by the Euler method.

Furthermore, the second-order operator splitting scheme demonstrates better efficiency and robustness when compared to the simulation schemes in \citet{Dassio2019Gamma,Dassio2021TS} and \citeauthor{sabino2020gamma} (\citeyear{sabino2020gamma}, \citeyear{sabino2020tempered}). While the scheme of Qu and his co-authors achieves comparably low relative errors, the second-order operator splitting scheme is shown to be computationally efficient. In the OU-TS process at $h=1/16$, it saves approximately 14\% in execution time (0.121s versus 0.141s). 
Additionally, the schemes of Sabino and his co-authors suffer significant precision degradation on coarse grids. For instance, in the TS-OU process at $h=1$, their scheme exhibits a relative error spike of $-0.3805\%$. In contrast, the second-order operator splitting scheme maintains a stable, low percentage relative error at low level of $-0.0134\%$.
This combination of computational scalability and coarse-grid stability establish that the second-order operator splitting scheme is a highly reliable and efficient choice for pricing VIX call options under the BN-S model.

\begin{table}[H] 
\centering
\footnotesize 
\setlength{\tabcolsep}{3.5pt} 
\renewcommand{\arraystretch}{0.95} 

\phantomsection
\centerline{(a) OU-Gamma}\label{tab:ou_gamma}
\vspace{3pt}
\begin{tabular}{lccccccccc}
\toprule
 & \multicolumn{3}{c}{\textbf{Splitting}} & \multicolumn{3}{c}{\textbf{Qu}} & \multicolumn{3}{c}{\textbf{Euler}} \\
\cmidrule(lr){2-4} \cmidrule(lr){5-7} \cmidrule(lr){8-10}
$h$ & RelErr & StDev & Time & RelErr & StDev & Time & RelErr & StDev & Time \\
\midrule
1    & -0.0294 & 1.84 & 0.31 & 0.0004  & 2.04 & 0.70 & 0.7657 & 2.79 & 0.30 \\
1/2  & -0.0059 & 2.33 & 0.33 & -0.0004 & 2.02 & 0.68 & 0.3187 & 2.52 & 0.32 \\
1/4  & -0.0022 & 2.38 & 1.83 & 0.0003  & 1.92 & 2.32 & 0.1469 & 2.38 & 1.82 \\
1/8  & -0.0003 & 2.33 & 3.01 & 0.0005  & 2.15 & 3.95 & 0.0705 & 2.09 & 3.02 \\
\bottomrule
\end{tabular}

\vspace{1.5em} 

\phantomsection
\centerline{(b) Gamma-OU}\label{tab:gamma_ou}
\vspace{3pt}
\resizebox{\textwidth}{!}{%
\begin{tabular}{lccccccccccccccc}
\toprule
 & \multicolumn{3}{c}{\textbf{Splitting}} & \multicolumn{3}{c}{\textbf{Qu}} & \multicolumn{3}{c}{\textbf{Sabino 1}} & \multicolumn{3}{c}{\textbf{Sabino 4}} & \multicolumn{3}{c}{\textbf{Euler}} \\
\cmidrule(lr){2-4} \cmidrule(lr){5-7} \cmidrule(lr){8-10} \cmidrule(lr){11-13} \cmidrule(lr){14-16}
$h$ & RelErr & StDev & Time & RelErr & StDev & Time & RelErr & StDev & Time & RelErr & StDev & Time & RelErr & StDev & Time \\
\midrule
1    & -0.0009 & 0.646 & 0.63 & -0.0001 & 0.604 & 0.66 & 0.0001 & 0.666 & 1.08 & 0.0000 & 0.555 & 2.17 & 0.0252 & 0.793 & 0.61 \\
1/2  & -0.0002 & 0.503 & 0.86 & -0.0001 & 0.635 & 0.81 & -0.0000 & 0.674 & 1.60 & 0.0001 & 0.630 & 3.10 & 0.0105 & 0.517 & 0.85 \\
1/4  & -0.0001 & 0.629 & 1.10 & -0.0000 & 0.614 & 1.01 & -0.0000 & 0.491 & 2.40 & 0.0000 & 0.644 & 4.61 & 0.0048 & 0.645 & 1.09 \\
1/8  & -0.0000 & 0.632 & 1.55 & 0.0000 & 0.614 & 1.51 & -0.0000 & 0.616 & 3.92 & 0.0000 & 0.551 & 7.26 & 0.0024 & 0.587 & 1.53 \\
\bottomrule
\end{tabular}%
}

\vspace{1.5em}

\phantomsection
\centerline{(c) OU-TS}\label{tab:ou_ts}
\vspace{3pt}
\resizebox{\textwidth}{!}{%
\begin{tabular}{lcccccccccccc}
\toprule
 & \multicolumn{3}{c}{\textbf{Splitting}} & \multicolumn{3}{c}{\textbf{Qu}} & \multicolumn{3}{c}{\textbf{Sabino}} & \multicolumn{3}{c}{\textbf{Euler}} \\
\cmidrule(lr){2-4} \cmidrule(lr){5-7} \cmidrule(lr){8-10} \cmidrule(lr){11-13}
$h$ & RelErr & StDev & Time & RelErr & StDev & Time & RelErr & StDev & Time & RelErr & StDev & Time \\
\midrule
1    & -0.0423 & 2.112 & 2.552 & -0.0017 & 2.079 & 3.329 & 0.0918  & 1.776 & 3.485  & 1.1104 & 2.302 & 2.531 \\
1/2  & -0.0101 & 2.142 & 2.726 & -0.0015 & 1.898 & 3.131 & 0.0231  & 2.119 & 3.206  & 0.4584 & 2.383 & 2.719 \\
1/4  & -0.0034 & 1.589 & 3.949 & -0.0004 & 1.713 & 4.487 & 0.0056  & 1.830 & 4.495  & 0.2099 & 1.632 & 3.912 \\
1/8  & -0.0009 & 1.669 & 6.365 & 0.0012  & 1.848 & 7.394 & -0.0004 & 1.584 & 7.256  & 0.1012 & 2.162 & 6.354 \\
\bottomrule
\end{tabular}%
}

\vspace{1.5em}

\phantomsection
\centerline{(d) TS-OU}\label{tab:ts_ou}
\vspace{3pt}
\resizebox{\textwidth}{!}{%
\begin{tabular}{lcccccccccccc}
\toprule
 & \multicolumn{3}{c}{\textbf{Splitting}} & \multicolumn{3}{c}{\textbf{Qu}} & \multicolumn{3}{c}{\textbf{Sabino}} & \multicolumn{3}{c}{\textbf{Euler}} \\
\cmidrule(lr){2-4} \cmidrule(lr){5-7} \cmidrule(lr){8-10} \cmidrule(lr){11-13}
$h$ & RelErr & StDev & Time & RelErr & StDev & Time & RelErr & StDev & Time & RelErr & StDev & Time \\
\midrule
1    & -0.0134 & 1.476 & 1.238  & 0.0102  & 1.833 & 1.317  & -0.3805 & 1.192 & 1.150  & 0.2569 & 1.999 & 1.138 \\
1/2  & -0.0033 & 1.575 & 1.882  & 0.0380  & 1.508 & 2.146  & -0.2502 & 1.207 & 1.915  & 0.1078 & 1.995 & 1.892 \\
1/4  & -0.0033 & 1.667 & 3.378  & -0.0129 & 1.448 & 3.862  & -0.1492 & 1.520 & 3.427  & 0.0453 & 1.663 & 3.391 \\
1/8  & 0.0008  & 1.419 & 6.423  & -0.0194 & 1.439 & 7.379  & -0.0792 & 1.769 & 6.401  & 0.0288 & 1.967 & 6.398 \\
\bottomrule
\end{tabular}%
}

\vspace{0.8em} 
\caption{
    Comparison of percentage relative error (RelErr in $\%$), standard deviation ($\text{StDev in units of } 10^{-4}$), 
    and CPU time ($\text{Time in units of } 10^{-2}$ s) for pricing VIX call option among 4 specifications of the BN-S model: 
    (a) OU-Gamma, (b) Gamma-OU, (c) OU-TS, and (d) TS-OU. We compare the results from the  
    second-order operator splitting scheme (``Splitting'') and standard Euler discretization scheme 
    (``Euler'') against the exact simulation benchmarks denoted by ``Qu'' and ``Sabino''. Specifically, 
    these benchmarks represent the following methods: in (a), ``Qu'' refers to Algorithm~4.1 in 
    \cite{Dassio2019Gamma}; in (b), ``Qu'' is Algorithm~3.1 in \cite{Dassio2019Gamma}, 
    while ``Sabino 1'' and ``Sabino 4'' refer to Algorithm~1 and Algorithm~4 in \cite{sabino2020gamma}, 
    respectively; in (c), ``Qu'' corresponds to Algorithm~3.1 in \cite{Dassio2021TS} and 
    ``Sabino'' corresponds to Algorithm~2 in \cite{sabino2020tempered}; in (d), ``Qu'' 
    represents Algorithm~4.1 in \cite{Dassio2021TS} and ``Sabino'' represents Algorithm~1 
    in \cite{sabino2020tempered}. 
}
\label{tab:vix_pricing_all}
\end{table}

We further examine the convergence behavior of the second-order operator splitting scheme for the VIX call option using the numerical result from the Fourier-cosine algorithm as the benchmark value. 
For each time step size \(h\), we compute the absolute error $\left|\bar P_h-P_{\mathrm{bench}}\right|$,
where \(\bar P_h\) denotes the Monte Carlo price estimator obtained from a given simulation scheme and \(P_{\mathrm{bench}}\) is the Fourier-cosine benchmark value.


\begin{figure}[H]
    \centering
    \begin{minipage}[t]{0.4\textwidth} 
        \centering
        {\small (a) OU-Gamma} \vspace{0.5em} \\ 
        \includegraphics[width=\textwidth]{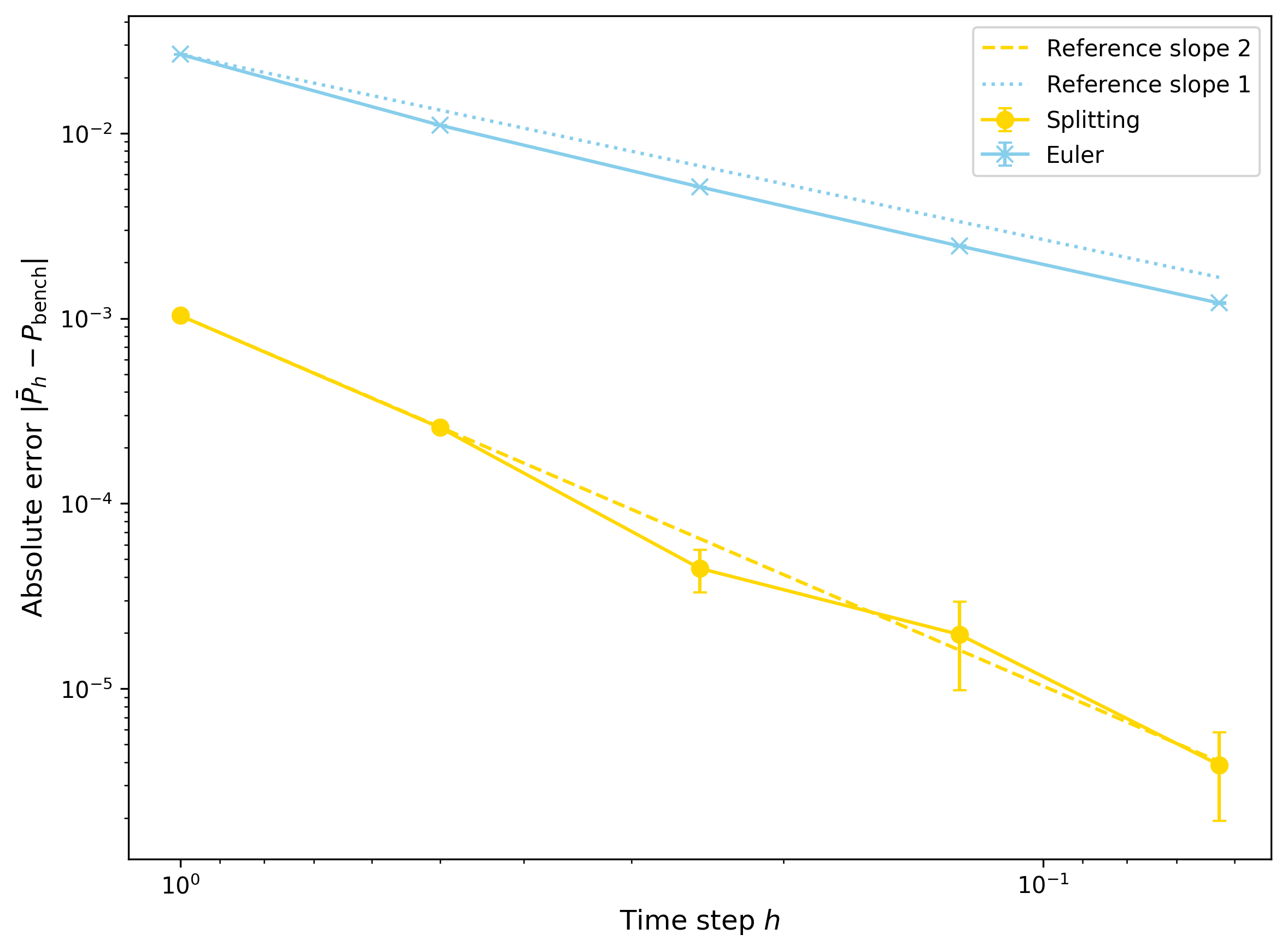}
        \label{fig:convergence_vix_ou_gamma}
    \end{minipage}
    \hfill
    \begin{minipage}[t]{0.4\textwidth}
        \centering
        {\small (b) Gamma-OU} \vspace{0.5em} \\
        \includegraphics[width=\textwidth]{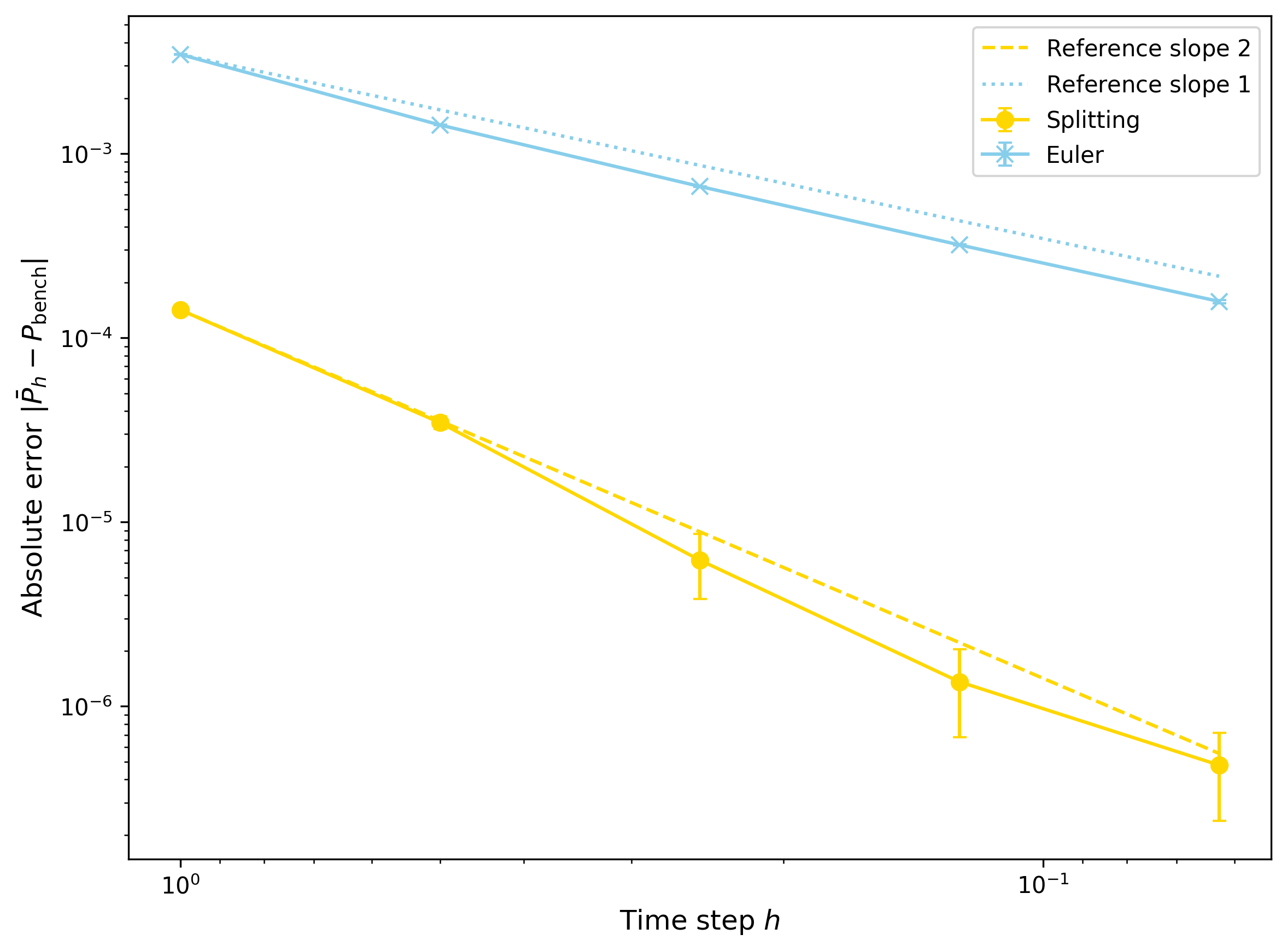}
        \label{fig:convergence_vix_gamma_ou}
    \end{minipage}

    \begin{minipage}[t]{0.4\textwidth}
        \centering
        {\small (c) OU-TS} \vspace{0.5em} \\
        \includegraphics[width=\textwidth]{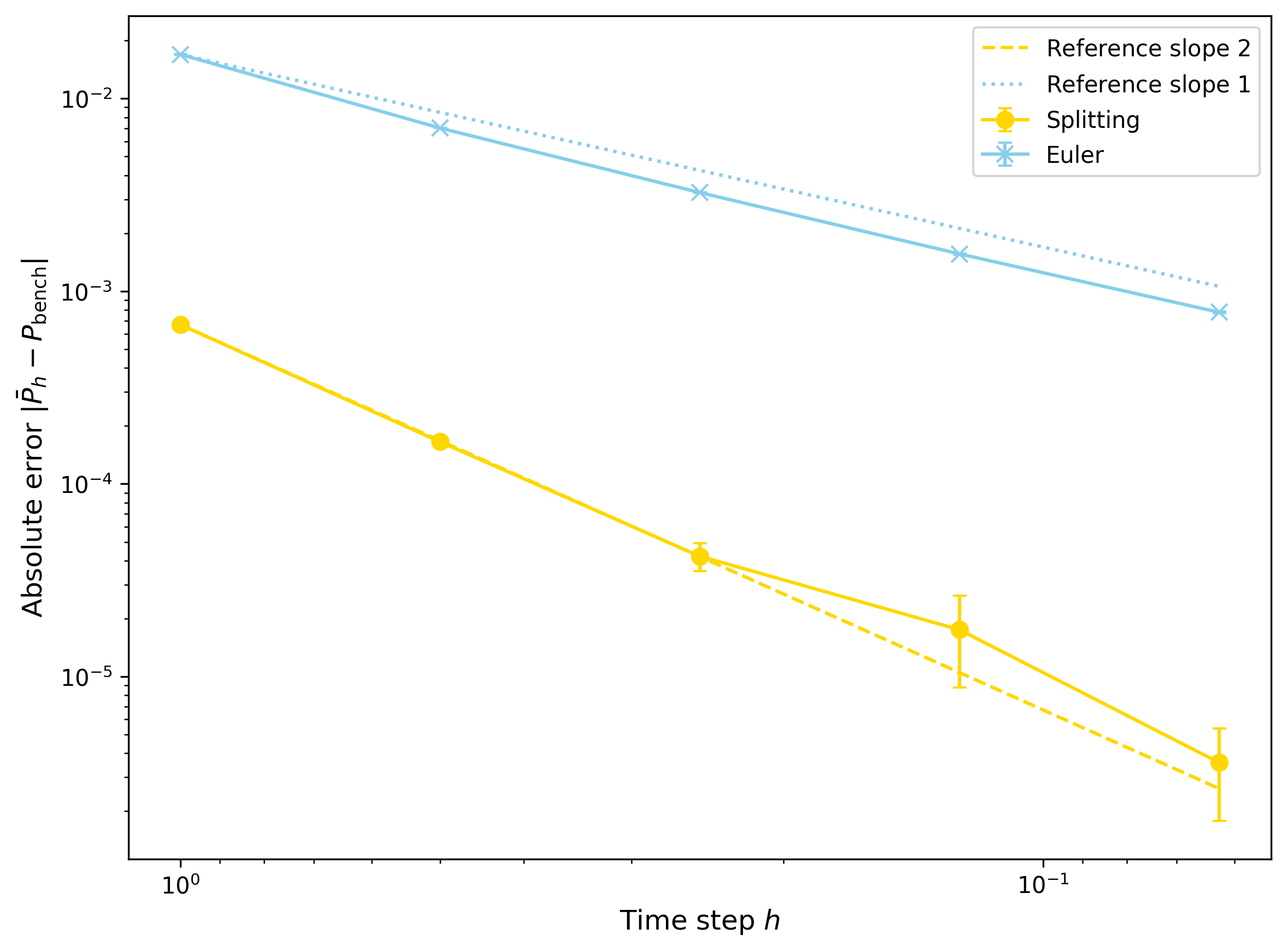}
        \label{fig:convergence_vix_ou_ts}
    \end{minipage}
    \hfill
    \begin{minipage}[t]{0.4\textwidth}
        \centering
        {\small (d) TS-OU} \vspace{0.5em} \\
        \includegraphics[width=\textwidth]{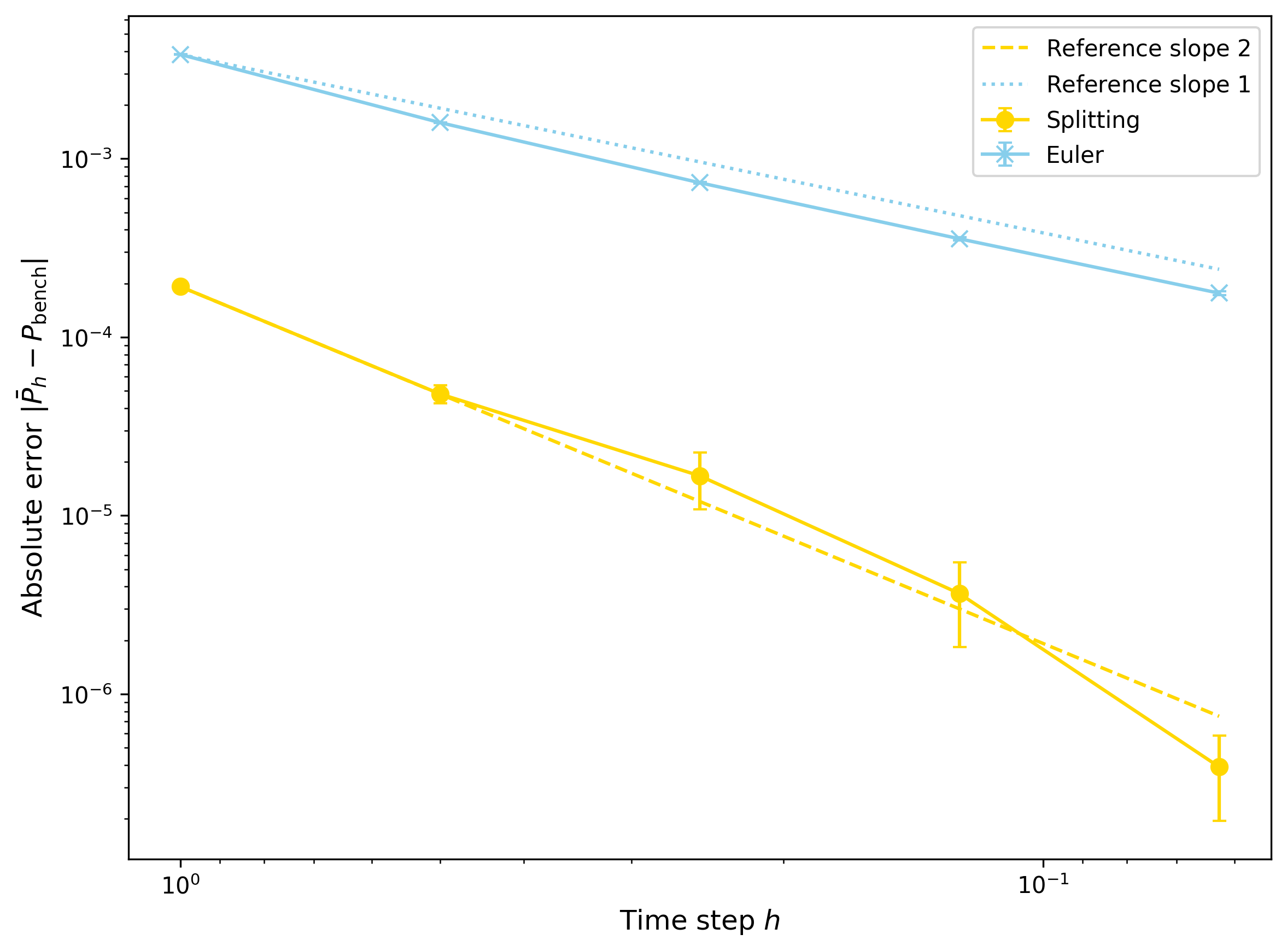}
        \label{fig:convergence_vix_ts_ou}
    \end{minipage}
    
	\caption{Test for order of convergence of operator splitting scheme and Euler scheme for pricing VIX call options under the BN-S model. We plot the absolute error \(|\bar P_h-P_{\mathrm{bench}}|\) against time step $h$, where \(P_{\mathrm{bench}}\) is computed by the Fourier-cosine method.}

	\label{fig:vix-weak-convergence}
\end{figure}

Suppose the weak error behaves asymptotically as
\[
\left|\bar{P}_h-P_{\mathrm{bench}}\right| \approx C h^p,
\]
exhibiting $p$th-order convergence, then its graph against \(h\) has slope \(p\) on a log--log scale. To facilitate a visual
assessment of the convergence order, we include reference lines proportional
to \(h^2\) and \(h\) for the operator splitting scheme and the Euler discretization scheme, respectively.
Each reference line is chosen to pass through the first point in the plot of the corresponding
method. Hence, the reference lines indicate the
decay rate of the absolute error that would be observed if the errors subsequently follows the predicted second-order
and first-order convergence, respectively.

Figure~\ref{fig:vix-weak-convergence} shows the log-log plots of the absolute error against the time step \(h\) for the VIX call option. The numerical plots reveal that the error of the operator splitting scheme exhibits approximately the second-order convergence, while the Euler scheme exhibits the first-order convergence. 

\subsubsection{Pricing Asian call options}


\noindent The exact simulation methods of \citet{Dassio2019Gamma,Dassio2021TS} and \citeauthor{sabino2020gamma} (\citeyear{sabino2020gamma}, \citeyear{sabino2020tempered}) focus on simulating the stochastic evolution 
of the spot volatility but fail to provide the exact conditional distribution of the integrated 
variance $I_t^h$. Therefore, they cannot be used to price Asian option. Consequently, we can only benchmark our second-order operator splitting scheme 
against the standard Euler discretization scheme for pricing discretely monitored European Asian options. 
The benchmark value is based on the numerical result obtained from the Euler scheme with a time step size of $h = \frac{1}{100}$ and $10^8$ simulation paths.
Here, the model parameters are set as follows: $M=4, \lambda=2.5, \rho=-5$, risk-free rate $r=0.02$, strike price $K=1$, initial volatility $V_0=0.04$, maturity $T=1$. The other parameters in BDLPs are calibrated as follows: 
    {(a) OU-Gamma:} $(\alpha, \beta) = (0.1, 15)$; 
    {(b) Gamma-OU:} $(\alpha, \beta) = (0.25, 15)$; 
    {(c) OU-TS:} $(\alpha, \beta, \theta) = (0.25, 15, 0.1)$; 
    {(d) TS-OU:} $(\alpha, \beta, \theta) = (0.25, 15, 0.1)$.

\begin{figure}[H]
    \centering
    \begin{minipage}[t]{0.4\textwidth} 
        \centering
        {\small (a) OU-Gamma} \vspace{0.5em} \\ 
        \includegraphics[width=\textwidth]{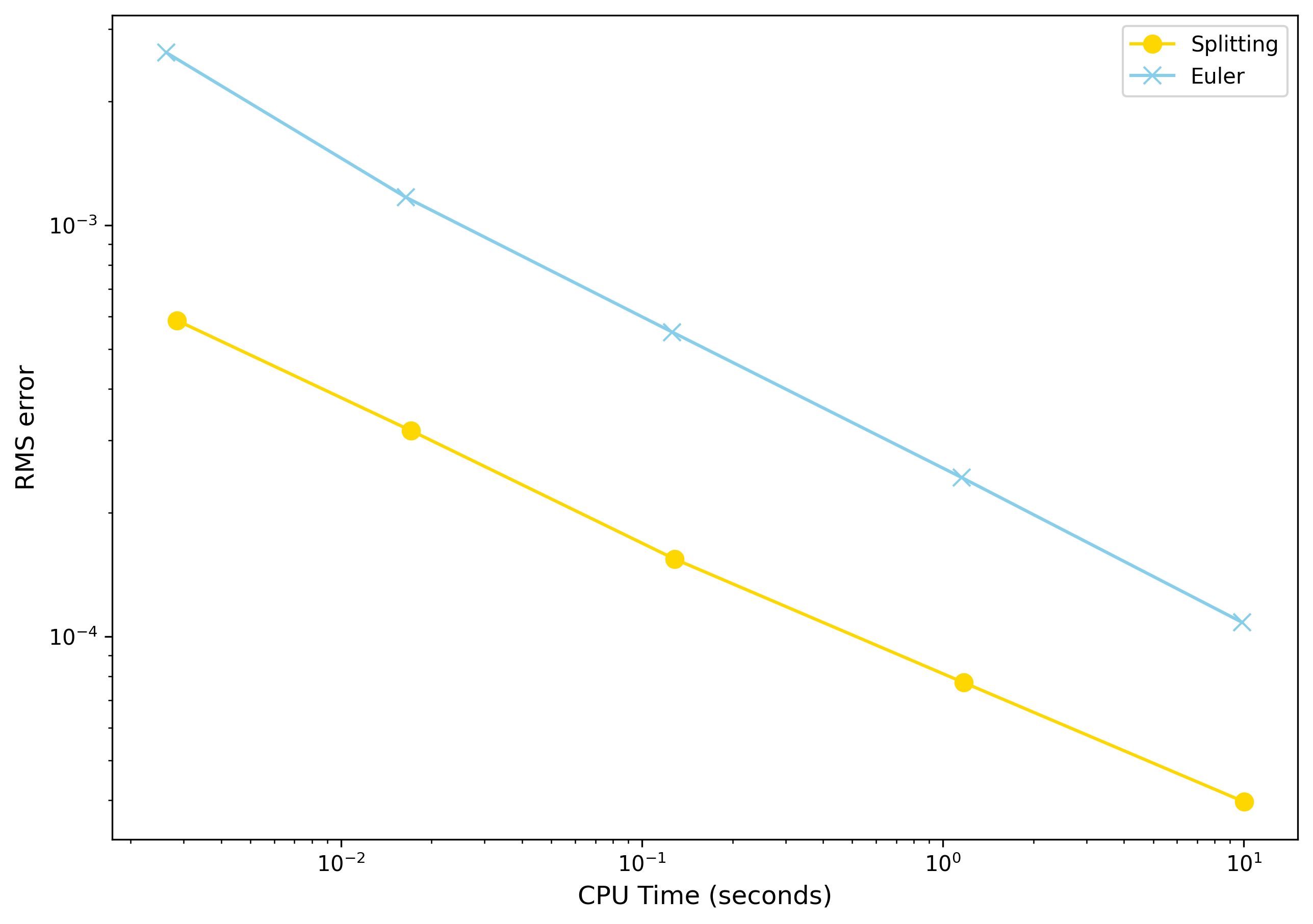}
        \label{fig:rmse_ou_gamma_asian}
    \end{minipage}
    \hfill
    \begin{minipage}[t]{0.4\textwidth}
        \centering
        {\small (b) Gamma-OU} \vspace{0.5em} \\
        \includegraphics[width=\textwidth]{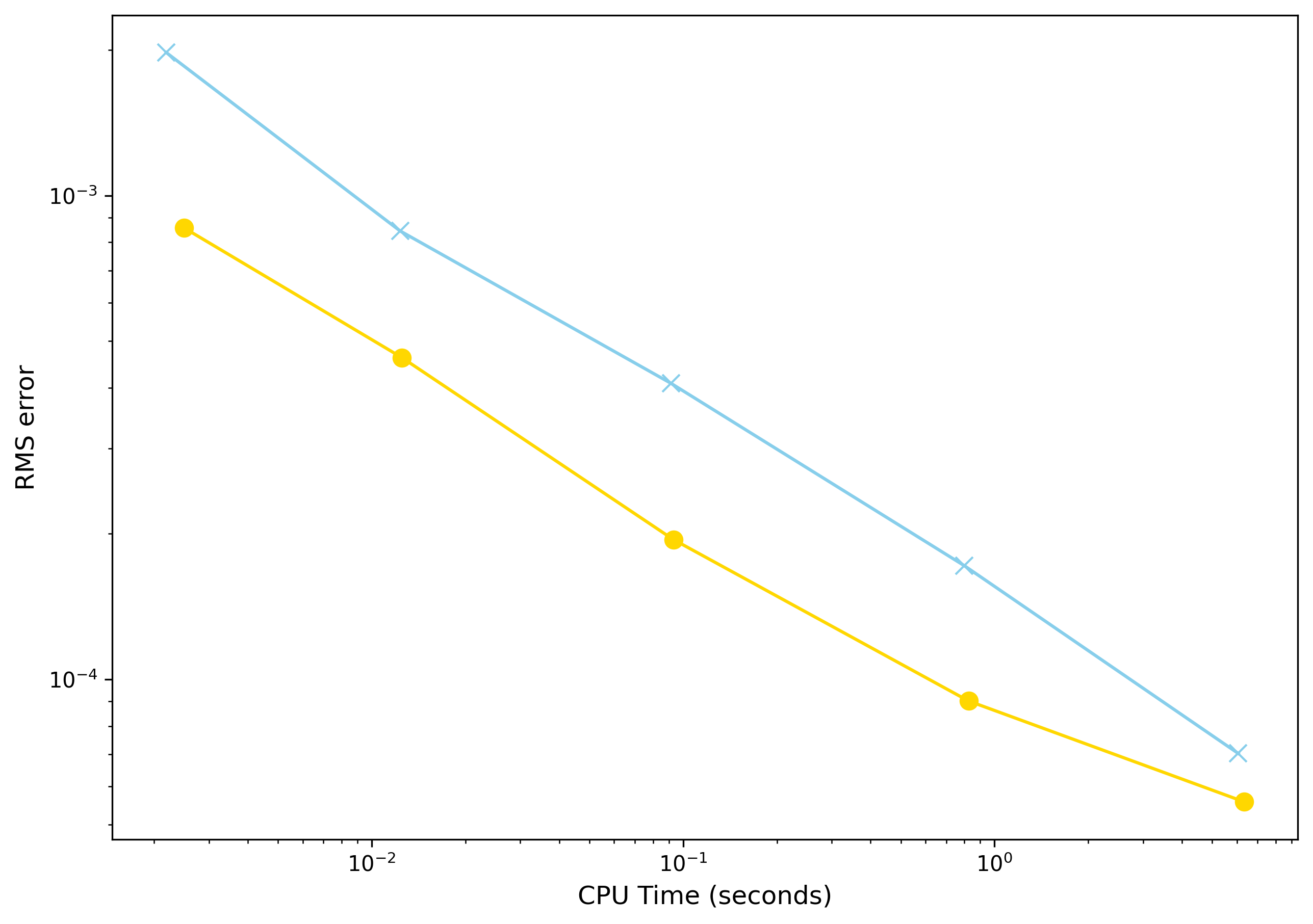}
        \label{fig:rmse_gamma_ou_asian}
    \end{minipage}

    \begin{minipage}[t]{0.4\textwidth}
        \centering
        {\small (c) OU-TS} \vspace{0.5em} \\
        \includegraphics[width=\textwidth]{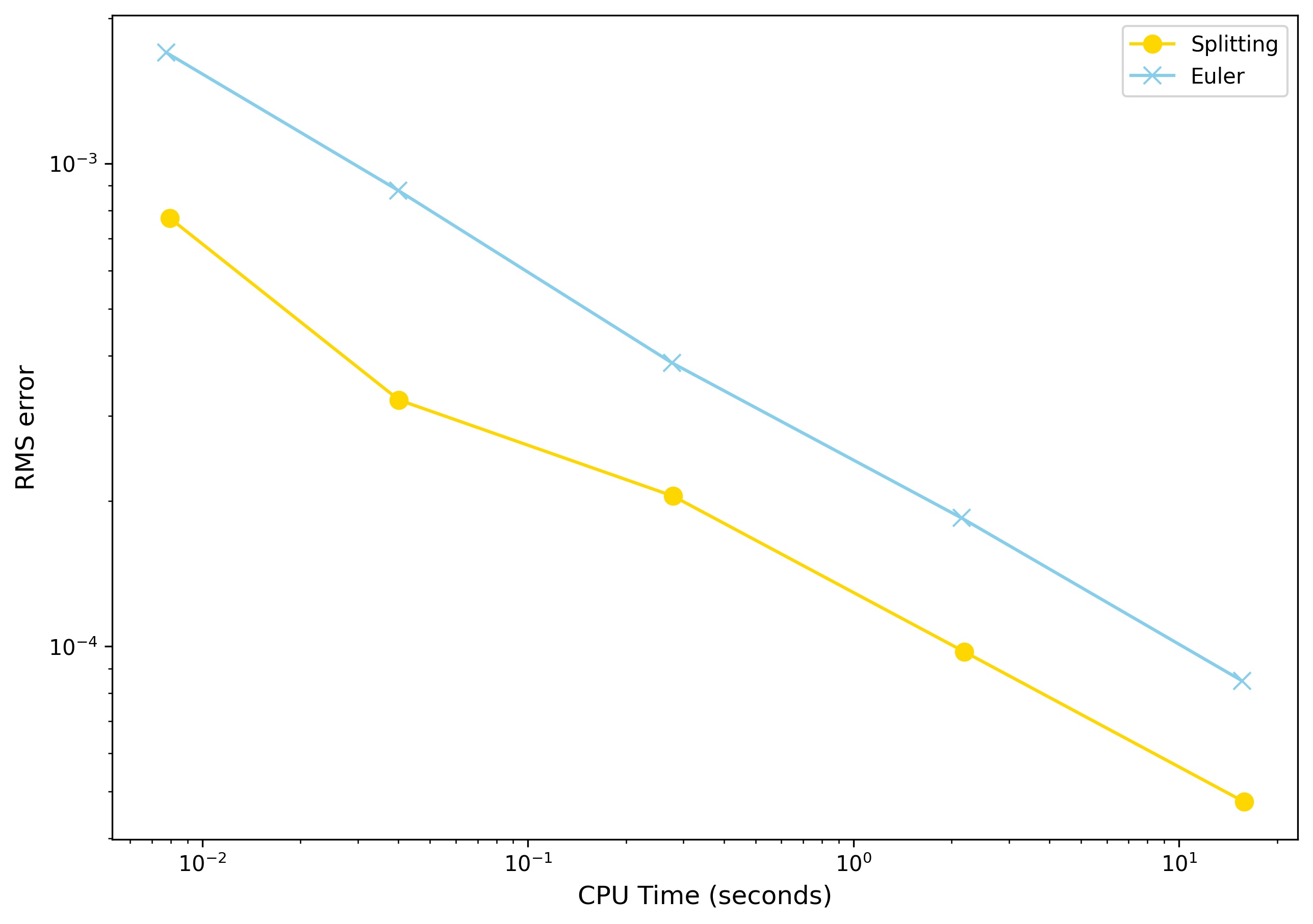}
        \label{fig:rmse_ou_ts_asian}
    \end{minipage}
    \hfill
    \begin{minipage}[t]{0.4\textwidth}
        \centering
        {\small (d) TS-OU} \vspace{0.5em} \\
        \includegraphics[width=\textwidth]{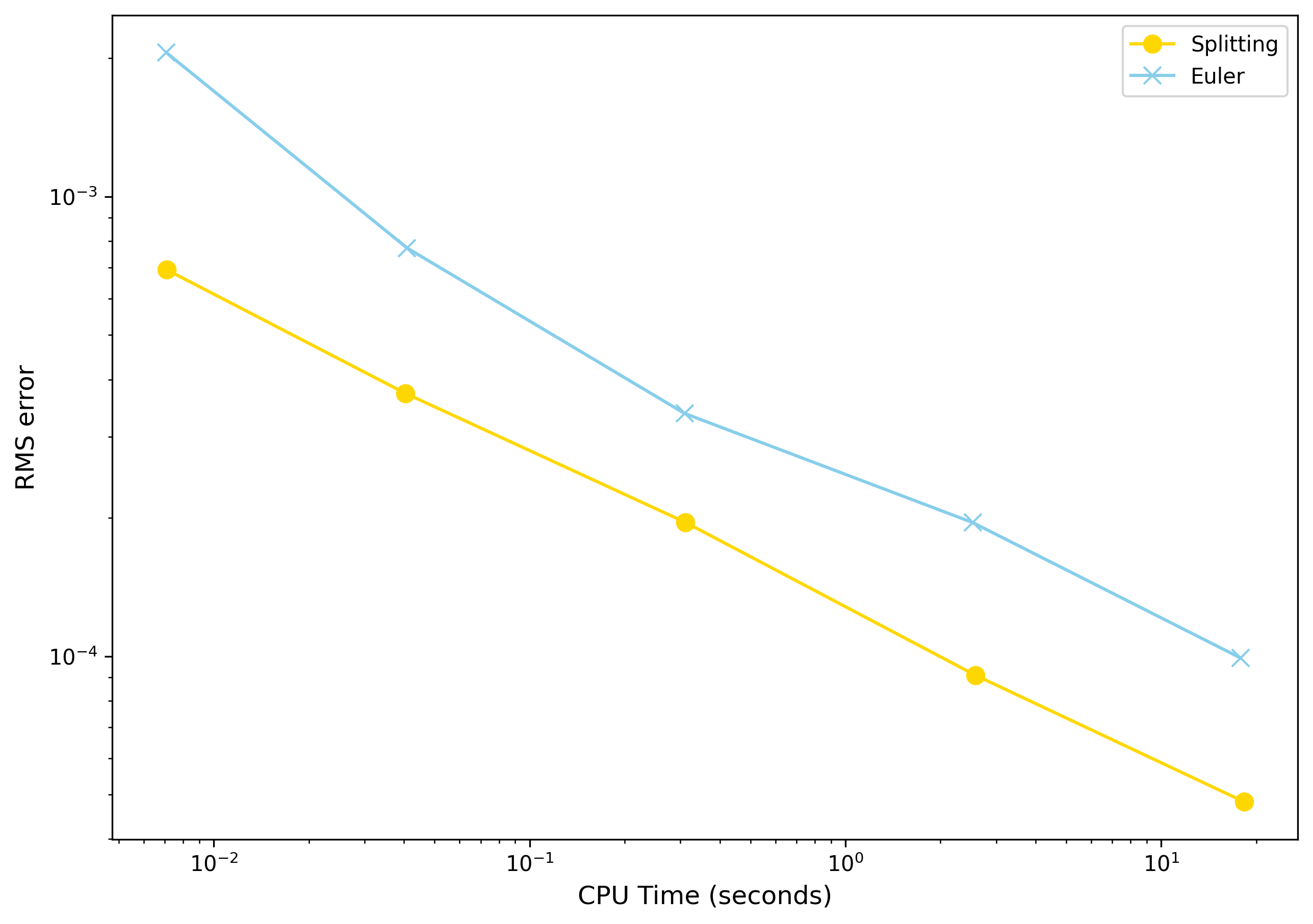}
        \label{fig:rmse_ts_ou_asian}
    \end{minipage}
    
    \caption{Plots of RMSE versus CPU time for the operator splitting scheme and Euler scheme under different specifications of the BN-S model for pricing European Asian call option: (a) OU-Gamma, (b) Gamma-OU, (c) OU-TS, and (d) TS-OU.}
    \label{fig:asian_convergence}
\end{figure}

As illustrated by the plots of RMSE versus CPU time for pricing European Asian call option under different specifications of the BN-S model in \Cref{fig:asian_convergence}, the second-order splitting scheme consistently yields smaller RMSE for the same CPU time. For instance, evaluating the OU-Gamma process at the coarsest time step ($h=1$), both methods require approximately equal amount of execution time (0.0021s). However, the second-order operator splitting scheme produces the RMSE of $5.87 \times 10^{-4}$, which is considerably lower than the RMSE of $2.63 \times 10^{-3}$ observed for the Euler scheme. This indicates that the second-order operator splitting 
scheme performs better in terms of accuracy-speed tradeoff.

We next examine the convergence behavior of the operator splitting scheme for
discretely monitored European Asian call options under the four BN-S specifications considered above. 
For such test of order of convergence, we use a fine-grid operator splitting scheme as the benchmark, 
which is different from the Euler benchmark used in the preceding RMSE comparison. 
Specifically, each monitoring
interval is divided into \(N_{\mathrm{bench}}=240\) substeps, corresponding to the
benchmark time step
$h_{\mathrm{bench}}
=
\frac{T}{M N_{\mathrm{bench}}}.$
To improve the robustness of the convergence analysis, the coarse-grid
approximation and the fine-grid benchmark are generated using the same
underlying random increments. The coarse and fine paths are therefore
constructed from consistent random sources, so that the difference between
their corresponding option prices mainly reflects the error of the
operator splitting approximation.


For a simulation method \(m\), we measure the error by
\[
\varepsilon_m(h)
=
\left|
\mathbb{E}\!\left[
G_h^{m}-G_{h_{\mathrm{bench}}}
\right]
\right|,
\]
where \(G_h^{m}\) denotes the discounted European Asian call option price generated by method \(m\)
with time step \(h\), and \(G_{h_{\mathrm{bench}}}\) denotes the option price generated
by the fine-grid operator splitting benchmark. The reference lines are
constructed in the same manner as in the VIX call option convergence test. 
Figure~\ref{fig:asian-coupled-convergence} shows the log-log plots of the absolute error
against the time step \(h\) for the discretely monitored European Asian call option under the four
BN-S specifications.
The numerical results show that the errors in the operator splitting scheme exhibit
approximately at the predicted second-order convergence, whereas the errors in the Euler scheme exhibit
a slower decay broadly consistent with first-order convergence.
This observed behavior is consistent with the second-order convergence of the operator splitting scheme established
in Section~\ref{sec:error-analysis}.
In conclusion, these results demonstrate that the second-order operator splitting scheme is a highly efficient and robust approach for simulating the complex heavy-tailed jump dynamics under the BN-S model when compared with the standard Euler scheme.

\begin{figure}[H]
    \centering
    \begin{minipage}[t]{0.4\textwidth} 
        \centering
        {\small (a) OU-Gamma} \vspace{0.5em} \\ 
        \includegraphics[width=\textwidth]{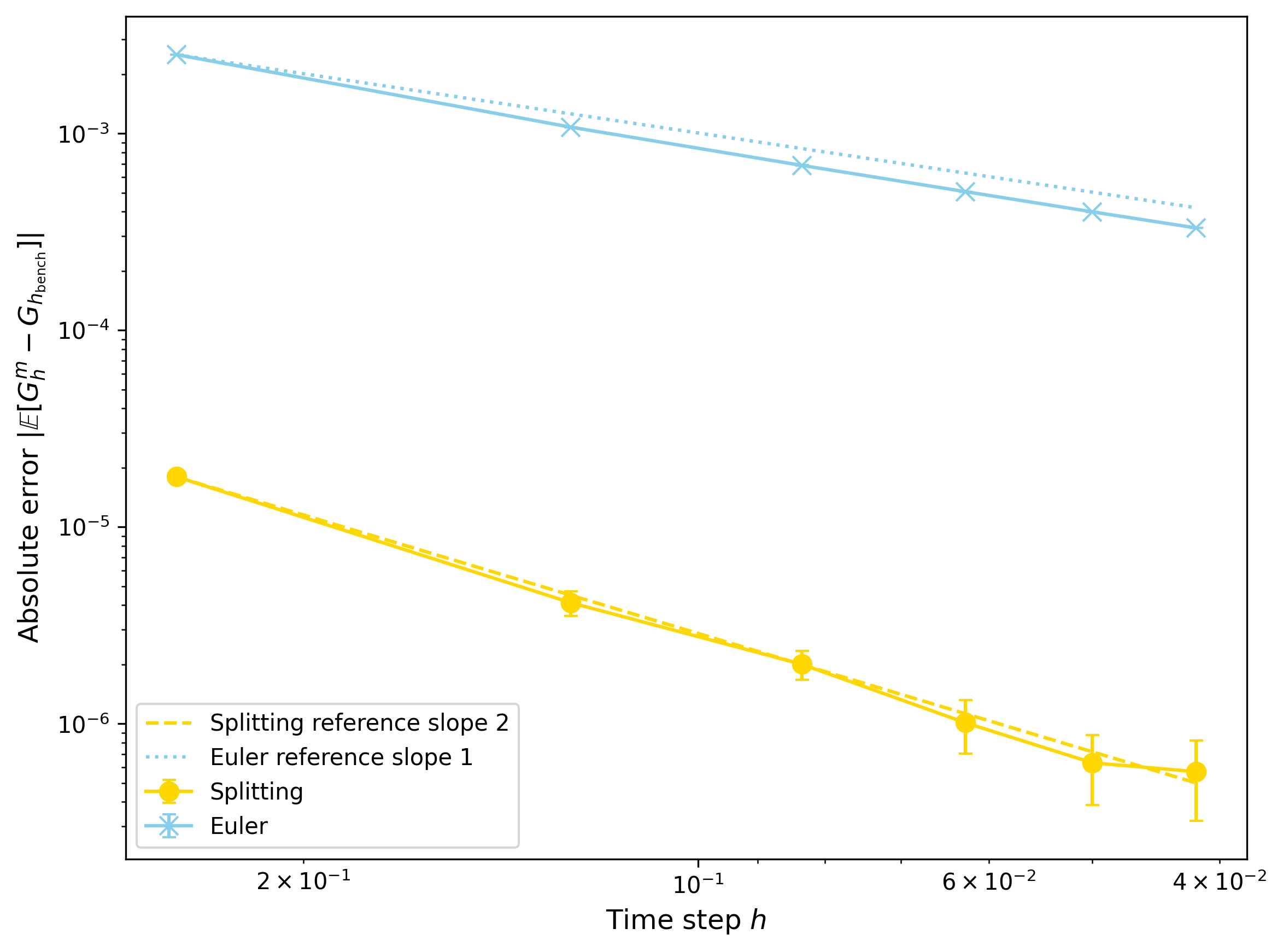}
        \label{fig:convergence_asian_ou_gamma}
    \end{minipage}
    \hfill
    \begin{minipage}[t]{0.4\textwidth}
        \centering
        {\small (b) Gamma-OU} \vspace{0.5em} \\
        \includegraphics[width=\textwidth]{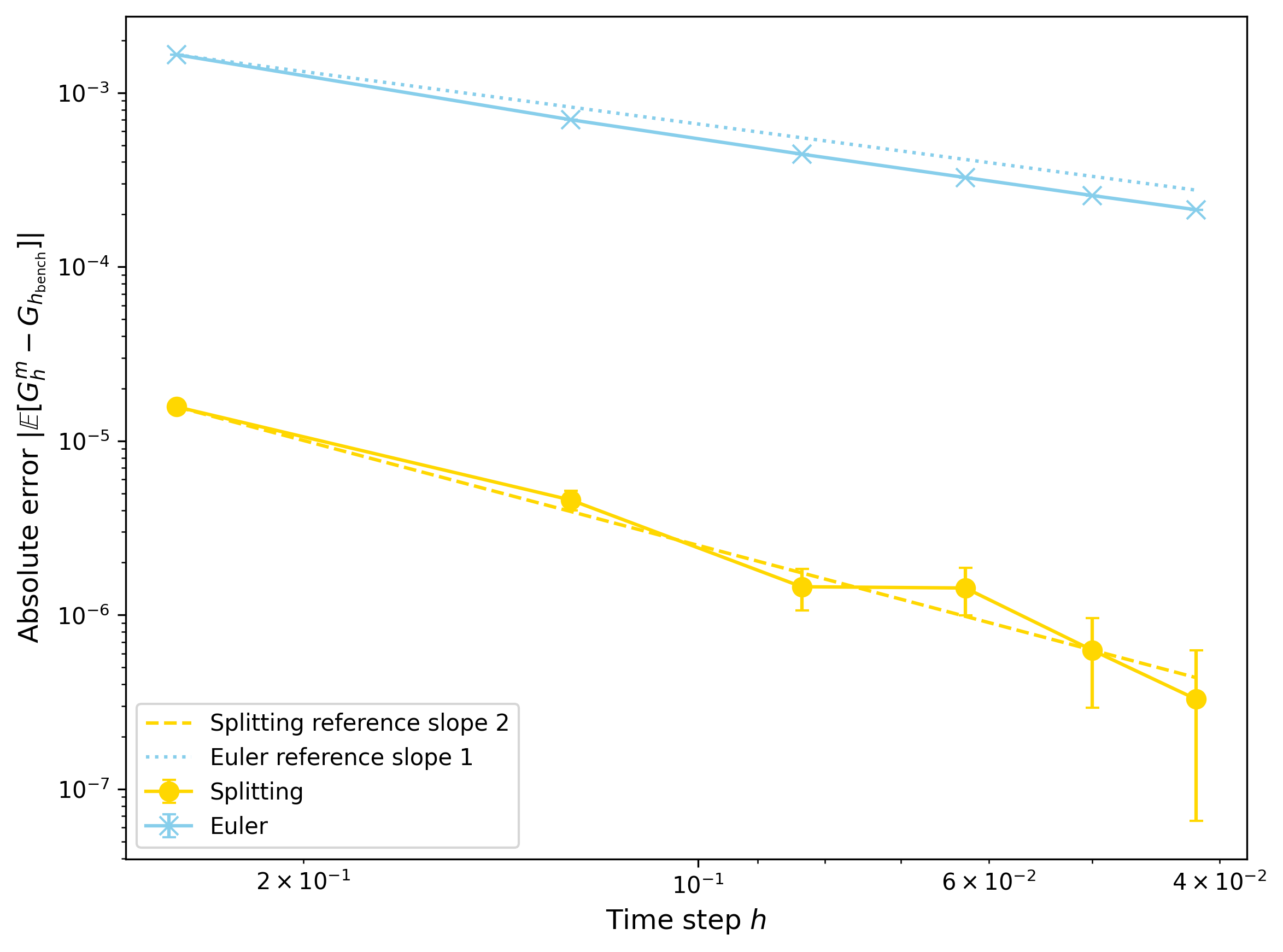}
        \label{fig:convergence_asian_gamma_ou}
    \end{minipage}

    \begin{minipage}[t]{0.4\textwidth}
        \centering
        {\small (c) OU-TS} \vspace{0.5em} \\
        \includegraphics[width=\textwidth]{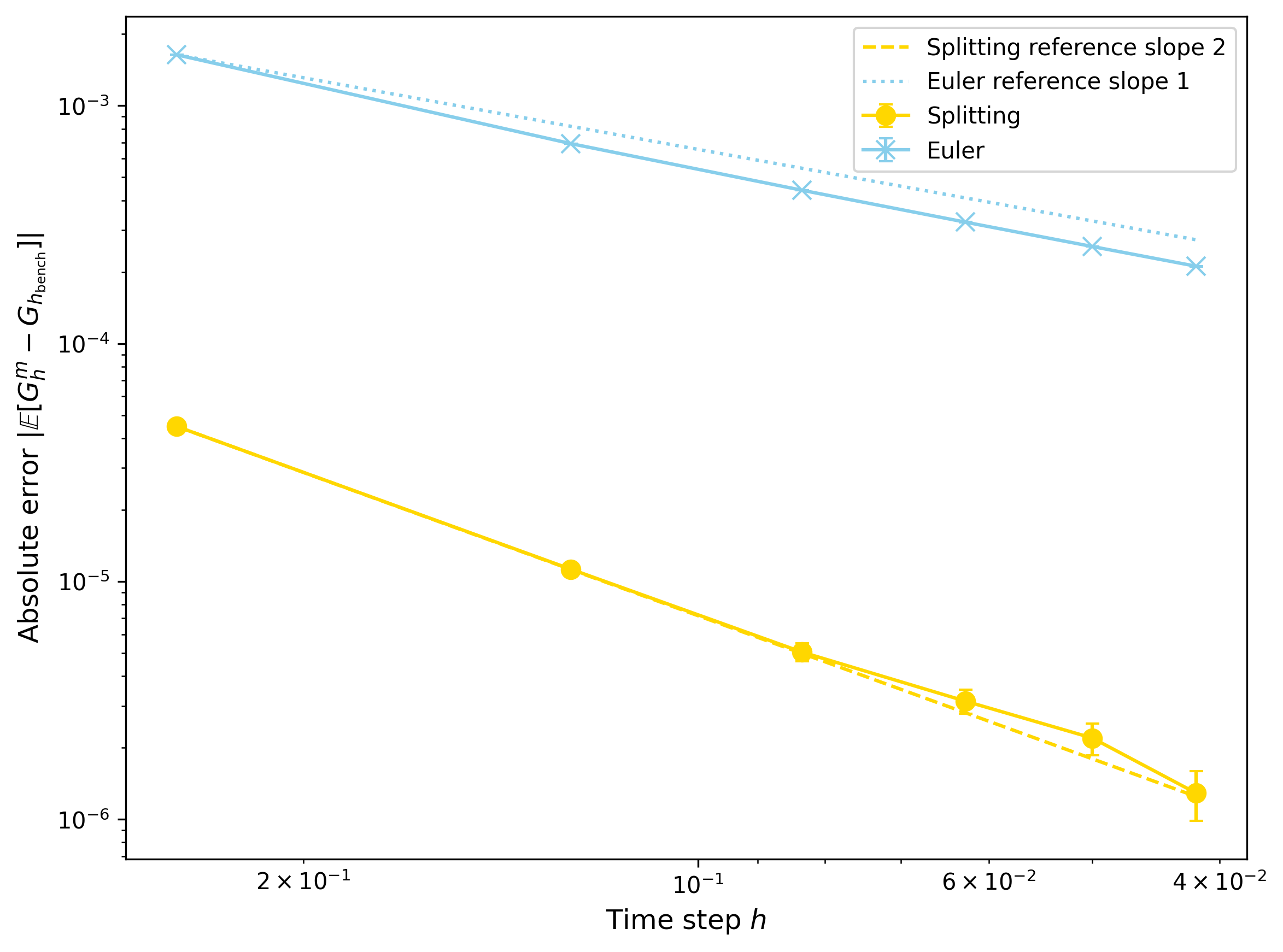}
        \label{fig:convergence_asian_ou_ts}
    \end{minipage}
    \hfill
    \begin{minipage}[t]{0.4\textwidth}
        \centering
        {\small (d) TS-OU} \vspace{0.5em} \\
        \includegraphics[width=\textwidth]{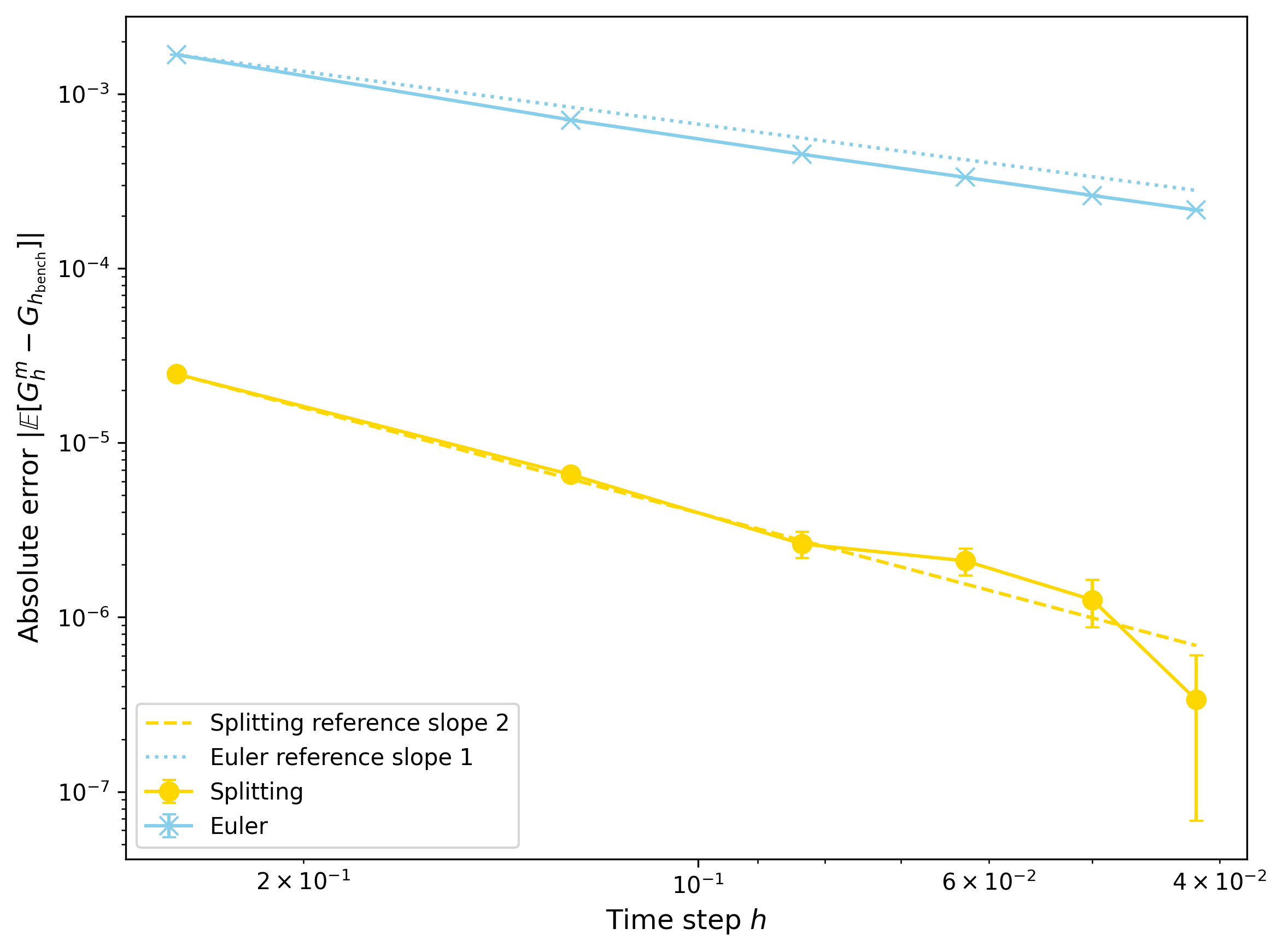}
        \label{fig:convergence_asian_ts_ou}
    \end{minipage}
    
\caption{Test for order of convergence of the operator splitting scheme and Euler scheme for pricing discretely monitored European Asian call options under the BN-S model.
We plot the absolute error
$|\mathbb{E}[G_h-G_{h_{\mathrm{bench}}}]|$ against time step $h$. }
	\label{fig:asian-coupled-convergence}
\end{figure}

\section{Conclusion}\label{s:conclusion}\noindent
There have been continual efforts over the past two decades to develop efficient, accurate, and reliable simulation schemes for pricing options under various stochastic volatility models. The dynamics of a typical stochastic volatility model involve the joint evolution of asset price and variance. The simulation of the asset price normally requires numerical evaluation of integrated variance conditional on terminal variance value as part of the numerical procedure. Earlier simulation schemes employ the path wise Fourier inversion of conditional characteristic function of integrated variance or numerical approximation via moment matched distribution of integrated variance.  

In this paper, we propose a general framework of deriving efficient and accurate simulation schemes for option pricing under most common types of stochastic volatility models. We design ingenious approaches of operator splitting of the coupled system of the dynamic equations of asset price and variance. The operator splitting approach is applicable to the Heston type models, Hull-White model and non-Gaussian Ornstein-Uhlenbeck models (or called Barndorff-Nielsen and Shephard models). Our second-order operator splitting schemes are easily implementable without much programming effort, avoiding sophisticated mathematical framework of deriving various forms of conditional characteristic function or moment matching of integrated variance. Our extensive numerical tests of pricing various types of options reveal the strong competitiveness of the operator splitting schemes in terms of tradeoffs in numerical accuracy versus computational cost when compared with other simulation schemes and the biased Euler scheme. The runtime efficiency can be several hundred times faster.  Our operator splitting approach is particularly successful in pricing path dependent options and VIX options under the Barndorff-Nielsen and Shephard model, where efficient and accurate simulation schemes are not commonly available for various specifications of the Background Driving L\'evy Process. Overall speaking, the operator splitting schemes can be derived systematically without much tedious and sophisticated mathematical derivation. The implementation of the operator splitting schemes typically involves straightforward numerical procedures, avoiding Fourier inversion for each simulation path, moment matching algorithms, or acceptance-rejection sampling.  
We also establish second-order weak convergence of the operator splitting
schemes under the stated regularity assumptions, with corresponding
implications for pricing European and discretely monitored path dependent
options whose payoff functions satisfy these assumptions.

\section*{Acknowledgment}
\noindent The work of Yue Kuen Kwok was supported by Project 2023CX10X1 and the Guangzhou-HKUST (GZ) Joint Funding Program under Grant Number 2024A03J0630.

\section*{Declaration of Interest Statement}
\noindent There is no interest to declare.

\bibliography{ref_split_trap}
    \appendix
    \section{Exact simulation of the CIR process}\label{ap:CIR_simulation}\noindent
    We present a summary of the exact simulation of the CIR process in Proposition~3.1.2 and 3.1.3 of \citet{alfonsi2015affine}. Consider the CIR process $V_t$ as governed by following SDE:
    \begin{equation*}
    	\text{d}V_t=(\alpha-\beta V_t)\,\text{d}t+\sigma\sqrt{V_t}\,\text{d}W_t,
    \end{equation*}
    with initial value $V_0$, $\alpha\geq0$ and $\sigma>0$. We define 
    \begin{equation*}
    	\zeta_\beta(t)=\begin{cases}
    		\frac{1-e^{-\beta t}}{\beta},\,\beta\neq0\\
    		t,\,\beta=0,
    	\end{cases}
    \end{equation*}
    and $c_t=\frac4{\sigma^2\zeta_\beta(t)}.$
    
    \textbf{Case 1: }Suppose $4\alpha\geq\sigma^2$, $V_t$ can be simulated as follows:
    \begin{equation}
    	V_t\sim\left(e^{-\beta t/2}\sqrt{V_0}+\frac{\sigma}{2}\sqrt{\zeta_\beta(t)}G\right)^2+Z,
    \end{equation}
    where $G\sim N(0,1)$, $Z\sim\Gamma(\frac{2\alpha}{\sigma^2}-\frac12,\frac{c_t}{2})$ if $4\alpha>\sigma^2$, and $Z=0$ if $4\alpha=\sigma^2$.
    
    \textbf{Case 2: }Suppose $4\alpha<\sigma^2$, $V_t$ can be simulated as follows:
    \begin{equation}
    	V_t\sim\mathbf{1}_{U\leq e^{-d_t V_0/2}}\tilde{Z}+\mathbf{1}_{U>e^{-d_t V_0/2}}\tilde{V}_t,
	\end{equation}
	where $U\sim uniform[0,1]$, $d_t=e^{-\beta t}c_t$, $\tilde{Z}\sim\Gamma(\frac{2\alpha}{\sigma^2},\frac{c_t}{2})$ if $\alpha>0$ and $\tilde{Z}=0$ if $\alpha=0$.\\
	Here, $\tilde{V}_t$ is a CIR process followed by
	\begin{equation*}
		\text{d}\tilde{V}_t=(\tilde{\alpha}-\beta \tilde{V}_t)\,\text{d}t+\sigma\sqrt{\tilde{V}_t}\,\text{d}W_t,\,\tilde{V}_0=V_0+\frac{2\ln U}{d_t},
	\end{equation*}
	where $\tilde{\alpha}=\alpha+\frac{\sigma^2}{2}$. Now, for $\tilde{V}_t$, we observe $4\tilde{\alpha}>\sigma^2$, so it can be simulated according to Case 1.

	\section{Proof of Proposition~\ref{prop:heston}}\label{ap:proof_heston}\noindent
By the Strang-Marchuk operator splitting procedure as illustrated in Eq.~\eqref{eq:split}, we deduce the solution of the second-order operator splitting scheme as follows:
\begin{equation}\label{eq:second_split}
\left\{
\begin{aligned}
    (\tilde{X}, \tilde{V}) &= \left( X^{(1)}(X_t, V_t, h/2), V^{(1)}(V_t, h/2) \right), \\
    (\hat{X}, \hat{V}) &= \left( X^{(2)}(\tilde{X}, \tilde{V}, h), V^{(2)}(\tilde{V}, h) \right), \\
    (X_{t+h}^{\mathrm{split}}, V_{t+h}^{\mathrm{split}}) &= \left( X^{(1)}(\hat{X}, \hat{V}, h/2), V^{(1)}(\hat{V}, h/2) \right).
\end{aligned} 
\right.
\end{equation}

\noindent We unravel the operator splitting procedure by formalizing the updates of the state variables step-by-step using the exact analytical solutions provided in Eqs.~(\ref{eq:heston_split_sol}a, b).

\vspace{0.5em}
\noindent\textit{Step 1: First half-step of Sub-operator 1 over $h/2$.} \\
By applying Eq.~\eqref{eq:heston_split_first_sol} over a fractional step $h/2$ with initial values $X_t$ and $V_t$, we obtain the intermediate states $\tilde{X}_{t+h/2}^{(1)}$ and $\tilde{V}_{t+h/2}^{(1)}$. 
The variance $\tilde{V}_{t + \frac{h}{2}}^{(1)}$ and the corresponding integrated variance $I_t^{(1), h/2}$ are given by
\[
\tilde{V}_{t+h/2}^{(1)} = e^{-kh/2}(V_t - \theta) + \theta, \quad \text{and} \quad I_t^{(1), h/2} = \theta \frac{h}{2} + \frac{V_t - \theta}{k}(1 - e^{-kh/2}).
\]
The intermediate log-asset price is updated as
\[
\tilde{X}_{t+h/2}^{(1)} = X_t + r\frac{h}{2} - \frac{1}{2}I_t^{(1), h/2} + \sqrt{(1-\rho^2)I_t^{(1), h/2}} \epsilon_1, \quad \epsilon_1 \sim N(0,1).
\]

\vspace{0.5em}
\noindent\textit{Step 2: Full step of Sub-operator 2 over $h$.} \\
Taking $\tilde{X}_{t+h/2}^{(1)}$ and $\tilde{V}_{t+h/2}^{(1)}$ as the inputs, we apply Eq.~\eqref{eq:heston_split_second_sol} over the full time step $h$. The intermediate variance $\hat{V}_{t+h}^{(2)}$ is simulated from the exact CIR distribution:
\[
\hat{V}_{t+h}^{(2)} \sim CIR(0, 0, \sigma, h, \tilde{V}_{t+h/2}^{(1)}).
\]
The log-asset price $\hat{X}_{t+h}^{(2)}$ is obtained by adding the correlation component:
\[
\hat{X}_{t+h}^{(2)} = \tilde{X}_{t+h/2}^{(1)} + \frac{\rho}{\sigma}(\hat{V}_{t+h}^{(2)} - \tilde{V}_{t+h/2}^{(1)}).
\]

\vspace{0.5em}
\noindent\textit{Step 3: Second half-step of Sub-operator 1 over $h/2$.} \\
Finally, applying Eq.~\eqref{eq:heston_split_first_sol} again over a fractional step $h/2$ using $\hat{X}_{t+h}^{(2)}$ and $\hat{V}_{t+h}^{(2)}$ as inputs yields the terminal states $X_{t+h}^{\text{split}}$ and $V_{t+h}^{\text{split}}$. The integrated variance in this step, denoted by $\hat{I}_{t+h/2}^{(1), h/2}$, is given by
\[
\hat{I}_{t+h/2}^{(1), h/2} = \theta \frac{h}{2} + \frac{\hat{V}_{t+h}^{(2)} - \theta}{k}(1 - e^{-kh/2}).
\]
The terminal variance is $V_{t+h}^{\text{split}} = e^{-kh/2}(\hat{V}_{t+h}^{(2)} - \theta) + \theta$. The final log-asset price is found to be 
\[
X_{t+h}^{\text{split}} = \hat{X}_{t+h}^{(2)} + r\frac{h}{2} - \frac{1}{2}\hat{I}_{t+h/2}^{(1), h/2} + \sqrt{(1-\rho^2)\hat{I}_{t+h/2}^{(1), h/2}} \epsilon_2, \quad \epsilon_2 \sim N(0,1).
\]

\vspace{0.5em}
\noindent\textit{Composition to yield Eq.~\eqref{eq:sol_heston_split}} \\
Consequently, the total integrated variance $I_t^{\text{split},h}$ is the sum of the integrated variances from the two fractional steps, where
\[
\begin{aligned}
I_t^{\text{split},h} = I_t^{(1), h/2} + \hat{I}_{t+h/2}^{(1), h/2} = \theta h + \frac{V_t + \hat{V}_{t+h}^{(2)} - 2\theta}{k}(1 - e^{-kh/2}).
\end{aligned}
\]
By substituting $\tilde{X}_{t+h/2}^{(1)}$ into $\hat{X}_{t+h}^{(2)}$, and subsequently into $X_{t+h}^{\text{split}}$, we aggregate the updates of the log-asset price. Since the sum of the independent normal variables $N(0, (1-\rho^2)I_t^{(1), h/2})$ and $N(0, (1-\rho^2)\hat{I}_{t+h/2}^{(1), h/2})$ is equivalent to $N(0, (1-\rho^2)I_t^{\text{split},h})$ in distribution, the final updates for the state variables are given by
\[
\begin{cases}
X_{t+h}^{\text{split}} = X_t + rh - \frac{1}{2}I_t^{\text{split},h} + \frac{\rho}{\sigma}(\hat{V}_{t+h}^{(2)} - \tilde{V}_{t+h/2}^{(1)}) + N\left(0, (1-\rho^2)I_t^{\text{split},h}\right), \\
V_{t+h}^{\text{split}} = e^{-kh/2}(\hat{V}_{t+h}^{(2)} - \theta) + \theta.
\end{cases}
\]

\section{Proof of \Cref{prop:second_order_convergence}}\label{appendix: proof_second_order_vix}
\noindent Let $\mu_Z = \mathbb{E}[Z_1]$. We compare the increments of the exact mean $\mathbb{E}[V_{t+h}]$ and the approximate mean $\mathbb{E}[\bar{V}_{t+h}]$ over single time step $h$. The exact solution and the second-order operator splitting approximation are given by
\begin{align*}
    \mathbb{E}[V_{t+h}] &= e^{-\lambda h} \mathbb{E}[V_t] + \mu_Z \frac{1 - e^{-\lambda h}}{\lambda}, \\
    \mathbb{E}[\bar{V}_{t+h}] &= e^{-\lambda h} \mathbb{E}[\bar{V}_t] + \mu_Z h e^{-\lambda h/2}.
\end{align*}
The local truncation error $e_{loc}(h)$ is the difference between the above two increment terms. Using the Taylor expansions $\frac{1-e^{-\lambda h}}{\lambda} = h - \frac{\lambda h^2}{2} + \frac{\lambda^2 h^3}{6} + {O}(h^4)$ and $h e^{-\lambda h/2} = h - \frac{\lambda h^2}{2} + \frac{\lambda^2 h^3}{8} + {O}(h^4)$, we obtain
\begin{equation*}
    e_{loc}(h) = \mu_Z \left( \frac{1 - e^{-\lambda h}}{\lambda} - h e^{-\lambda h/2} \right) = \frac{\mu_Z \lambda^2}{24} h^3 + {O}(h^4).
\end{equation*}
The global bias at $T = Nh$ is the accumulation of these local errors, observing decay factor $e^{-\lambda h}$ over each time step. We then obtain
\begin{equation*}
    \left| \mathbb{E}[V_T] - \mathbb{E}[\bar{V}_T] \right| = \left| \sum_{k=0}^{N-1} e^{-\lambda (N-1-k)h} e_{loc}(h) \right| = |e_{loc}(h)| \frac{1 - e^{-\lambda T}}{1 - e^{-\lambda h}}.
\end{equation*}
Since the geometric sum factor has the same order as $\frac{1}{1 - e^{-\lambda h}} \approx \frac{1}{\lambda h} = {O}(h^{-1})$, the global bias satisfies the relation:
\begin{equation*}
    \text{Global Bias} \approx {O}(h^3) \cdot {O}(h^{-1}) = {O}(h^2). \qedhere
\end{equation*} 

	\section{Proof of Proposition~\ref{prop:vix_affine_cf}}\label{appendix:proof_vix_affine_cf}\noindent
	The proof consists of two parts: establishing the affine structure of the $\text{VIX}_T^2$ index and deriving the corresponding characteristic function.

	\vspace{1em}
\noindent \textit{Part 1: Affine structure of $\text{VIX}_T^2$} 

\noindent The squared VIX index is defined via the risk-neutral expectation of the integrated variance plus a jump correction term:
\begin{equation*}
    \text{VIX}_T^2 = \frac{1}{\tau_{vix}} \mathbb{E}^{\mathbb{Q}} \left[ \int_{T}^{T + \tau_{vix}} V_s \, \mathrm{d}s \bigg| \mathcal{F}_{T} \right] + 2 \int_{0}^{\infty} (e^{\rho x} - 1 - \rho x) \, \nu(\mathrm{d}x).
\end{equation*}
Using the integrated form of the variance process $V_s = V_T e^{-\lambda (s-T)} + \int_T^s e^{-\lambda (s-u)} \, \mathrm{d}Z_u$ and Fubini's theorem, the expected integrated variance is given by:
\begin{equation*}
    \mathbb{E}^{\mathbb{Q}} \left[ \int_{T}^{T + \tau_{vix}} V_s \, \mathrm{d}s \bigg| \mathcal{F}_{T} \right] = V_T \frac{1 - e^{-\lambda \tau_{vix}}}{\lambda} + \left( \int_0^\infty x \nu(\mathrm{d}x) \right) \frac{\lambda \tau_{vix} - (1 - e^{-\lambda \tau_{vix}})}{\lambda^2}.
\end{equation*}
Grouping the terms involving the Lévy measure $\nu(\mathrm{d}x)$, we obtain the affine relation between $\text{VIX}_T^2$ and $V_T$:
\begin{equation}
    \text{VIX}_T^2 = \mathcal{A}(\tau_{vix}) V_T + \mathcal{B}(\tau_{vix}),
    \label{eq:affine_vix}
\end{equation}
where $\mathcal{A}(\tau_{vix}) = \frac{1 - e^{-\lambda \tau_{vix}}}{\lambda \tau_{vix}}$ and $\mathcal{B}(\tau_{vix})$ collects the sum of constant drift and jump adjustments. 

\vspace{1em}
\noindent \textit{Part 2: Characteristic function} 

\noindent Recall the conditional characteristic function: $\Phi(\phi; V_t, \tau) = \mathbb{E}^{\mathbb{Q}}[e^{i\phi \text{VIX}_T^2} | V_t]$, where $\tau = T-t$. Based on the affine structure in Eq.~\eqref{eq:affine_vix}, we conjecture the exponential affine form:
    \begin{equation*}
        \Phi(\phi; V_t, \tau) = \exp\left( A(\phi; \tau)V_t + C(\phi; \tau) \right).
    \end{equation*}
    According to the Feynman-Kac Representation Theorem, $\Phi$ satisfies the following Partial Integro-Differential Equation (PIDE):
    \begin{equation*}
        \frac{\partial \Phi}{\partial \tau} = -\lambda V_t \frac{\partial \Phi}{\partial V_t} + \int_0^{\infty} (\Phi(V_t + y, \tau) - \Phi(V_t, \tau)) \, \nu(\mathrm{d}y),
    \end{equation*}
    subject to the terminal condition $\Phi(\phi; V_T, 0) = \exp(i\phi (\mathcal{A}(\tau_{vix}) V_T + \mathcal{B}(\tau_{vix})))$. Substituting the ansatz into the PIDE and separating terms involving $V_t$ yields the following system of ODEs:
    \begin{equation*}
    \begin{cases}
        \frac{\partial A}{\partial \tau} = -\lambda A(\phi; \tau), & A(\phi; 0) = i\phi \mathcal{A}(\tau_{vix}) \\
        \frac{\partial C}{\partial \tau} = \int_0^{\infty} (e^{A(\phi; \tau)y} - 1) \, \nu(\mathrm{d}y), & C(\phi; 0) = i\phi \mathcal{B}(\tau_{vix}).
    \end{cases}
    \end{equation*}
    Solving these ODEs explicitly gives the solutions of the coefficients functions:
    \begin{equation}
        A(\phi; \tau) = i\phi \mathcal{A}(\tau_{vix}) e^{-\lambda \tau}, \quad C(\phi; \tau) = i\phi \mathcal{B}(\tau_{vix}) + \int_0^{\tau} \int_0^{\infty} (e^{A(\phi, s)y} - 1) \, \nu(\mathrm{d}y) \, \mathrm{d}s.
    \end{equation}

\end{document}